\documentclass[12pt,letterpaper]{article}
\usepackage[top=1in, left=1in, right=1in, bottom=1in]{geometry}	
\usepackage{amsmath, amssymb, amsthm, bm}
\usepackage{natbib}
\setcitestyle{authoryear,aysep={}}
\usepackage[lined, ruled, commentsnumbered]{algorithm2e}
\usepackage{booktabs}
\usepackage{multirow}
\usepackage{centernot}
\usepackage{cancel}
\usepackage{bm}
\usepackage{tkz-graph}  
\usepackage{verbatim}
\usepackage[final]{pdfpages}
\usepackage{xcolor,colortbl} 
\usepackage{tcolorbox}
\usepackage{caption}
\usepackage{subfigure}
\usepackage{hyperref}
\usepackage{enumerate}
\usepackage{setspace}
\usepackage{graphicx}
\usepackage{subcaption}
\usepackage{bbm}
\usepackage{times}
\usepackage{arydshln}
\SetAlgorithmName{Procedure}{procedure}{List of Procedures}
\newcommand{\com}[2]{
  \left(\begin{array}{cc}
    #1 \\ #2
  \end{array}\right)
}

\newcommand{\lb}{\left(}
  \newcommand{\rb}{\right)}
\newcommand{\td}{\tilde}

\newcommand{\indep}{\ \raisebox{0.05em}{\rotatebox[origin=c]{90}{$\models$}} \ }

\newcommand{\E}{\operatorname{\mathbb{E}}}

\newcommand{\T}{\mathcal{T}}

\renewcommand{\u}{\mathrm{u}}

\newtheorem{assumption}{Assumption}[section]

\newtheorem{theorem}{Theorem}[section]
\newtheorem{lemma}{Lemma}[section]
\newtheorem{proposition}{Proposition}[section]

\newtheorem{remark}{Remark}[section]

\allowdisplaybreaks

\title{Compound Estimation for Binomials
}
\author{Yan Chen\thanks{Fuqua School of Business, Duke University; Email: \href{mailto:yc555@duke.edu}{yc555@duke.edu}}\quad\quad Lihua Lei\thanks{Stanford Graduate School of Business and Department of Statistics (by courtesy); Email: \href{mailto:lihualei@stanford.edu}{lihualei@stanford.edu}}}

\date{ }

\begin{document}

\maketitle

\onehalfspacing
\begin{abstract}
Many applications involve estimating the mean of multiple binomial outcomes as a common problem -- assessing intergenerational mobility of census tracts, estimating prevalence of infectious diseases across countries, and measuring click-through rates for different  demographic groups. The most standard approach is to report the plain average of each outcome. Despite simplicity, the estimates are noisy when the sample sizes or mean parameters are small. In contrast, the Empirical Bayes (EB) methods are able to boost the average accuracy by borrowing information across tasks. Nevertheless, 
the EB methods require a Bayesian model where the parameters are sampled from a prior distribution which, unlike the commonly-studied Gaussian case, is unidentified due to discreteness of binomial measurements. Even if the prior distribution is known, the computation is difficult when the sample sizes are heterogeneous as there is no simple joint conjugate 
prior for the sample size and mean parameter. 

In this paper, we consider the compound decision framework which treats the sample size and mean parameters as fixed quantities. We develop an approximate Stein's Unbiased Risk Estimator (SURE) for the average mean squared error given any class of estimators. For a class of machine learning-assisted linear shrinkage estimators, we establish asymptotic optimality, regret bounds, and valid inference. Unlike existing work, we work with the binomials directly without resorting to Gaussian approximations. This allows us to  work with small sample sizes and/or mean parameters in both one-sample and two-sample settings. We demonstrate our approach using three datasets on firm discrimination, education outcomes, and innovation rates.

\vspace{1em}
\end{abstract}

\newpage
\onehalfspacing

\section{Introduction}
Estimating mean of multiple binomial outcomes is a common task in applied economics, public policy, and experimental analysis. In settings ranging from labor market discrimination and education interventions to development programs and health policy, researchers often observe binomial outcomes—such as employment rate, school attendance, or treatment uptake—and aim to estimate the mean parameters (one-sample setting) or treatment effects (two-sample setting). For instance, \cite{bell2019becomes} link the tax records with the patent records and report the innovation rates across hundreds of American colleges. These rates can be viewed as the estimated mean parameters from multiple binomials, each corresponding to the number of inventors in a college. As an example of the two-sample setting, \citet{kline2024discrimination} analyze a large-scale correspondence experiment, sending up to 1,000 fictitious job applications with randomly assigned race and gender indicators to 108 firms. Differences in callback rates across race–gender groups yield firm-level binomial outcomes indicative of discriminatory practices, analyzed via Empirical Bayes estimation to produce adjusted discrimination measures. 


Existing methods for estimating multiple binomial means can be broadly classified into three categories. The most common approach is to treat multiple binomials independently. This includes the simple estimator via the empirical average and prediction-based estimator when covariates are available. However, these methods do not allow for information sharing across different binomial units, limiting their efficiency when estimating a large number of binomial outcomes. The second category treats multiple binomial outcomes within a Bayesian framework, using either fully Bayesian or Empirical Bayes (EB) methods. When sample sizes $n_i$ are equal, both Bayes and EB can be applied with the beta priors on the mean parameters \citep[][etc.]{griffin1971empirical}. However, when $n_i$ varies across observations \citep[e.g.][]{fienberg1973simultaneous}, inference requires specifying a joint prior over $(n, p)$ (where $p$ is the binomial parameter), which complicates computation due to discreteness and the lack of a conjugate prior for $n$, making the effectiveness of Bayes or EB methods less clear in this setting. The third category 
treat the binomials as approximate Gaussian variables with variances imputed as the large-sample approximation
\citep[e.g.][]{Brown2008inseason,xie2012sure,chen2022empirical,chen2025compound}, and develop theoretical guarantees assuming the measurements are exactly Gaussian. Nevertheless, the Gaussian approximation tends to be inaccurate when the sample sizes or the mean parameters are small. In addition, under the Bayesian or EB framework where the size and mean parameters are sampled from a distribution, the prior is only partially identified due to discreteness of binomial measurements \citep{kline2021reasonable}. By contrast, the prior is point identified with Gaussian measurements under fairly mild regularity conditions \citep[e.g.][]{efron2012large,chen2022empirical}. This fundamental difference poses another threat to the Gaussian-approximation-based EB methods.


\subsection{Contributions}
Our paper directly addresses the limitation of the Gaussian approximation by directly accommodating binomial outcomes through their exact sampling distribution. Unlike most previous work that develop EB methods, we formulate the problem as a compound estimation problem \citep{robbins1951asymptotic}, where we treat all size and mean parameters as fixed, rather than random variables generated from a prior distribution. In particular, we allow arbitrarily heterogeneous and bounded sample sizes across observations. We show that, in terms of average mean squared error, our proposed estimator outperforms the maximum likelihood estimator and any single machine-learning estimator under mild regularity conditions. When covariate information is available, the estimator naturally incorporates arbitrary machine learning predictors. In addition, the proposed estimator satisfies a reporting consistency property: its weighted average coincides exactly with the naïve weighted average, a property not shared by existing methods. By contrast, many alternative shrinkage estimators fail to satisfy this property, leading to accounting inconsistency that could be misleading in high-stake settings.

Our methodology proceeds as follows. First, we derive a Stein’s unbiased risk estimator (SURE) for squared-error loss specifically tailored to binomial parameters. Using this, we propose a family of shrinkage estimators that combine the maximum likelihood estimator (MLE), the grand mean, and predictions from machine learning models. Our approach is closely related to the shrinkage estimator introduced by \citet{xie2012sure}, which shrinks estimates toward either the grand mean or a data-driven location under a Gaussian model assumption with known variances. In contrast, our shrinkage estimator minimizes the SURE associated with average mean squared error under a binomial model assumption. We derive the regret bound, establish the asymptotic normality of our proposed estimator, develop a corresponding statistical inference procedure, and provide practical guidance for constructing confidence regions.

\subsection{Related Work}

\paragraph{Stein's Identity for Binomial Distributions} Our estimator is inspired by Stein's identity for binomial distributions. Stein's identity was originally introduced by \citet{stein1972bound} for approximating the distribution of sums of dependent random variables by a normal distribution. Later, \citet{barbour2001compound} extended Stein's identity to binomial and related discrete distributions. Other related developments include work by \citet{ehm1991binomial} and \citet{soon1996binomial}, etc. 

\paragraph{Stein's Unbiased Risk Estimator}
Additionally, our estimator is based on the \textit{Stein's unbiased risk estimator} (SURE), which was first proposed by \citet{stein1981estimation}, as an unbiased estimator of the mean-squared error for estimating the mean of a multivariate normal distribution. Since then, there is a considerable amount of literature that studies the minimization of a SURE-type risk estimate via relatively simple estimators (e.g. linear smoothers) \citep[][etc.]{li1985stein,li1986asymptotic,li1987asymptotic,johnstone1988inadmissibility,kneip1994ordered,donoho1995adapting}. More recently, there has been a line of work applying SURE for tuning regularization parameters for high-dimensional methods such as the Lasso, reduced-rank regression, and singular value thresholding \citep[e.g.,][]{tibshirani2012degrees,candes2013unbiased,mukherjee2015degrees}. Further, \citet{xie2012sure} proposed a class of SURE-based shrinkage estimators and showed a uniform consistency property for SURE in a hierarchical model. 

Notably, both the earlier work and recent studies \citep[e.g.,][]{xie2012sure,ghosh2025stein,nobel2023tractable,karamikabir2021wavelet,kim2021noise2score, chen2025compound} rely on the Gaussian distribution assumption. While \citet{eldar2008generalized} derived the SURE for mean squared error within general exponential families, their primary goal was to select regularization parameters for rank-deficient Gaussian models and linear Gaussian models.

\paragraph{Bayes and Empirical Bayes Methods} A substantial body of literature also addresses the estimation of binomial outcomes through Bayes and empirical Bayes methods. \citet{griffin1971empirical} derived a Bayes estimator expressed in terms of the marginal probabilities rather than the prior, thereby obtaining an empirical Bayes estimator. \citet{berry1979empirical} applied the theory of Dirichlet processes to the empirical Bayes estimation for the binomial outcomes. \citet{albert1984empirical} proposed a empirical Bayes method by defining a class of prior distributions for a set of binomial probabilities to reflect the user's prior belief about the similarity of the probabilities. \citet{sobel1993bayes} constructed ranking procedures for comparing multiple binomial parameters via both Bayes and empirical Bayes approaches. \citet{sivaganesan1993robust} proposed an empirical Bayes approach that partially identifies the posterior means of binomial outcomes by imposing moment conditions on the unknown prior. \citet{consonni1995bayesian} adapted a Bayesian approach to estimate the binomial parameters by imposing prior information about the partitions of the binomial experiments. \citet{weiss2010bayesian} used a Bayesian hierarchical approach to simultaneously estimate the parameters of multiple binomial distributions. \citet{kline2021reasonable} employed an empirical Bayes approach to estimate binomial parameters while treating sample sizes as fixed, and develop partial-identification methods for moments in the two-sample setting, with an application to job-level discrimination detection. \citet{gu2025reasonable} further advanced this line of work by incorporating both partial-identification and sampling uncertainty to construct valid confidence intervals for empirical Bayes estimators. 

\subsection{Basic Notations}
For any positive integer $n$, we use $[n]$ to denote $\{1,2,\ldots,n\}$. Given a random variale $A$ and a distribution $\mathcal{F}$, we write $A\sim\mathcal{F}$ to imply that $A$ follows distribution $\mathcal{F}$. For any two random variables $A$ and $B$, $A\indep B$ means $A$ is independent of $B$. We use $\mathbf{1}(\cdot)$ to denote the indicator function. For any matrix or vector $\boldsymbol{\nu}$, we use $\boldsymbol{\nu}^T$ to denote the transpose of $\boldsymbol{\nu}$. Given any matrix $\mathbf{M}$, we use $\|\mathbf{M}\|_{\infty,\infty}$ to denote the maximum absolute value for the entries of $\mathbf{M}$. For any vector $\boldsymbol{\nu}=(\nu_1,\ldots,\nu_k)\in\mathbb{R}^d$, $\|\boldsymbol{\nu}\|_{\infty}=\max_{k\in[d]}|\nu_k|$ and $\|\boldsymbol{\nu}\|_{2}=\sqrt{\sum_{k=1}^d\nu_k^2}$. Given any two matrix $\mathbf{M}_1$ and $\mathbf{M}_2$, we write $\mathbf{M}_1\preceq\mathbf{M}_2$ to mean that $\mathbf{M}_2-\mathbf{M}_1$ is positive semi-definite. For any sequence of random variables $X_n$, we say that $X_n=\mathrm{o}_p(1)$ if $X_n$ converges in probability to $0$, and we say that $X_n=\mathrm{O}_p(R_n)$ for some real sequence $R_n$, if $X_n=Y_nR_n$, and $\{Y_n\}$ is uniformly tight. We use $\rightsquigarrow$ to imply ``converges in distribution to'' and $\overset{p}{\rightarrow}$ to imply ``converges in probability to''. For any function $f(\cdot)$, we use $\nabla f(\cdot)$ to denote the gradient of $f$. For any $d\in\mathbb{Z}_{+}$, we use $\mathbf{I}_d$ to denote the $d$-by-$d$ identity matrix. 

\section{Stein's Unbiased Risk Estimators for Binomials}\label{sec:SURE}
\subsection{Setup} 
We aim to estimate the unknown binomial parameters $\{\theta_i^{\mathrm{o}}\}_{i\in[N]}$ in the one-sample setting and $\{\theta_{i1}^{\mathrm{t}},\theta_{i2}^{\mathrm{t}}\}_{i\in[N]}$ in the two-sample setting. 

In the one-sample setting, we observe $\{n_i,\mathbf{X}_i,Y_i\}_{i\in[N]}$, where $n_i\in\mathbb{Z}_{+}$, $\mathbf{X}_i\in\mathbb{R}^d$ are fixed for any $i\in[N]$, $\{Y_i\}_{i\in[N]}$ are \textit{independently but not necessarily identically distributed} (i.n.i.d.) random variables such that $Y_{i}\sim \mathrm{Bin}(n_i, \theta_i^{\textrm{o}})$, and
\begin{equation}\label{eq:one-sample:dependence:covariate}
    \theta_i^{\textrm{o}} = g(\mathbf{X}_i)+\eta_i,\quad \mathbb{E}[\eta_i|\mathbf{X}_i]=0,\quad \forall i\in[N].
\end{equation}
Let $\hat{\theta}_i^{\textrm{o}}$ denote the estimator for $\theta_i^{\textrm{o}}$, and $\hat{\theta}_i^{\textrm{o}}$ could depend on all observations. 

In the two-sample setting, we observe $\{n_{i1},\mathbf{X}_{i1},Y_{i1}\}_{i\in[N]}$ and $\{n_{i2},\mathbf{X}_{i2},Y_{i2}\}_{i\in[N]}$ from group one and two, where $n_{i\ell}\in\mathbb{Z}_{+}$ and $\mathbf{X}_{i\ell}\in\mathbb{R}^d$ are fixed, $\forall i\in[N], \ell\in\{1,2\}$. The group one binomial outcomes $\{Y_{i1}\}_{i\in[N]}$ are independent from $\{Y_{i2}\}_{i\in[N]}$ of group two. Here $Y_{i1}\sim \mathrm{Bin}(n_{i1}, \theta_{i1}^{\textrm{t}})$ and $Y_{i2}\sim \mathrm{Bin}(n_{i2},\theta_{i2}^{\textrm{t}})$, where 
\begin{equation}\label{eq:depend:covariate:two-sample}
    \theta_{i\ell}^{\mathrm{t}}=g_{\ell}(\mathbf{X}_{i\ell})+\eta_{i\ell},\quad \mathbb{E}[\eta_{i\ell}|\mathbf{X}_{i\ell}]=0,\quad \forall i\in[N], \ \ell\in\{1,2\}.
\end{equation}
Suppose $\{Y_{i1}\}_{i\in[N]}$ and $\{Y_{i2}\}_{i\in[N]}$ are both i.n.i.d. across $i\in[N]$. We estimate $\theta_{i1}^{\textrm{t}}-\theta_{i2}^{\textrm{t}}$ using the two-sample estimator $\hat{\theta}_i^{\textrm{t}}$ which could also depend on all observations. 

Denote the vectors of the unknown estimands as $\theta^{\mathrm{o}} := (\theta_1^{\mathrm{o}}, \ldots, \theta_N^{\mathrm{o}})$, $\theta_{\ell}^{\mathrm{t}}:=(\theta_{1\ell}^{\mathrm{t}},\ldots,\theta_{N\ell}^{\mathrm{t}})$, $\ell\in\{1,2\}$. Denote the vectors of the one-sample and two-sample estimators as $\hat{\theta}^{\mathrm{o}} := (\hat{\theta}_1^{\mathrm{o}}, \ldots, \hat{\theta}_N^{\mathrm{o}})$ and $\hat{\theta}^{\mathrm{t}}:=(\hat{\theta}_1^{\mathrm{t}},\ldots,\hat{\theta}_N^{\mathrm{t}})$. The objective is to assess the performance of the estimators in both one-sample and two-sample settings through their $L_2$ risks, specified in \eqref{eq:L2:one-sample} and \eqref{eq:L2:two-sample}, respectively:
\begin{equation}\label{eq:L2:one-sample}
   L_{2}^{\mathrm{o}}\left(\hat{\theta}^{\mathrm{o}}; \theta^{\mathrm{o}}\right) = \frac{1}{N}\sum_{i=1}^{N} \E\left[\left(\hat{\theta}_i^{\mathrm{o}} - \theta_i^{\mathrm{o}}\right)^2\right],
\end{equation}
\begin{equation}\label{eq:L2:two-sample}
    L_2^{\mathrm{t}}\left(\hat{\theta}^{\mathrm{t}};\theta_1^{\mathrm{t}},\theta_2^{\mathrm{t}}\right)=\frac{1}{N}\sum_{i=1}^{N}\mathbb{E}\left[\left\{\hat{\theta}_i^{\mathrm{t}}-\left(\theta_{i1}^{\mathrm{t}}-\theta_{i2}^{\mathrm{t}}\right)\right\}^2\right].
\end{equation}

\noindent Since $\sum_{i=1}^N \left(\theta_i^{\mathrm{o}}\right)^2$ and $\sum_{i=1}^N \left(\theta_{i1}^{\mathrm{t}} - \theta_{i2}^{\mathrm{t}}\right)^2$ are constants, minimizing the one-sample $L_2$ risk is equivalent to minimizing $L_{\mathrm{o}}$ as follows:
\begin{equation}\label{eq:unweighted:one-sample}
    L_{\mathrm{o}}\left(\hat{\theta}^{\mathrm{o}}; \theta^{\mathrm{o}}\right) = \frac{1}{N}\sum_{i=1}^{N}\E\left[\left(\hat{\theta}_i^{\mathrm{o}}\right)^2\right] - 2\theta_i^{\mathrm{o}} \E\left[\hat{\theta}_i^{\mathrm{o}}\right],
\end{equation}
and minimizing the two-sample $L_2$ risk is equivalent to minimizing $L_{\mathrm{t}}$ defined as follows:
\begin{equation}\label{eq:unweighted:L2:risk}
L_{\mathrm{t}}\left(\hat{\theta}^{\mathrm{t}};\theta_1^{\mathrm{t}}, \theta_2^{\mathrm{t}}\right)=\frac{1}{N}\sum_{i=1}^N\mathbb{E}\left[\left(\hat{\theta}_i^{\mathrm{t}}\right)^2\right]-2\left(\theta_{i1}^{\mathrm{t}}-\theta_{i2}^{\mathrm{t}}\right)\mathbb{E}\left[\hat{\theta}_i^{\mathrm{t}}\right].
\end{equation}
The first terms in \eqref{eq:unweighted:one-sample} and \eqref{eq:unweighted:L2:risk} can be unbiasedly estimated by the plug-in estimator $(1/N) \sum_{i=1}^N \left(\hat{\theta}_i^{\mathrm{o}}\right)^2$ and $(1/N) \sum_{i=1}^N \left(\hat{\theta}_i^{\mathrm{t}}\right)^2$, respectively. The main challenge stems from the second terms because of the unknown parameters $\theta_i^\mathrm{o}$ and $(\theta_{i1}^\mathrm{t} - \theta_{i2}^\mathrm{t})$.
In the following subsections, we derive a Stein identity for the binomial distribution and use it to construct Stein’s unbiased risk estimators (SUREs) for the second terms in the one-sample and two-sample objectives, respectively. 


\subsection{Stein's Identity for a Binomial Distribution}
We start with a simpler problem of estimating a single Binomial parameter $\theta$ for $Y \sim \mathrm{Bin}(n,\theta)$, without incorporating covariates, by applying Stein’s identity for Binomial distributions.
\begin{proposition}\label{prop:stein}
  Let $Y\sim \mathrm{Bin}(n, \theta)$. For any function $g$ on $\{0, \ldots, n\}$, 
  \[(1 - \theta)\E[Yg(Y)] = \theta \E[(n - Y)g(Y + 1)].\]
\end{proposition}

In what follows, we show by Theorem~\ref{thm:better} that a SURE does exist for the binomial case when the estimator is constructed via polynomial functions. For any function $h$ on $\{0, \ldots, n\}$, define 
  \begin{equation}
    \label{eq:T1}
    \T_1 h(y; n) := \mathbf{1}\{y > 0\}\sum_{j=0}^{n-y}h(y + j)(-1)^{j}\frac{(n-y)!}{(n-y-j)!}\frac{y!}{(y+j)!},
  \end{equation}
  \begin{equation}
    \label{eq:T2}
    \T_2 h(y; n) := h(y) - \mathbf{1}\{y < n\}\sum_{j=0}^{y}h(y - j)(-1)^{j}\frac{y!}{(y-j)!}\frac{(n - y)!}{(n-y+j)!}
  \end{equation}
  and
  \begin{equation}
      \label{eq:Deltah}
      \Delta h := \sum_{j=0}^{n}h(j)(-1)^{j}\com{n}{j}.
    \end{equation}
For any $a\in \{0, 1, \ldots, n\}$, define 
\begin{equation}
  \label{eq:T}
  \T h(y; n, a) := \T_1 h(y; n) \cdot \mathbf{1}\{y > a\} + \T_2 h(y; n) \cdot \mathbf{1}\{y \le a\}.
\end{equation}
\begin{theorem}\label{thm:better}
Let $Y\sim \mathrm{Bin}(n, \theta)$. For any function $h$ on $\{0, \ldots, n\}$,
  \begin{equation}\label{eq:bias:robust}
      \theta \E[h(Y)] - \E[\T h(Y; n, a)] = (-1)^{a+1}\theta^{a+1} (1 - \theta)^{n-a} (\Delta h).
  \end{equation} 
In particular,
\begin{itemize}
    \item[(i)] When $a=\lfloor n/2\rfloor$ in \eqref{eq:bias:robust}, $\left|\theta \E[h(Y)] - \E[\T h(Y; n, a)]\right|\leq2^{-n}|\Delta h|$.
    \item[(ii)] $\Delta h=0$ if $h$ is a polynomial of degree less than $n$.
\end{itemize}
\end{theorem}
Theorem \ref{thm:better} implies that if $h(\cdot)$ is a polynomial of degree less than $n$, then $\mathcal{T}h(Y;n,a)$ is an unbiased estimator of $\theta\mathbb{E}[h(Y)]$. Throughout the remainder of the paper, we set $a=\lfloor n/2\rfloor$, which is a robust choice according to statement (i) of Theorem~\ref{thm:better}. 

\begin{remark}\label{rmk:stein's identity}
Recall that the Stein's identity for a given Gaussian random variable $Y\sim\mathcal{N}(\theta,\sigma^2)$ is as follows: For any differentiable function $F$ with derivative $F'$, we have 
\begin{equation}\label{eq:stein:Gaussian}
\mathbb{E}_{Y\sim\mathcal{N}(\theta,\sigma^2)}[(Y-\theta)F(Y)]=\sigma^2\mathbb{E}_{Y\sim\mathcal{N}(\theta,\sigma^2)}[F'(Y)].
\end{equation}
Specifically, when the variance $\sigma^2$ is known, if we use estimator $\hat{\theta}=h(Y)$ to estimate $\theta$, where $F$ is a known differentiable function, then rearranging the terms of \eqref{eq:stein:Gaussian} we get 
\begin{equation}\label{eq:unbiased:gaussian}
\theta\mathbb{E}[\hat{\theta}]=\mathbb{E}[Y\hat{\theta}-\sigma^2h'(Y)],
\end{equation}
which immediately implies that $Y\hat{\theta}-\sigma^2h'(Y)$ is an unbiased estimator of $\theta\mathbb{E}[\hat{\theta}]$. Since minimizing the $L_2$ risk of $\hat{\theta}$ is equivalent to estimate $\mathbb{E}[\hat{\theta}^2]-2\theta\mathbb{E}[\hat{\theta}]$, the Stein's identity for the Gaussian random distribution gives a straightforward Stein's unbiased risk estimator (SURE) for the $L_2$ risk. 


Unlike in the Gaussian case, Proposition~\ref{prop:stein} shows that Stein’s identity for the binomial distribution does not directly produce a straightforward SURE expression for terms like $\theta \E[h(Y)]$. Instead, we need to find a function $g$ such that $\E[Yg(Y)+ (n-Y)g(Y+1)]$ matches $h$, which is only possible when $h$ is a polynomial of degree less than $n$. In particular, unbiased estimation requires $n \ge 2$ and binary measurements with $n = 1$ do not work.
\end{remark}

\subsection{A Class of Estimators}
We propose a class of estimators for the one-sample estimands \eqref{eq:one-sample:dependence:covariate}, where for any $i\in[N]$,
\begin{equation}\label{eq:two-sample:estimator:lambda:one-sample}
  \hat{\theta}_i^{\mathrm{o}}(\boldsymbol{\lambda}) := \lambda_1\frac{Y_i}{n_i} + (1-\lambda_1)\frac{\sum_{i=1}^{N}Y_i}{\sum_{i=1}^{N}n_i}+\lambda_2\left(\hat{g}(\mathbf{X}_i)-\frac{\sum_{j=1}^Nn_j\hat{g}(\mathbf{X}_{j})}{\sum_{j=1}^Nn_j}\right), \boldsymbol{\lambda}=(\lambda_1,\lambda_2)\in[0,1]\times\mathbb{R},
\end{equation}
and $\hat{g}(\mathbf{X}_i)$ is the machine learning (ML) estimator for $g(\mathbf{X}_i)=\mathbb{E}[\theta_i^{\mathrm{o}}|\mathbf{X}_i]$ defined as \eqref{eq:one-sample:dependence:covariate}. Similarly, for two-sample case, for any $i\in[N]$, $\ell\in\{1,2\}$ and $\boldsymbol{\lambda}=(\lambda_1,\lambda_2)\in[0,1]\times\mathbb{R}$, define
\begin{equation}\label{eq:ell:two-sample:estimator}
    \hat{\theta}_{i\ell}^{\mathrm{o}}(\boldsymbol{\lambda}):=\lambda_1\frac{Y_{i\ell}}{n_{i\ell}} + (1-\lambda_1)\frac{\sum_{i=1}^{N}Y_{i\ell}}{\sum_{i=1}^{N}n_{i\ell}}+\lambda_2\left(\hat{g}_{\ell}(\mathbf{X}_{i\ell})-\frac{\sum_{j=1}^Nn_{j\ell}\hat{g}_{\ell}(\mathbf{X}_{j\ell})}{\sum_{j=1}^Nn_{j\ell}}\right), 
\end{equation}
where for any $\ell\in\{1,2\}$, $\hat{g}_{\ell}(\mathbf{X}_{i1})$ is the ML estimator for $g_{\ell}(\mathbf{X}_{i1})=\mathbb{E}[\theta_{i\ell}^{\mathrm{t}}|\mathbf{X}_{i1}]$ defined as \eqref{eq:depend:covariate:two-sample}. We propose a class of estimators for the two-sample estimands \eqref{eq:depend:covariate:two-sample} as $\hat{\theta}_i^{\mathrm{t}}(\boldsymbol\lambda) := \hat{\theta}_{i1}^{\mathrm{o}}(\boldsymbol{\lambda}) - \hat{\theta}_{i2}^{\mathrm{o}}(\boldsymbol{\lambda})$. 

It is straightforward to see that both the one-sample and two-sample estimators satisfy a reporting consistency property: the aggregate of the unit-level estimates exactly equals the overall empirical proportion.
\begin{proposition}[Reporting consistency]\label{prop:reporting:consistency}
For any $\boldsymbol{\lambda}\in[0,1]\times\mathbb{R}$,
$$\frac{\sum_{i=1}^Nn_i\hat{\theta}_i^{\mathrm{o}}(\boldsymbol{\lambda})}{\sum_{i=1}^Nn_i}=\frac{\sum_{i=1}^{N}Y_i}{\sum_{i=1}^{N}n_i}\ \mbox{ and }\ \frac{\sum_{i=1}^Nn_{i1}\hat{\theta}_{i1}^{\mathrm{o}}(\boldsymbol{\lambda})}{\sum_{i=1}^Nn_{i1}}-\frac{\sum_{i=1}^Nn_{i2}\hat{\theta}_{i2}^{\mathrm{o}}(\boldsymbol{\lambda})}{\sum_{i=1}^Nn_{i2}}=\frac{\sum_{i=1}^{N}Y_{i1}}{\sum_{i=1}^{N}n_{i1}}-\frac{\sum_{i=1}^{N}Y_{i2}}{\sum_{i=1}^{N}n_{i2}}.$$
\end{proposition}
This feature of reporting consistency is typically not satisfied by standard estimators (e.g., machine learning models, empirical Bayes), whose fitted unit-level values need not aggregate back to the raw overall proportion. Such inconsistency could be problematic in litigation and regulatory settings, where differences between the reported overall rate and the model-based unit-level estimates may call the analysis into question. In applications such as discrimination or pay‐equity studies, where a single firm-level figure is reported to regulators or courts, our estimator avoids this problem: the job-level estimates always average exactly to the overall firm-level rate used for reporting. 

\subsection{Cross-Fitting with Covariates}
We use the $K$-fold cross-fitting method to estimate $\hat{g}(\cdot)$ for the one-sample estimator or $\hat{g}_{\ell}(\cdot)$, $\forall\ell\in\{1,2\}$ for the two-sample estimator. In particular, we split the sample index $[N]$ into $K$ disjoint folds $\mathcal{I}_1,\ldots,\mathcal{I}_K$. For any $i\in[N]$, let $k(i)$ denote the fold that the $i$-th sample belongs to. Notice that it is not required that the folds have exactly equal size, for example each $k(i)$ could be drawn uniformly from $[K]$. We proceed with having equal-size folds for simplicity, without loss of generality. We let $\hat{g}^{-k}(\cdot)$ (resp. $\hat{g}_{\ell}^{-k}(\cdot)$, $\forall\ell\in\{1,2\}$) denote $\hat{g}(\cdot)$ (resp. $\hat{g}_{\ell}(\cdot), \forall\ell\in\{1,2\}$) computed without using observations from fold $k$. Further, we clip $\hat{g}(\cdot)$, $\hat{g}_{\ell}(\cdot), \forall\ell\in\{1,2\}$ between $0$ and $1$ as the final ML outputs. 
Specifically, given any $k\in[K]$, for any $i\in\mathcal{I}_k$, we set $\hat{g}(\mathbf{X}_i)=\hat{g}^{-k}(\mathbf{X}_i)$ in \eqref{eq:two-sample:estimator:lambda:one-sample}, and we set $\hat{g}_{\ell}(\mathbf{X}_{i\ell})=\hat{g}_{\ell}^{-k}(\mathbf{X}_{i\ell})$ in \eqref{eq:ell:two-sample:estimator}. 

Then the one-sample binomial estimator parametrized by $\boldsymbol{\lambda}=(\lambda_1,\lambda_2)\in[0,1]\times\mathbb{R}$ can be rewritten as 
\begin{equation}\label{eq:one-sample:estimator:cross-fitted}
\hat{\theta}_i^{\mathrm{o}}(\boldsymbol{\lambda})=\lambda_1\frac{Y_i}{n_i} + (1-\lambda_1)\frac{\sum_{j=1}^{N}Y_j}{\sum_{i=1}^{N}n_i}+\lambda_2\left(\hat{g}^{-k(i)}(\mathbf{X}_i)-\frac{\sum_{k=1}^K\sum_{j\in\mathcal{I}_k}n_j\hat{g}^{-k}(\mathbf{X}_{j})}{\sum_{j=1}^Nn_j}\right).
\end{equation}
Similarly, both components \eqref{eq:ell:two-sample:estimator} of the two-sample binomial estimator parametrized by $\boldsymbol{\lambda}=(\lambda_1,\lambda_2)\in[0,1]\times\mathbb{R}$ for each $\ell\in\{1,2\}$ can be rewritten as 
\begin{equation}\label{eq:two-sample:estimator:cross-fitted:ell}
\hat{\theta}_{i\ell}^{\mathrm{o}}(\boldsymbol{\lambda})=\lambda_1\frac{Y_{i\ell}}{n_{i\ell}} + (1-\lambda_1)\frac{\sum_{j=1}^{N}Y_{j\ell}}{\sum_{i=1}^{N}n_{i\ell}}+\lambda_2\left(\hat{g}_{\ell}^{-k(i)}(\mathbf{X}_{i\ell})-\frac{\sum_{k=1}^K\sum_{j\in\mathcal{I}_k}n_{j\ell}\hat{g}_{\ell}^{-k}(\mathbf{X}_{j\ell})}{\sum_{j=1}^Nn_{j\ell}}\right),
\end{equation}
and 
\begin{equation}
  \label{eq:two-sample:estimator:lambda}
  \hat{\theta}_i^{\mathrm{t}}(\boldsymbol\lambda) = \hat{\theta}_{i1}^{\mathrm{o}}(\boldsymbol{\lambda}) - \hat{\theta}_{i2}^{\mathrm{o}}(\boldsymbol{\lambda}),\ \boldsymbol{\lambda}=(\lambda_1,\lambda_2)\in[0,1]\times\mathbb{R},
\end{equation}
We refer to the estimators \eqref{eq:one-sample:estimator:cross-fitted} and \eqref{eq:two-sample:estimator:lambda} as the \textit{one-sample binomial-shrinkage estimator} and \textit{two-sample binomial-shrinkage estimator} respectively. Both estimators interpolate the \textit{maximum likelihood estimator} (MLE), the grand mean and the machine learning (ML) model estimator. They can also be viewed as shrinking MLE towards the grand mean and the ML estimates. In the absence of covariates, we simply set $\lambda_2 = 0$, so the estimators for this case reduce to special instances of \eqref{eq:one-sample:estimator:cross-fitted} and \eqref{eq:two-sample:estimator:lambda}.

\subsection{Approximate Binomial SURE}
For any $\boldsymbol{\lambda}\in[0,1]\times\mathbb{R}$, define 
\begin{equation}\label{eq:L2:one-sample:lambda}
    L_{\mathrm{o}}(\boldsymbol{\lambda}):=\frac{1}{N}\sum_{i=1}^N\left\{\mathbb{E}\left[\hat{\theta}_i^{\mathrm{o}}(\boldsymbol{\lambda})^2\right]-2\theta_i^{\mathrm{o}}\mathbb{E}\left[\hat{\theta}_i^{\mathrm{o}}(\boldsymbol{\lambda})\right]\right\},
\end{equation}
\begin{equation}\label{eq:L2:two-sample:lambda}
    L_{\mathrm{t}}(\boldsymbol{\lambda}):=\frac{1}{N}\sum_{i=1}^N\left\{\mathbb{E}\left[\hat{\theta}_i^{\mathrm{t}}(\boldsymbol{\lambda})^2\right]-2(\theta_{i1}^{\mathrm{t}}-\theta_{i2}^{\mathrm{t}})\mathbb{E}\left[\hat{\theta}_i^{\mathrm{t}}(\boldsymbol{\lambda})\right]\right\}.
\end{equation}
Thus, \eqref{eq:L2:one-sample:lambda} is equivalent to the one-sample objective function \eqref{eq:unweighted:one-sample} with respect to the binomial shrinkage estimator $\hat{\theta}_i^{\mathrm{o}}(\boldsymbol{\lambda})$, and \eqref{eq:L2:two-sample:lambda} is equivalent to the two-sample objective function \eqref{eq:unweighted:L2:risk} with respect to the binomial shrinkage estimator $\hat{\theta}_i^{\mathrm{t}}(\boldsymbol{\lambda})$.

Recall that the functional $\mathcal{T}$ is defined as \eqref{eq:T}, where we omit $a$ in the notation since throughout we set $a=\lfloor n/2\rfloor$ in \eqref{eq:T}. In order to align the form with the definition of the operator $\mathcal{T}$ in \eqref{eq:T}, we write the one-sample binomial shrinkage estimator in \eqref{eq:one-sample:estimator:cross-fitted} as $\hat{\theta}_i^{\mathrm{o}}(Y_i;n_i| \boldsymbol{\lambda})$ to emphasize its dependence on $Y_i$ and $n_i$. Specifically, we define $\mathcal{T}\hat{\theta}_i^{\mathrm{o}}(\boldsymbol{\lambda}):=\mathcal{T}\hat{\theta}_i^{\mathrm{o}}(Y_i;n_i|\boldsymbol{\lambda})$, where $\mathcal{T}\hat{\theta}_i^{\mathrm{o}}(Y_i;n_i|\boldsymbol{\lambda})$ denotes the term where we apply functional $\mathcal{T}$ to $\hat{\theta}_i^{\mathrm{o}}(\boldsymbol{\lambda})$ by fixing the grand mean $(\sum_{i=1}^NY_i)/(\sum_{i=1}^Nn_i)$ and the ML model outputs $\hat{g}^{-k(j)}(\mathbf{X}_j)$, $\forall j\in[N]$. 

For the two-sample binomial-shrinkage estimator \eqref{eq:two-sample:estimator:lambda}, we write it as $\hat{\theta}_i^{\mathrm{t}}(Y_{i1};n_{i1}|\boldsymbol{\lambda})$ (resp. $\hat{\theta}_i^{\mathrm{t}}(Y_{i2};n_{i2}|\boldsymbol{\lambda})$) to emphasize it dependence on the parameter $\boldsymbol{\lambda}$, $Y_{i1}$ and $n_{i1}$ (resp. $Y_{i2}$ and $n_{i2}$). Specifically, for any $\ell\in\{1,2\}$, $\mathcal{T}\hat{\theta}_i^{\mathrm{t}}(Y_{i\ell};n_{i\ell}|\boldsymbol{\lambda})$ denotes the term where we apply the functional $\mathcal{T}$ to $\hat{\theta}_i^{\mathrm{t}}(\boldsymbol{\lambda})$ by fixing $\{Y_{im}$, $n_{im}\}$ for $m=3-\ell$, the grand means $(\sum_{i=1}^NY_{i\ell})/(\sum_{i=1}^Nn_{i\ell}),\ \forall\ell\in\{1,2\}$ and the ML model outputs $\hat{g}^{-k(j)}(\mathbf{X}_{j\ell})$ for any $\ell\in\{1,2\}, j\in[N]$. 
Define 
\begin{equation}\label{eq:one-sample:SURE:lambda}
    \hat{L}_{\mathrm{o}}(\boldsymbol{\lambda}):=\frac{1}{N}\sum_{i=1}^{N}\left\{\hat{\theta}_i^{\mathrm{o}}(\boldsymbol{\lambda})^2-2\mathcal{T}\hat{\theta}_i^{\mathrm{o}}(\boldsymbol{\lambda})\right\},
\end{equation}
\begin{equation}\label{eq:two-sample:SURE:lambda}
    \hat{L}_{\mathrm{t}}(\boldsymbol{\lambda}):=\frac{1}{N}\sum_{i=1}^{N}\left\{\hat{\theta}_i^{\mathrm{t}}(\boldsymbol{\lambda})^2-2\mathcal{T}\hat{\theta}_{i}^{\mathrm{t}}(Y_{i1};n_{i1}|\boldsymbol{\lambda})+2\mathcal{T}\hat{\theta}_{i}^{\mathrm{t}}(Y_{i2};n_{i2}|\boldsymbol{\lambda})\right\}.
\end{equation}
The explicit expressions obtained by expanding \eqref{eq:one-sample:SURE:lambda} and \eqref{eq:two-sample:SURE:lambda} are given in \eqref{eq:SURE:one-sample} and \eqref{eq:L2:risk:estimator} in the Appendix~\ref{appendix:explicit:SURE:approx}. 

To derive the bias bounds in the one- and two-sample settings, we impose the following assumptions. 

\begin{assumption}\label{assump:one_sample}
In the one-sample setting,
\begin{enumerate}[(a)]
\item (Bounded size parameters) for any $i\in [N]$, $2 \le n_i \le \bar n$.
\item (Consistent cross-fit prediction models) $\displaystyle\max_{k\in[K]}\sup_{x\in\mathcal{X}}|\hat{g}^{-k}(x)-g(x)|=\mathrm{o}_p(1)$
\end{enumerate}
\end{assumption}

\begin{assumption}\label{assump:two_sample}
In the two-sample setting,
\begin{enumerate}[(a)]
\item (Bounded size parameters) for any $i\in [N]$ and $\ell\in \{1,2\}$, $2 \le n_{i\ell} \le \bar n$.
\item (Consistent cross-fit prediction models) $\displaystyle\max_{k\in[K], \ell\in \{1,2\}}\sup_{x\in\mathcal{X}}|\hat{g}_\ell^{-k}(x)-g_\ell(x)|=\mathrm{o}_p(1)$
\end{enumerate}
\end{assumption}

\begin{remark}
Part (a) of Assumptions \ref{assump:one_sample} and \ref{assump:two_sample} require uniformly bounded sample sizes. Though we can relax it to allow $\max_i n_i$ to grow with $n$ at a slow rate, we stick with the simpler one to avoid mathematical complications. Part (b) of both assumptions imposes consistency of the cross-fitted estimators. Our second condition in statement (ii) is weaker than the commonly used $\mathrm{o}_p\left(N^{-1/4}\right)$ rate, which typically appears as $\displaystyle\max_{i\in[N]}\mathbb{E}\left[|\hat{g}^{-k(i)}(\mathbf{X}_i)-g(\mathbf{X}_i)|^2\right]=\mathrm{o}(1/\sqrt{N})$ or $\displaystyle\max_{i\in[N]}\mathbb{E}\left[|\hat{g}_{\ell}^{-k(i)}(\mathbf{X}_i)-g_{\ell}(\mathbf{X}_i)|^2\right]=\mathrm{o}(1/\sqrt{N})$
\citep[e.g.][]{chernozhukov2018double,newey2018cross}. 
\end{remark}

Proposition~\ref{prop:approximate:sure} below implies that $\hat{L}_{\mathrm{o}}(\boldsymbol{\lambda})$ is an approximate SURE for $L_{\mathrm{o}}(\boldsymbol{\lambda})$ defined in \eqref{eq:L2:one-sample:lambda}, and that $\hat{L}_{\mathrm{t}}(\boldsymbol{\lambda})$ is an approximate SURE for $L_{\mathrm{t}}(\boldsymbol{\lambda})$ defined in \eqref{eq:unweighted:L2:risk}. 

\begin{proposition}[Bias Bound for Binomial SUREs]\label{prop:approximate:sure}
For any $\Lambda>0$,
~\begin{enumerate}[(i)]
\item Under Assumption \ref{assump:one_sample}, 
\[\max_{\boldsymbol{\lambda}\in[0,1]\times[-\Lambda,\Lambda]}\left|\mathbb{E}[\hat{L}_{\mathrm{o}}(\boldsymbol{\lambda})]-L_{\mathrm{o}}(\boldsymbol{\lambda})\right|\le \frac{\bar{n}}{N} + o(1).\]
\item Under Assumption \ref{assump:two_sample}, 
\[\max_{\boldsymbol{\lambda}\in[0,1]\times[-\Lambda,\Lambda]}\left|\mathbb{E}[\hat{L}_{\mathrm{t}}(\boldsymbol{\lambda})]-L_{\mathrm{t}}(\boldsymbol{\lambda})\right|\le \frac{4\bar{n}}{N} + o(1).\]
\end{enumerate}
\end{proposition}
Proposition~\ref{prop:approximate:sure} shows that, as $N\rightarrow \infty$, $\mathbb{E}[\hat{L}_{\mathrm{o}}(\boldsymbol{\lambda})]$ and $\mathbb{E}[\hat{L}_{\mathrm{t}}(\boldsymbol{\lambda})]$ closely approximate the true one-sample and two-sample objectives $L_{\mathrm{o}}(\boldsymbol{\lambda})$ and $L_{\mathrm{t}}(\boldsymbol{\lambda})$ up to the constant terms, respectively. Hence, $\hat{L}_{\mathrm{o}}(\boldsymbol{\lambda})$ and $\hat{L}_{\mathrm{t}}(\boldsymbol{\lambda})$ are approximate SUREs for their corresponding objectives. A natural approach, therefore, is to use the minimizers of $\hat{L}_{\mathrm{o}}(\boldsymbol{\lambda})$ and $\hat{L}_{\mathrm{t}}(\boldsymbol{\lambda})$ as approximations to the minimizers of $L_{\mathrm{o}}(\boldsymbol{\lambda})$ and $L_{\mathrm{t}}(\boldsymbol{\lambda})$, respectively. 


\section{Theoretical Analysis}
We note that using the class of estimators as proposed in \eqref{eq:one-sample:estimator:cross-fitted} and \eqref{eq:two-sample:estimator:lambda} the one-sample and two-sample objectives \eqref{eq:L2:one-sample:lambda}, \eqref{eq:L2:two-sample:lambda}, together with their approximate SUREs, are quadratic functions of $\boldsymbol{\lambda}$, as summarized in the following proposition:
\begin{proposition}\label{prop:quadratic}
There exists semipositive definite matrices $\mathbf{C}_{N,2}, \mathbf{D}_{N,2}, \mathbf{C}_2,\mathbf{D}_2\in\mathbb{R}^{2\times 2}$, vectors $\mathbf{C}_{N,1}, \mathbf{D}_{N,1}, \mathbf{C}_1, \mathbf{D}_1\in\mathbb{R}^{2\times1}$, constants $C_0,D_0,C_0^*,D_0^*\in\mathbb{R}$, such that for any $\boldsymbol{\lambda}\in[0,1]\times\mathbb{R}$,
$$\hat{L}_{\mathrm{o}}(\boldsymbol{\lambda})=\boldsymbol{\lambda}^T\mathbf{C}_{N,2}\boldsymbol{\lambda}+\mathbf{C}_{N,1}^T\boldsymbol{\lambda}+C_0,\ \hat{L}_{\mathrm{t}}(\boldsymbol{\lambda})=\boldsymbol{\lambda}^T\mathbf{D}_{N,2}\boldsymbol{\lambda}+\mathbf{D}_{N,1}^T\boldsymbol{\lambda}+D_0,$$ 
$$L_{\mathrm{o}}(\boldsymbol{\lambda})=\boldsymbol{\lambda}^T\mathbf{C}_{2}\boldsymbol{\lambda}+\mathbf{C}_{1}^T\boldsymbol{\lambda}+C_0^*,\ L_{\mathrm{t}}(\boldsymbol{\lambda})=\boldsymbol{\lambda}^T\mathbf{D}_{2}\boldsymbol{\lambda}+\mathbf{D}_{1}^T\boldsymbol{\lambda}+D_0^*,$$ 
where $\mathbf{C}_{N,2}, \mathbf{C}_{N,1}, \mathbf{C}_2, \mathbf{C}_1, \mathbf{D}_{N,2}, \mathbf{D}_{N,1}, \mathbf{D}_2$ and $\mathbf{D}_1$ are given in \eqref{eq:coefficient:quadratic:ML}, \eqref{eq:coeff:linear:one-sample}, \eqref{eq:coefficient:quadratic:ML:objective}, \eqref{eq:coeff:linear:one-sample:objective}, \eqref{eq:coefficient:quadratic:ML:two-sample}, \eqref{eq:D:N:1}, \eqref{eq:coefficient:quadratic:two-sample:objective}, and \eqref{eq:coeff:linear:two-sample:objective}, respectively, and $C_0, D_0, C_0^*, D_0^*$ are constants independent of $\boldsymbol{\lambda}$.
\end{proposition}
By first-order condition, Proposition~\ref{prop:quadratic} implies that the unconstrained one-sample parameter estimator $\hat{\boldsymbol{\lambda}}_{\mathrm{o}}$\footnote{In the following we use subscript $\mathrm{o}$ (resp. $\mathrm{t}$) with $\hat{\boldsymbol{\lambda}}_\mathrm{o}$ (resp. $\hat{\boldsymbol{\lambda}}_\mathrm{t}$) or $\boldsymbol{\lambda}_\mathrm{o}^*$ (resp. $\boldsymbol{\lambda}_\mathrm{t}^*$) to imply the one-sample (resp. two-sample) estimated $\hat{\boldsymbol{\lambda}}$ or the $\boldsymbol{\lambda}^*$ as the true minimizer related to the one-sample (resp. two-sample) $L_2$ loss function.} that minimizes $\hat{L}_{\mathrm{o}}(\boldsymbol{\lambda})$ over $\boldsymbol{\lambda}\in\mathbb{R}\times\mathbb{R}$ is 
\begin{equation}\label{eq:unconstrained:one-sample:optimal}
\hat{\boldsymbol\lambda}_{\mathrm{o}}=-\frac{1}{2}\mathbf{C}_{N,2}^{-1}\mathbf{C}_{N,1}.
\end{equation}
The unconstrained two-sample parameter estimator $\hat{\boldsymbol{\lambda}}_{\mathrm{t}}$ that minimizes $\hat{L}_{\mathrm{t}}(\boldsymbol{\lambda})$ over $\boldsymbol{\lambda}\in\mathbb{R}\times\mathbb{R}$ is 
\begin{equation}\label{eq:unconstrained:two-sample}
\hat{\boldsymbol{\lambda}}_{\mathrm{t}}=-\frac{1}{2}\mathbf{D}_{N,2}^{-1}\mathbf{D}_{N,1}.
\end{equation}
Let $\boldsymbol{\lambda}_{\mathrm{o}}^*$ and $\boldsymbol{\lambda}_{\mathrm{t}}^*$ denote the minimizers of $L_{\mathrm{o}}(\boldsymbol{\lambda})$ \eqref{eq:L2:one-sample:lambda} and $L_{\mathrm{t}}(\boldsymbol{\lambda})$ \eqref{eq:L2:two-sample:lambda} over $[0,1]\times\mathbb{R}$, respectively. We refer to $\hat{\boldsymbol{\lambda}}_\mathrm{o}$ and $\hat{\boldsymbol{\lambda}}_\mathrm{t}$ as the \textit{one-sample approximate SURE} and the \textit{two-sample approximate SURE}, respectively. 

\subsection{Asymptotic Normality and Regret Analysis}





\noindent When $\boldsymbol{\lambda}_{\mathrm{o}}^*$ and $\boldsymbol{\lambda}_{\mathrm{t}}^*$ are unconstrained, i.e. $\lambda_{\mathrm{o}1}^*\in(0,1)$ and $\lambda_{\mathrm{t}1}^*\in(0,1)$, the asymptotic distributions of $\hat{\boldsymbol{\lambda}}_{\mathrm{o}}$ and $\hat{\boldsymbol{\lambda}}_{\mathrm{t}}$ are as follows:
\begin{theorem}[Asymptotic Normality]\label{thm:asymptotics}
Suppose Assumption \ref{assump:one_sample} and Assumption~\ref{ass:additional:one-sample} in Appendix~\ref{appendix:one-sample:asymptotis} hold, and $\boldsymbol{\lambda}_{\mathrm{o}}^*$ is unconstrained. Then 
$$\sqrt{N}\left(\hat{\boldsymbol{\lambda}}_{\mathrm{o}}-\boldsymbol{\lambda}_{\mathrm{o}}^*\right)\rightsquigarrow\mathcal{N}(\mathbf{0},\mathbf{V}),$$ 
where $\mathbf{V}\preceq\bar{C}\mathbf{I}_2$ for some absolute constant $\bar{C}$. Suppose Assumption \ref{assump:two_sample} and Assumption \ref{ass:additional:two-sample} in Appendix~\ref{appendix:two-sample:case} hold, and $\boldsymbol{\lambda}_{\mathrm{t}}^*$ is unconstrained. Then 
$$\sqrt{N}\left(\hat{\boldsymbol{\lambda}}_{\mathrm{t}}-\boldsymbol{\lambda}_{\mathrm{t}}^*\right)\rightsquigarrow\mathcal{N}(\mathbf{0},\overline{\mathbf{V}}),$$ 
where $\overline{\mathbf{V}}\preceq\bar{C}'\mathbf{I}_2$ for some absolute constant $\bar{C}'$.   
\end{theorem}
Assumption~\ref{ass:additional:one-sample} and Assumption~\ref{ass:additional:two-sample} include Lindeberg–Feller type conditions, which are used to establish a multivariate central limit theorem for i.n.i.d. samples. Recall from \eqref{eq:two-sample:estimator:lambda:one-sample} that the feasible parameter $\boldsymbol{\lambda}=(\lambda_1,\lambda_2)\in[0,1]\times\mathbb{R}$. The asymptotic normality result in Theorem \ref{thm:one-sample:asymptotics} holds only when $\boldsymbol\lambda_{\mathrm{o}}^*$ and $\boldsymbol\lambda_{\mathrm{t}}^*$ lie in the interior of the feasible region. So the validity of statistical inference based on Theorem \ref{thm:one-sample:asymptotics} fails when $\boldsymbol\lambda_{\mathrm{o}}^*$ or $\boldsymbol\lambda_{\mathrm{t}}^*$ is on the boundary (i.e., $\boldsymbol{\lambda}_{\mathrm{o}1}^*\in\{0,1\}$ or $\boldsymbol{\lambda}_{\mathrm{t}1}^*\in\{0,1\}$). In such cases, we adopt the inference method for constrained extremum estimators proposed by \citet{li2024inference}, with details provided in \ref{subsub:inference:constrained}.

The regret bounds for both one-sample and two-sample objectives follow immediately from Theorem \ref{thm:asymptotics} as follows:
\begin{theorem}[Regret Bounds]\label{thm:regret}
Suppose Assumption \ref{assump:one_sample} and Assumption~\ref{ass:additional:one-sample} hold, and $\boldsymbol{\lambda}_{\mathrm{o}}^*$ is unconstrained, then $\left|L_{\mathrm{o}}(\hat{\boldsymbol{\lambda}}_{\mathrm{o}})-L_{\mathrm{o}}(\boldsymbol{\lambda}_{\mathrm{o}}^*)\right|=\mathrm{O}_{p}\left(\frac{1}{N}\right)$. Suppose Assumption \ref{assump:two_sample} and Assumption \ref{ass:additional:two-sample} hold, and $\boldsymbol{\lambda}_{\mathrm{t}}^*$ is unconstrained, then $\left|L_{\mathrm{t}}(\hat{\boldsymbol{\lambda}}_{\mathrm{t}})-L_{\mathrm{t}}(\boldsymbol{\lambda}_{\mathrm{t}}^*)\right|=\mathrm{O}_{p}\left(\frac{1}{N}\right)$.   
\end{theorem}

\subsection[Statistical Inference for lambda]{Statistical Inference for $\boldsymbol{\lambda}$}
In some applications, it may be costly to replace a simple status-quo estimator,  such as the empirical average, with the more sophisticated estimator given by our SURE method. As a result, the policy maker would adopt the new estimator only if it has a sufficient efficiency gain over the incumbent. Suppose the class of estimators nests the current estimator with $\boldsymbol{\lambda} = \boldsymbol{\lambda}_0$, we can formulate the problem of whether to use $\hat{\boldsymbol{\lambda}}$ in place of $\boldsymbol{\lambda}_0$ through a hypothesis test with the null hypothesis $H_0: \boldsymbol{\lambda}^* = \boldsymbol{\lambda}_0$.

In this section, we consider the more general problem of on constructing confidence regions for $\boldsymbol{\lambda}^*$ based on the approximate SURE estimators $\hat{\boldsymbol{\lambda}}_{\mathrm{o}}$ and $\hat{\boldsymbol{\lambda}}_{\mathrm{t}}$. Specifically, we consider two scenarios. First, if the true value of $\boldsymbol{\lambda}_{\mathrm{o}}^*$ or $\boldsymbol{\lambda}_{\mathrm{t}}^*$ is believed to lie in the interior of the parameter space $[0,1]\times\mathbb{R}$, we construct the confidence region using standard inference techniques for unconstrained estimators, based upon the asymptotic normality results of Theorem \ref{thm:one-sample:asymptotics} and Theorem \ref{thm:CLT:two-sample}. Alternatively, if the true value of $\boldsymbol{\lambda}_{\mathrm{o}}^*$ or $\boldsymbol{\lambda}_{\mathrm{t}}^*$ is suspected to lie on the boundary of $[0,1]\times\mathbb{R}$, we apply inference methods tailored to constrained estimators following \citet{li2024inference}.

\subsubsection{Inference for the Unconstrained Case}\label{subsub:inference:unconstrained}
To perform statistical inference on the unconstrained $\boldsymbol{\lambda}^*$ (where $\boldsymbol{\lambda}^*=\boldsymbol{\lambda}_{\mathrm{o}}^*$ for one-sample case and $\boldsymbol{\lambda}^*=\boldsymbol{\lambda}_{\mathrm{t}}^*$ for two-sample case), we utilize the asymptotic normality results derived for unconstrained $\boldsymbol{\lambda}$ in Theorem \ref{thm:one-sample:asymptotics} and Theorem \ref{thm:CLT:two-sample},  and estimate the variance of $\hat{\boldsymbol{\lambda}}$ using bootstrap methods (where $\hat{\boldsymbol{\lambda}}=\hat{\boldsymbol{\lambda}}_{\mathrm{o}}$ for one-sample case and $\hat{\boldsymbol{\lambda}}=\hat{\boldsymbol{\lambda}}_{\mathrm{t}}$ for two-sample case). Specifically, we generate $B$ bootstrap samples, and compute $\{\hat{\boldsymbol{\lambda}}^{(b)}\}_{b\in[B]}$ for each bootstrap sample, and then compute the covariance matrix $\hat{\boldsymbol{V}}$ of the resulting bootstrap estimates $\hat{\boldsymbol{\lambda}}^{(b)}$. Then we construct the level $1-\alpha$ confidence set for $\boldsymbol{\lambda}^*$ as 
\begin{equation}\label{eq:confidence:set:one-sample}
    \mathcal{C}_{1-\alpha} = \left\{\boldsymbol{\lambda}: N(\hat{\boldsymbol{\lambda}}-\boldsymbol{\lambda})'\hat{\boldsymbol{V}}^{-1}(\hat{\boldsymbol{\lambda}}-\boldsymbol{\lambda})\leq\chi_{2,1-\alpha}^2\right\},
\end{equation}
where $\chi_{2,1-\alpha}^2$ is the chi-square critical value with $2$ degrees of freedom. 

\subsubsection{Inference for the Constrained Case}\label{subsub:inference:constrained}
For statistical inference on the constrained $\boldsymbol{\lambda}^*$, we follow the procedure outlined in Section 2 of \citet{li2024inference} to construct the confidence set. Specifically, we detail the steps for computing the confidence set for the constrained estimator in the two-sample case, as this method is only applied to the two-sample discrimination report application in the empirical application in Section~\ref{sec:empirical}, whose $\hat{\boldsymbol{\lambda}}=\hat{\boldsymbol{\lambda}}_{\mathrm{t}}$ \eqref{eq:lambda:two-sample:discrimination} falls on the boundary of the feasible region (thus it is believed that the true $\boldsymbol{\lambda}_{\mathrm{t}1}^*=0$ for this application). 
The procedure for the one-sample case is very similar, and thus is omitted here for brevity.

Suppose we have already computed $\mathbf{D}_{N,1}, \mathbf{D}_{N,2}$ and $\hat{\boldsymbol{\lambda}}=\hat{\boldsymbol{\lambda}}_{\mathrm{t}}$ according to Lemma \ref{lemma:coeff:two-sample:SURE} and \eqref{eq:unconstrained:two-sample}, the steps to construct the confidence set are as follows: 
\begin{itemize}
    \item[(1)] Repeat for $B$ bootstrap iterations: draw a bootstrap sample $\mathbf{Z}_1^*,\ldots,\mathbf{Z}_n^*$ and compute $\mathbf{D}_{N,1}^*$, $\mathbf{D}_{N,2}^*$ in the same way as computing $\mathbf{D}_{N,1}, \mathbf{D}_{N,2}$ in the original dataset. Then compute $-\displaystyle\inf_{h\in\mathbb{R}^2}\hat{H}_n(h)$, where 
    \begin{equation}\label{eq:objective:bootstrap}
        \hat{H}_n(h)=\frac{1}{2}h'\{2\mathbf{D}_{N,2}^*-\mathbf{D}_{N,2}\}h+n^{\gamma}h'\{2(\mathbf{D}_{N,2}^*-\mathbf{D}_{N,2})\hat{\boldsymbol{\lambda}}+(\mathbf{D}_{N,1}^*-\mathbf{D}_{N,1})\}.
    \end{equation}
    \item[(2)] Compute $\hat{c}_{1-\alpha}^*$, the $1-\alpha$ conditional quantile of $-\displaystyle\inf_{h\in\mathbb{R}^2}\hat{H}_n(h)$.
    \item[(3)] Choose some $\kappa\in(0,\infty]$ and sequence $\delta_n\rightarrow0$ satisfying $n^{\gamma}\delta_n\rightarrow\kappa$. 
    \item[(4)] Construct the uniformly asymptoically valid nominal $1-\alpha$ confidence set given by
    \begin{equation}\label{eq:confidence:set:1-alpha}
        \mathcal{C}_{1-\alpha}^*=\left\{\boldsymbol{\lambda}\in[0,1]\times\mathbb{R}: n^{2\gamma}\left(\hat{L}_t(\boldsymbol{\lambda})-\inf_{h\in\mathcal{C}_{\delta_n}^{\boldsymbol{\lambda}}}\hat{L}_t(\boldsymbol{\lambda}+h/n^{\gamma})\right)\leq\hat{c}_{1-\alpha}^*\right\},
    \end{equation}
    where
    \begin{equation}\label{eq:C:confidence:term}
        \mathcal{C}_{\delta_n}^{\boldsymbol{\lambda}}=\left\{h\in n^{\gamma}([0,1]\times\mathbb{R}-\boldsymbol{\lambda}): \frac{\|h\|}{n^{\gamma}}\leq\delta_n\right\},\ \ \delta_n\rightarrow0.
    \end{equation}
\end{itemize}
Under general regularity conditions, \citet{li2024inference} shows that the confidence set constructed as described above for the constrained estimator is uniformly valid. In our applications, we choose $\gamma=0.5$, $\delta_n=1/\sqrt{n}$.

\subsection{Performance Validation by Data Thinning}\label{subsubsec:data-thinning}
We close this section by introducing a procedure that splits a single binomial observation into two independent binomial observations that share the same parameter. This is useful in empirical settings where only aggregated binomial outcomes are available. For example, the total number of students in a school and how many pass a test, or how many are classified as innovators. We use the data thinning/fission method \citep{neufeld2024data,leiner2025data} to construct separate training and holdout samples from such aggregated outcomes, which allows us to compare our binomial-shrinkage estimators with alternative methods, as we do extensively in Section~\ref{sec:empirical}. This is in similar spirit to the Coupled Bootstrap method that validates the EB methods with Gaussian measurements \citep{chen2022empirical}.

For the one-sample case, suppose each observation indexed by $i\in[N]$ corresponds to $n_i$ binary samples, among which $Y_i$ of them have outcome $1$, so $Y_i\sim\mathrm{Binomial}(n_i,\theta_i^{\mathrm{o}})$. We select $m_i$ samples from the original $n_i$ samples, where $m_i < n_i$. We then generate $Y_i^{(1)}$ according to a hyper-geometric distribution with parameters $(m_i, n_i - m_i, Y_i)$, and define $Y_i^{(2)} = Y_i - Y_i^{(1)}$. We take $\{Y_i^{(2)},n_i - m_i\}_{i\in[N]}$ as the training set and $\{Y_i^{(1)}, m_i\}_{i\in[N]}$ as the holdout set. Consequently, according to \citet{neufeld2024data,leiner2025data}, we have 
$$Y_i^{(1)}\sim\mathrm{Binomial}(m_i,\theta_i^{\mathrm{o}}) \mbox{ and } Y_i^{(2)}\sim\mathrm{Binomial}(n_i - m_i,\theta_i^{\mathrm{o}}).$$ 
We might then compute the binomial estimators on the training set $\mathcal{F}_T^{\mathrm{o}}:=\{Y_i^{(2)},n_i - m_i\}_{i\in[N]}$ and honestly evaluate them with the holdout set $\mathcal{F}_H^{\mathrm{o}}:=\{Y_i^{(1)}, m_i\}_{i\in[N]}$. 

Similarly, in the two-sample application, suppose each observation indexed by $i\in[N]$ corresponds to two independent populations: population one with $n_{i1}$ samples (of which $Y_{i1}$ have outcome 1), and population two with $n_{i2}$ samples (of which $Y_{i2}$ have outcome 1). Thus, $Y_{i1}\sim\textrm{Binomial}(n_{i1},\theta_{i1})$, $Y_{i2}\sim\textrm{Binomial}(n_{i2},\theta_{i2})$, and the two-sample binomial parameter $\theta_i^{\mathrm{t}}=\theta_{i1}-\theta_{i2}$. To create the holdout set, we choose $m_{i1}<n_{i1}$ samples from population one and $m_{i2}<n_{i2}$ samples from population two. We then generate $Y_{i1}^{(1)}$ and $Y_{i2}^{(1)}$ independently from hypergeometric distributions with parameters $(m_{i1},n_{i1}-m_{i1},Y_{i1})$ and $(m_{i2},n_{i2}-m_{i2},Y_{i2})$, respectively, and define $Y_{i1}^{(2)}=Y_{i1}-Y_{i1}^{(1)}$ and $Y_{i2}^{(2)}=Y_{i2}-Y_{i2}^{(1)}$. As a result, we have independence: $Y_{i1}^{(1)}\indep Y_{i1}^{(2)}$ and $Y_{i2}^{(1)}\indep Y_{i2}^{(2)}$. According to \citet{neufeld2024data,leiner2025data}, we have
$$Y_{i1}^{(1)}\sim\mathrm{Binomial}(m_{i1},\theta_{i1}),\ Y_{i1}^{(2)}\sim\mathrm{Binomial}(n_{i1}-m_{i1},\theta_{i1}),$$   
$$Y_{i2}^{(1)}\sim\mathrm{Binomial}(m_{i2},\theta_{i2}),\ Y_{i2}^{(2)}\sim\mathrm{Binomial}(n_{i2}-m_{i2},\theta_{i2}).$$
Hence, we might take $\mathcal{F}_T^{\mathrm{t}}:=\{n_{i1}-m_{i1},Y_{i1}^{(2)},n_{i2}-m_{i2},Y_{i2}^{(2)}\}_{i\in[N]}$ as the training set, while $\mathcal{F}_H^{\mathrm{t}}:=\{m_{i1},Y_{i1}^{(1)},m_{i2},Y_{i2}^{(1)}\}_{i\in[N]}$ is taken as the holdout set. 

The following proposition shows that the data thinning/fission procedures yield unbiased estimators of the $L_2$ risk function for a generic estimator of the binomial parameters.
\begin{proposition}\label{prop:data:split}
Let $\hat{\theta}_i^{\mathrm{o}}$ be some generic estimator for the one-sample parameter $\theta_i^{\mathrm{o}}$ constructed from the training data $\mathcal{F}_T^{\mathrm{o}}$. Let $\hat{\theta}_i^{\mathrm{t}}$ be some generic estimator for the two-sample parameter $\theta_i^{\mathrm{t}}$ constructed from the training data $\mathcal{F}_T^{\mathrm{t}}$. Then 
$$\mathbb{E}\left[\frac{1}{N}\sum_{i=1}^N\left(\hat{\theta}_i^{\mathrm{o}}\right)^2-2\frac{Y_i^{(1)}}{m_i}\hat{\theta}_i^{\mathrm{o}}\right]=\frac{1}{N}\sum_{i=1}^N\mathbb{E}\left[\left(\hat{\theta}_i^{\mathrm{o}}\right)^2\right]-\theta_i^{\mathrm{o}}\mathbb{E}\left[\hat{\theta}_i^{\mathrm{o}}\right],$$
$$\mathbb{E}\left[\frac{1}{N}\sum_{i=1}^N\left(\hat{\theta}_i^{\mathrm{t}}\right)^2-2\left\{\frac{Y_{i1}^{(1)}}{m_{i1}}-\frac{Y_{i2}^{(1)}}{m_{i2}}\right\}\hat{\theta}_i^{\mathrm{t}}\right]=\frac{1}{N}\sum_{i=1}^N\mathbb{E}\left[\left(\hat{\theta}_i^{\mathrm{t}}\right)^2\right]-\theta_i^{\mathrm{t}}\mathbb{E}\left[\hat{\theta}_i^{\mathrm{t}}\right].$$
\end{proposition}
Proposition~\ref{prop:data:split} enables out-of-sample evaluation of binomial estimators even when separate training and holdout datasets are not directly available. We later apply this method to validate the performance of our estimators in the innovation and education empirical applications.

\section{Empirical Illustration}\label{sec:empirical}
To examine the practical performance of the one-sample and two-sample binomial-shrinkage estimators, we now present three empirical applications related to \citet{bell2019becomes}, \citet{kline2024discrimination} and \citet{gang2023ranking}. 

Using data on inventors from patent records linked to tax records, \citet{bell2019becomes} investigated the determinants of becoming a successful inventor and published an ``Opportunity Atlas'' detailing patent rates across various population groups, segmented by neighborhood, college attendance, parental income level, and racial background. 

Leveraging data from a large-scale resume correspondence experiment, which signaled race and gender to employers through randomly assigned distinctive names, \citet{kline2024discrimination} measured disparities in contact rates across race and gender categories, yielding noisy estimates of discriminatory behavior at the firm level. They subsequently constructed a "discrimination report card" summarizing experimental evidence on biases exhibited by a broad range of Fortune 500 companies using an Empirical Bayes approach. 

\citet{gang2023ranking} proposed a ranking and selection framework as an alternative to conventional false-discovery-rate analyses. They validated their approach using an empirical study of K-12 school test performance data, identifying significant differences in passing rates between students from socioeconomically advantaged (SEA) and disadvantaged (SED) backgrounds within each school. Specifically, they computed $p$-values based on a normal approximation, despite over $30\%$ of the schools having fewer than $100$ data points in at least one of the two groups.

\subsection{Data Description}
We demonstrate our methods on three applications, one in the one-sample setting and two in the two-sample settings. The one-sample problem we consider is the Opportunity Atlas innovation dataset\footnote{The dataset description is available at \url{https://opportunityinsights.org/wp-content/uploads/2018/04/Inventors-Codebook-Table-3.pdf}.}, which documents innovation rates across colleges by linking tax records to inventor information from patent records \citep{bell2019becomes}. 
The dataset contains the fraction of inventors, defined as individuals who were listed on a patent application between 2001 and 2012 or granted a patent between 1996 and 2014, from $423$ colleges. For college $i\in [423]$, we will treat this share measurement as $Y_i$ and the total number of students as $n_i$. In this application, $Y_i$s are generally close to zero.

 

The first two-sample application is to detect employment discrimation using a large-scale resume correspondence experiment dataset\footnote{The dataset is available at 
 \url{https://www.aeaweb.org/articles?id=10.1257/aer.20230700}} from \citet{kline2022systemic}. Each observation corresponds to a job applicant and contains his/her demographic information such as race and gender, the job identifier for which the candidate applied, an indicator of whether the candidate received a callback, and additional characteristics. In this dataset, every job receives up to four resumes from each of the two racial groups. We analyze $9821$ jobs, each with exactly four white applicants and four black applicants. Following \cite{kline2021reasonable, kline2022systemic, kline2024discrimination}, we measure the extent of discrimination for a job as the callback rate gap white and black job applicants. For any job $i\in [9821]$, $n_{i1} = n_{i2} = 4$,  $Y_{i1}$ and $Y_{i2}$ are the fractions of callbacks among the four applications in each racial group.
 
 The second two-sample application utilizes the dataset\footnote{The dataset we use is `AYP\_05.csv' from \url{https://github.com/bgang92/rankingselection}} on K-12 school test performance data drawn from the 2005 Annual Yearly Performance (AYP) study, which is analyzed in \cite{gang2023ranking}. Schools in this dataset are categorized into three types: 'H' for high schools, 'M' for middle schools, and 'E' for elementary schools. For each school $i\in[6398]$, $Y_{i1}$ and $Y_{i2}$ measure number of socially-economically advantaged (SEA) and socially-economically disadvantaged (SED) students who took and passed the test, and $n_{i1}$ and $n_{i2}$ measure the number of SEA and SED students. In this application, $(n_{i1}, n_{i2})$ are vastly heterogeneous. 


\subsection{SURE Estimates and Confidence Regions}
\paragraph{Reporting Innovation Rates.} 
We perform a linear regression to construct $\hat{g}(x)$ with 10-fold cross-fitting. In particular, we regress $Y_i$ on two covariates: `patent' (the total number of patents granted to students) and `total\_cites' (the total number of patent citations received by students). The regression model is specified as follows:
$$\textrm{inventor}=\beta_1\cdot\textrm{patent}+\beta_2\cdot\textrm{total\_cites}+\epsilon.$$
 We use the one-sample binomial-shrinkage estimator \eqref{eq:two-sample:estimator:lambda:one-sample} to estimate the inventor fraction for each institution. The constrained estimate $\hat{\boldsymbol\lambda}$ coincides with the unconstrained one, which is 
\begin{equation}\label{eq:lambda:innovation}
    \hat{\boldsymbol\lambda}=(\hat{\lambda}_1,\hat{\lambda}_2)=(0.9831, 0.0134).
\end{equation}
The resulting estimated inventor fractions for each college are illustrated in Figure~\ref{fig:estimated_theta:innovation} and the $95\%$ confidence region is plotted in Figure \ref{fig:confidence:innovation}. While $\lambda_1\approx 1$ and $\lambda_2\approx 0$, we can reject that the MLE or global mean is MSE-optimal.  

\paragraph{Estimating Employment Discrimination.} 
We apply a gradient boosted tree learner to classify the callback status (`cb') with 10-fold cross-fitting. 
The feature set includes job and market characteristics --- such as the state and census region of the employer, the experimental wave, and the industry of the job posting --- as well as application-level information, including the applicant’s implied age at submission, education credentials, and the order in which applications were sent. We additionally incorporate demographic and identity signals embedded in the resume, including gender indicators, age-over-40 status, involvement in LGBTQ, political, or academic organizations, and the use of gender-related pronouns. Finally, the model controls for experiment-design variables, including paired-application identifiers, balance indicators, occupational skill categories, and the number of experimental waves associated with each posting. We define $\hat{g}_{i1}$ and $\hat{g}_{i2}$ as the average predicted callback rates from the trained classifier for white and black applicants for each job $i$, respectively. We then interpolate between the maximum likelihood estimate (MLE), the machine learning predictions, and the grand mean within the estimator framework defined in \eqref{eq:two-sample:estimator:lambda}. 
The constrained estimate, $\hat{\boldsymbol\lambda}$, is computed accordingly as 
\begin{equation}\label{eq:lambda:two-sample:discrimination}
    \hat{\boldsymbol\lambda}=(\hat{\lambda}_1,\hat{\lambda}_2)=(0,0.2377).
\end{equation}
The resulting estimates of the callback rate differences for each job are depicted in Figure \ref{fig:estimated_theta:discrimination}  and the $95\%$ confidence region is plotted in Figure \ref{fig:confidence:discrimination}. In this case, we cannot reject the null that the global mean is MSE-optimal, though we can confidently claim that the MLE is so.

\paragraph{Estimating Test Passing Rate Gaps.} 
We use the school type as the only covariate and define $\hat{g}_{i1}$ and $\hat{g}_{i2}$ as the average passing rate of SEA and SED students within the same school type as school $i$. We then apply the same procedure as in the second application. The resulting constrained estimate $\hat{\boldsymbol\lambda}$ matches the unconstrained solution:
\begin{equation}\label{eq:lambda:education}
    \hat{\boldsymbol\lambda} = (\hat{\lambda}_1,\hat{\lambda}_2)=(0.6447,-4.4989).
\end{equation}
Figure \ref{fig:estimated_theta:education} plots the estimated differences in test passing rates between SEA and SED students across individual schools. The $95\%$ confidence region is plotted in Figure \ref{fig:confidence:discrimination}. Unlike the previous two cases, we find strong evidence that the MSE-optimal estimator within the class must utilize both the shrinkage to global mean and the assistance by the prediction model. 

\begin{figure}[htbp]
    \centering
    \begin{minipage}{0.7\textwidth} 
        \centering
        \includegraphics[width=\linewidth]{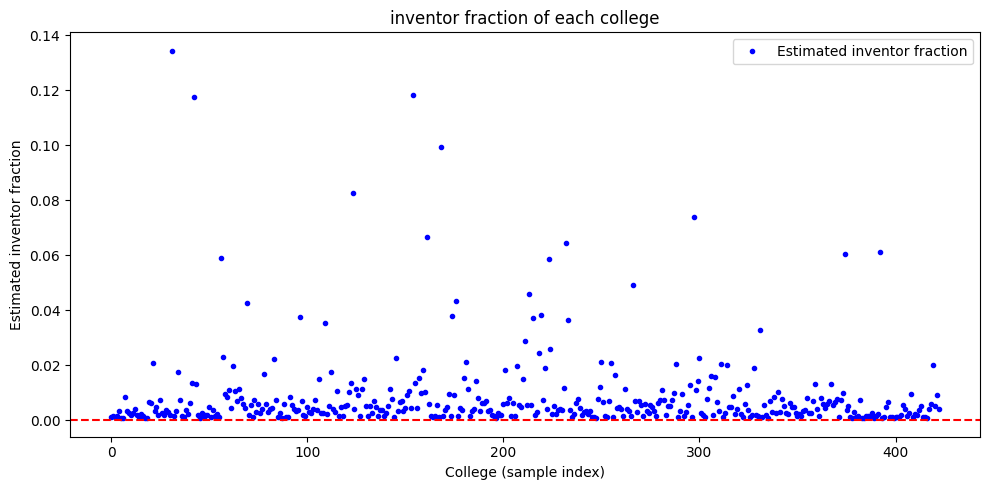}
        \caption{Estimated inventor fraction for each college}
        \label{fig:estimated_theta:innovation}
    \end{minipage}
    \vfill 
    \begin{minipage}{0.7\textwidth} 
        \centering
        \includegraphics[width=\linewidth]{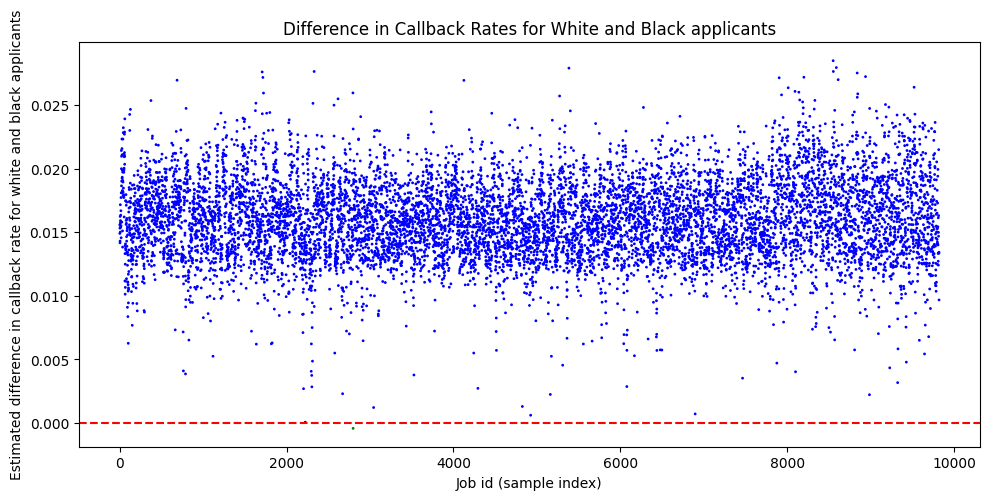}
        \caption{Estimated differences in callback rates between white and black applicants across all jobs.}
        \label{fig:estimated_theta:discrimination}
    \end{minipage}
    \vfill 
    \begin{minipage}{0.7\textwidth} 
        \centering
        \includegraphics[width=\linewidth]{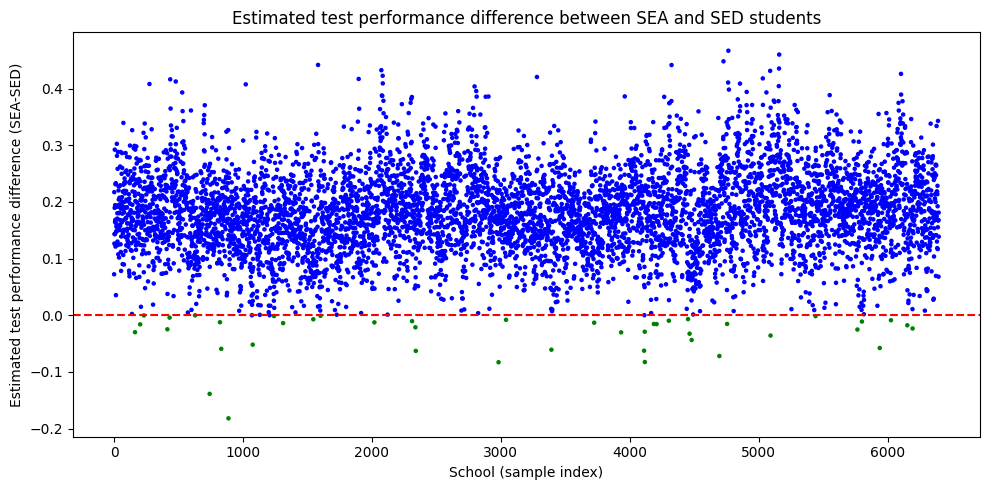}
        \caption{Difference in test performance between SEA and SED students in K–12 schools. Blue dots indicate positive point estimates (SEA students outperform SED students), while green dots indicate negative point estimates (SED students outperform SEA students).}
        \label{fig:estimated_theta:education}
    \end{minipage}
\end{figure}

\begin{figure}[htbp]
    \centering
    \begin{minipage}{0.6\textwidth} 
        \centering
        \includegraphics[width=\linewidth]{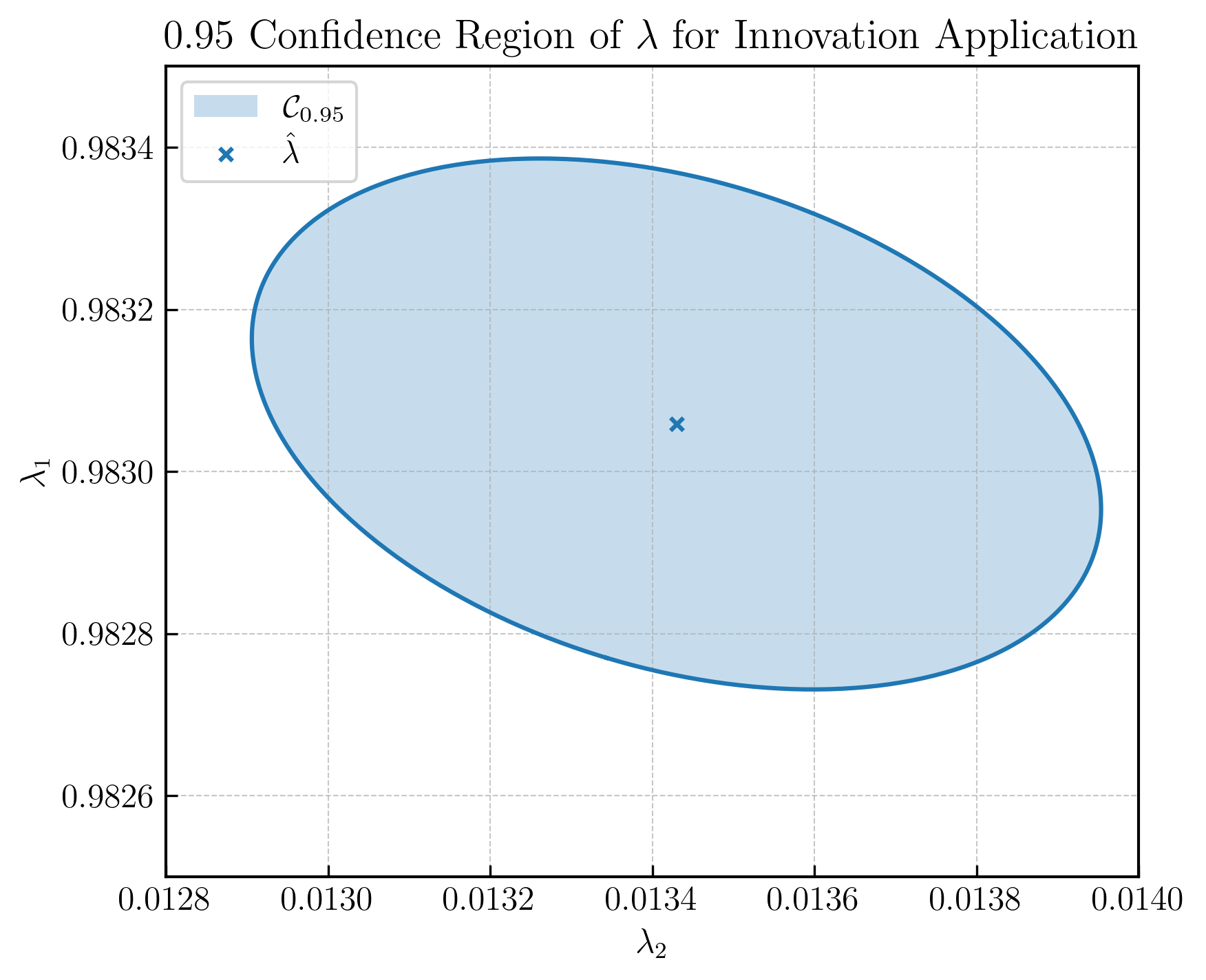}
        \caption{95\% confidence region of $\boldsymbol{\lambda}$ for innovation application.}
        \label{fig:confidence:innovation}
    \end{minipage}
    \vfill
    \begin{minipage}{0.8\textwidth} 
        \centering
        \includegraphics[width=\linewidth]{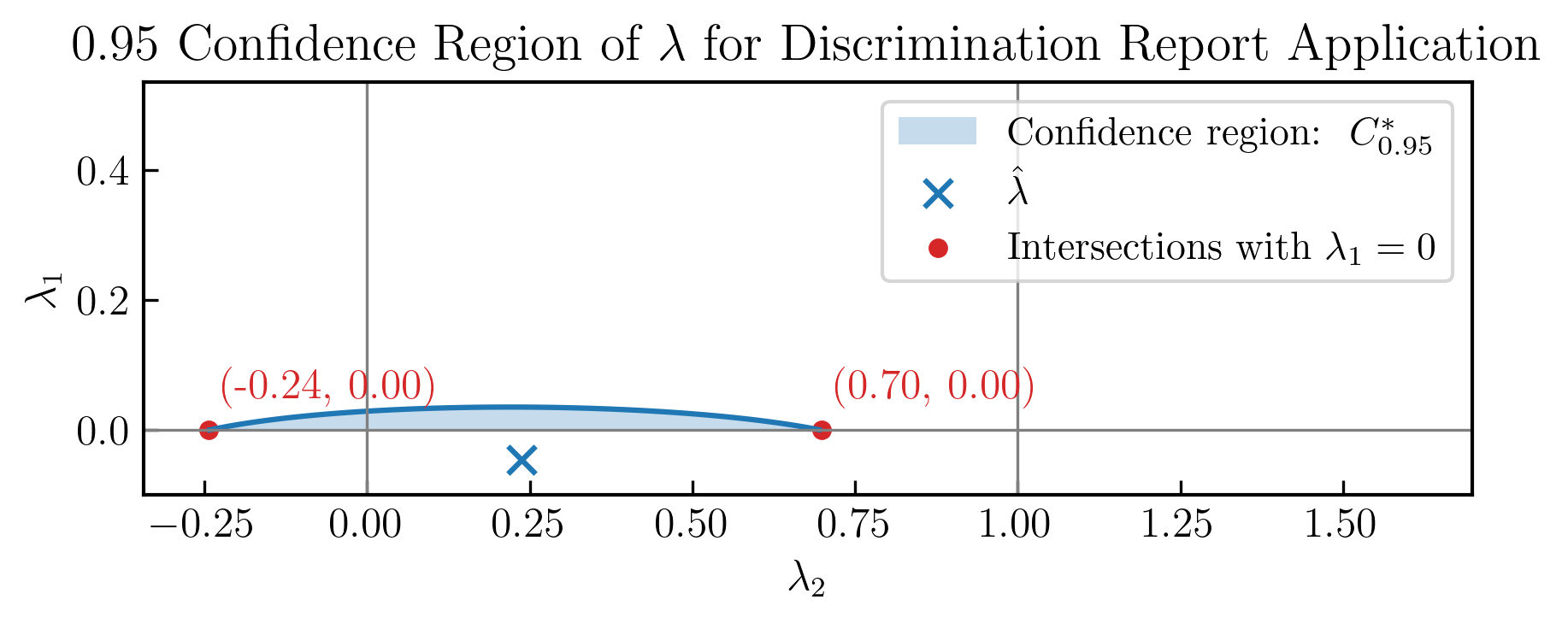}
        \caption{$95\%$ confidence region of $\boldsymbol{\lambda}$ for discrimination report application}
        \label{fig:confidence:discrimination}
    \end{minipage}
    \vfill
    \begin{minipage}{0.8\textwidth} 
        \centering
        \includegraphics[width=\linewidth]{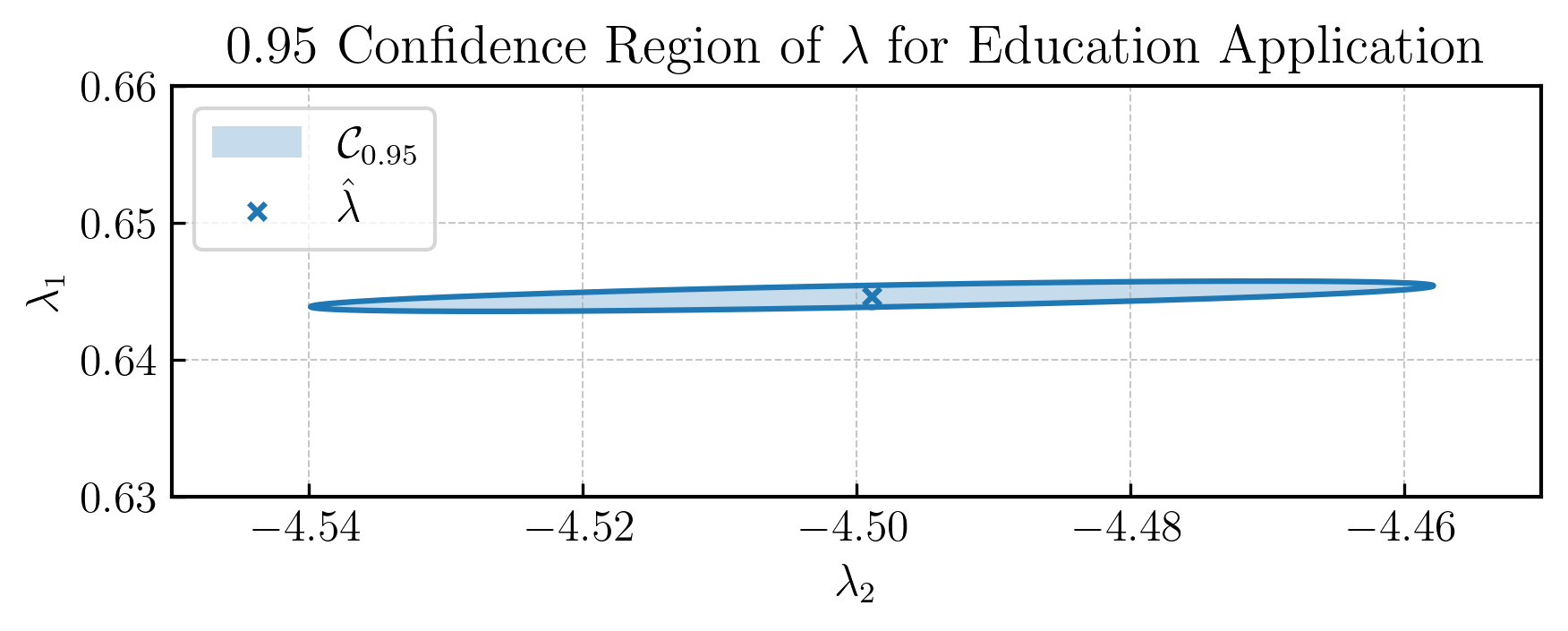}
        \caption{$95\%$ confidence region of $\boldsymbol{\lambda}$ for education application.}
        \label{fig:confidence:education}
    \end{minipage}
\end{figure}

\subsection{Validation Via Data Thinning}
To empirically validate the effectiveness of the binomial-shrinkage estimator, we plot both the one-sample and two-sample SURE surface, defined by \eqref{eq:SURE:one-sample} and \eqref{eq:L2:risk:estimator}, respectively, alongside the corresponding holdout \textit{mean squared error} (MSE). Specifically, we randomly split each dataset into training and holdout subsets. On the training set, we compute the SURE values across a range of $\boldsymbol{\lambda}$ values; on the holdout set, we evaluate the holdout MSE for our binomial-shrinkage estimators at these same $\boldsymbol{\lambda}$ values.

Note that the one-sample and two-sample SUREs are unbiased estimators for \eqref{eq:unweighted:one-sample} and \eqref{eq:unweighted:L2:risk}, respectively, up to constants from \eqref{eq:L2:one-sample} and \eqref{eq:L2:two-sample}. Consequently, if our method is valid, the SURE surface and holdout MSE surface should be parallel up to a constant. 

To approximate the true risk, we apply the standard data splitting to the employment discrimination application where the individual-level measurements are available. Specifically, for each job position $i$, we randomly select one white applicant and one black applicant as the holdout observation and leave the other three in each group as the training observations. For the other two applications, we only have access to the aggregate data. Thus, we apply the data-thinning procedure described in \ref{subsubsec:data-thinning} with $m_i=\lfloor 0.2n_i\rfloor$ for the application on reporting the innovation rates and $m_{i1}=\lfloor 0.2n_{i1}\rfloor, m_{i2}=\lfloor 0.2n_{i2}\rfloor$ for the application on estimating the test passing rate gaps.

The SURE surfaces based on the training data and the holdout MSE surface are plotted in Figure \ref{fig:sure_vs_holdout:innovation} - \ref{fig:sure_vs_holdout:education}. Clearly, the surfaces are closed to parallel, justifying the validity of our SURE approach.

\begin{figure}[htbp]
    \centering
    \begin{minipage}{0.48\textwidth} 
        \centering
        \includegraphics[width=\linewidth]{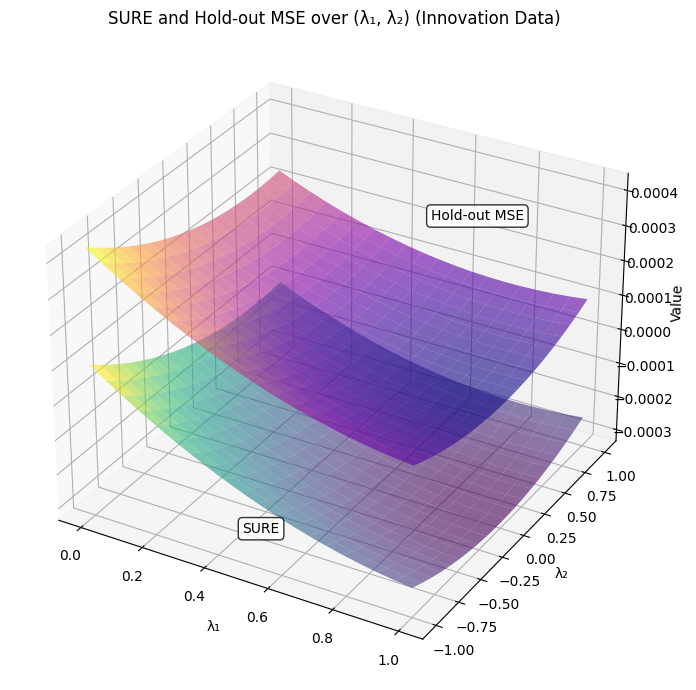}
        \caption{SURE versus holdout MSE for inventor fraction of innovation application.}
        \label{fig:sure_vs_holdout:innovation}
    \end{minipage}
    \hfill
    \begin{minipage}{0.48\textwidth} 
        \centering
        \includegraphics[width=\linewidth]{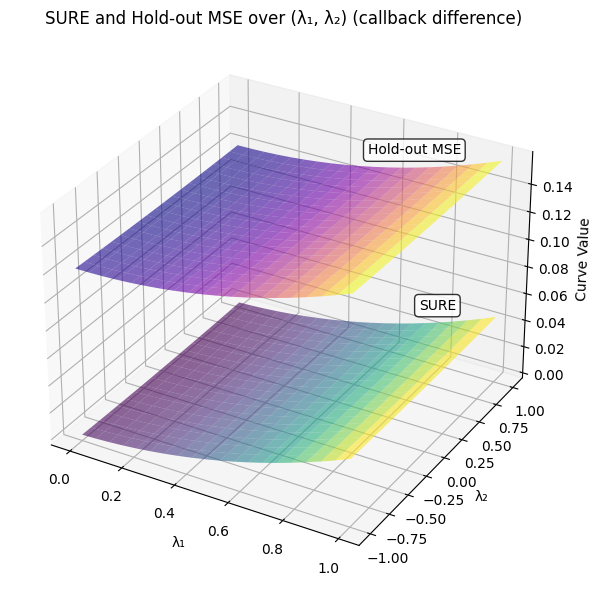}
        \caption{SURE versus holdout MSE for callback rate difference estimation of discrimination application.}
        \label{fig:sure_vs_holdout:discrimination}
    \end{minipage}
    \vfill
    \begin{minipage}{0.48\textwidth} 
        \centering
        \includegraphics[width=\linewidth]{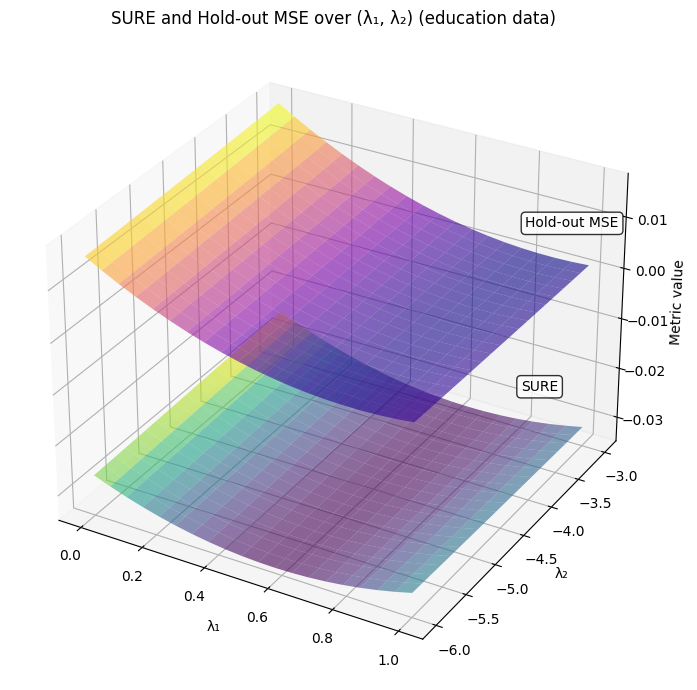}
        \caption{SURE versus holdout MSE for test performance estimation of K-12 schools of education application.}
        \label{fig:sure_vs_holdout:education}
    \end{minipage}
\end{figure}

\subsection{Comparison with Existing Methods}
We also compare our binomial-shrinkage estimators with existing methods proposed by \citet{xie2012sure} and \citet{chen2022empirical} for all three empirical applications. Both papers work with Gaussian measurements  $Z_i|\theta_i\sim\mathcal{N}(\theta_i,\sigma_i^2)$, where $Z_i$ are independent observations with known variances $\sigma_i^2$. Under this assumption of normality, \citet{xie2012sure} derive a Stein's unbiased risk estimator (SURE) for the squared-error loss and propose two classes of shrinkage estimators: one is a linear interpolation between the maximum likelihood estimator (MLE) and the grand mean; the other is a linear interpolation between the MLE and a data-driven location. Another recent work by \citet{chen2022empirical} highlights that Gaussian empirical Bayes methods commonly rely on a precision-independence assumption, which posits that the parameters of interest are independent of their known standard errors—an assumption that is often theoretically problematic and empirically rejected. Consequently, \citet{chen2022empirical} propose the CLOSE framework to estimate Gaussian means with known variances that accommodates precision dependence. 

For the method proposed by \citet{xie2012sure}, we construct three distinct estimators. The first, labeled \textit{SURE (grand mean)}, linearly interpolates between the MLE and the grand mean. The second estimator, labeled \textit{SURE (weighted ML mean)}, is a linear interpolation between the MLE and the weighted average of cross-fitted machine learning predictions, where each prediction is weighted by the total sample size $n_i$ of observation $i$. The third estimator, labeled \textit{SURE (data-driven location)}, shrinks toward a data-driven location computed following the procedure proposed by \citet{xie2012sure}. For the approach proposed by \citet{chen2022empirical}, we implement two estimators: \textit{CLOSE-NPMLE} and \textit{CLOSE-Gauss}. The CLOSE-NPMLE estimator employs a \textit{nonparametric maximum likelihood} (NPMLE) approach to estimate the prior distribution of $\theta_i|\sigma_i$, while CLOSE-Gauss assumes a standard Gaussian prior for the same distribution. We implemented the ``close'' R package to compute the posterior means using both of the CLOSE estimators. We compute posterior means under both specifications using the \texttt{close} R package \citep{closeRpackage}.

We first compare the holdout mean squared error (MSE) of all the aforementioned estimators with that of our binomial-shrinkage estimator across the three empirical applications. The methods proposed by \citet{xie2012sure} and \citet{chen2022empirical} assume known variances; hence, following common practice, we plug in the estimated variance $Y_i/n_i(1 - Y_i/n_i)/n_i$ for the one-sample case, where $n_i$ is the total sample size and $Y_i$ is the number of positive outcomes (value $1$) for observation $i$. Similarly, for the two-sample case, we use the variance estimator $Y_{i1}/n_{i1}(1-Y_{i1}/n_{i1})/n_{i1}+Y_{i2}/n_{i2}(1-Y_{i2}/n_{i2})/n_{i2}$, where $n_{i1}$ and $Y_{i1}$ represent the sample size and count of positive outcomes from population one for observation $i$, and similarly, $n_{i2}$ and $Y_{i2}$ correspond to population two. One important note is that the CLOSE method is not applicable to the employment discrimination application. Specifically, in the training set for this application, we have $n_{i1}=n_{i2}=3$, implying that $Y_{i1}, Y_{i2}\in\{0,1,2,3\}$. Consequently, there are only a few distinct values for the variance, making it unsuitable for the nonparametric estimation required by the CLOSE framework. 

Additionally, we also compute the weighted average of the estimates for all three applications. According to \eqref{eq:two-sample:estimator:lambda:one-sample}, the weighted average of the one-sample binomial-shrinkage estimator coincides with the grand mean, i.e. 
\begin{equation}\label{eq:grand-mean:exact:one-sample}
    \frac{\sum_{i=1}^Nn_i\hat{\theta}_i(\mathbf{Y}; \boldsymbol{\lambda})}{\sum_{i=1}^Nn_i}=\frac{\sum_{i=1}^NY_i}{\sum_{i=1}^Nn_i}.
\end{equation}
Further, note that the two-sample binomial-shrinkage estimator defined in equation~\eqref{eq:two-sample:estimator:lambda} can be equivalently expressed as the difference between two one-sample estimators applied separately to populations one and two, denoted as $\hat{\theta}_{i1}(\mathbf{Y};\boldsymbol{\lambda})$ and $\hat{\theta}_{i2}(\mathbf{Y};\boldsymbol{\lambda})$, respectively. Similarly, the difference between the weighted averages of these estimators, $\hat{\theta}_{i1}(\mathbf{Y};\boldsymbol{\lambda})$ and $\hat{\theta}_{i2}(\mathbf{Y};\boldsymbol{\lambda})$, corresponds exactly to the grand-mean difference between the two populations, i.e. 
\begin{equation}\label{eq:grand-mean:exact:two-sample}
    \frac{\sum_{i=1}^Nn_{i1}\hat{\theta}_{i1}(\mathbf{Y}; \boldsymbol{\lambda})}{\sum_{i=1}^Nn_{i1}}-\frac{\sum_{i=1}^Nn_{i2}\hat{\theta}_{i2}(\mathbf{Y}; \boldsymbol{\lambda})}{\sum_{i=1}^Nn_{i2}}=\frac{\sum_{i=1}^NY_{i1}}{\sum_{i=1}^Nn_{i1}}-\frac{\sum_{i=1}^NY_{i2}}{\sum_{i=1}^Nn_{i2}}. 
\end{equation}

\paragraph{Results. } Table~\ref{tab:comparison:MSE} presents the MSE values for all estimators. Our binomial-shrinkage estimator achieves the lowest MSE across all three applications. 
Table~\ref{tab:comparison:weighted:average} reports the weighted-average estimates from the one-sample innovation application, as well as the differences between the weighted-average estimates for two populations in the two-sample applications (discrimination report and education applications), computed over the entire dataset. Notably, the weighted-average estimates from all the considered estimators differ from the grand mean calculated from the full dataset. In contrast, our binomial-shrinkage estimator satisfies this property exactly, as shown in equations~\eqref{eq:grand-mean:exact:one-sample} and \eqref{eq:grand-mean:exact:two-sample}.

\begin{table}[htbp]
    \centering
    \caption{Holdout MSE Comparison}
    \label{tab:comparison:MSE}
    \begin{tabular}{l||c|c|c}
        \hline
        & Innovation & Employment Discrimination & Education \\
        \hline
        SURE (grand mean)
        &$2.8941\times10^{-5}$ & $0.1260$ & $0.0249$\\
        SURE (weighted ML mean) 
        &$2.8940\times10^{-5}$ & $0.1260$ & $0.0258$\\
        SURE (data-driven location) &$2.9406\times10^{-5}$ & $0.1257$ & $0.0258$\\
        CLOSE-NPMLE &$3.4379\times10^{-5}$  & not applicable & $0.0259$\\
        CLOSE-Gauss &$2.9051\times10^{-5}$  & not applicable &$0.0258$\\
        \textbf{Binomial-shrinkage estimator} 
        &$\mathbf{2.8666\times10^{-5}}$ & $\mathbf{0.1204}$ & $\mathbf{0.0246}$\\
        \hline
    \end{tabular}
\end{table}

\begin{table}[htbp]
    \centering
    \caption{Weighted Average Values Comparison}
    \label{tab:comparison:weighted:average}
    \begin{tabular}{l||c|c|c}
        \hline
        & Innovation & Discrimination Report & Education \\
        \hline
        Grand mean & $5.6212\times10^{-3}$ & $0.0840$ & $0.2620$ \\
        \hline
        SURE (grand mean)
        &$5.6184\times10^{-3}$ & $0.0124$ & $0.2572$\\
        SURE (weighted ML mean) 
        &$5.6182\times10^{-3}$ & $0.0108$ & $0.2284$\\
        SURE (data-driven location) & $5.6184\times10^{-3}$ & $0.0290$ & $0.2373$\\
        CLOSE-NPMLE & $5.6030\times10^{-3}$  & not applicable & $0.6032$\\
        CLOSE-Gauss & $5.5913\times10^{-3}$  & not applicable & $0.6029$\\
        \textbf{Binomial-shrinkage estimator} 
        & \textbf{exact} & \textbf{exact} & \textbf{exact}\\
        \hline
    \end{tabular}
\end{table}

\section{Conclusion and Extensions}
This paper develops a compound decision approach for estimating many binomial parameters by working directly with the exact binomial distribution. We construct an approximate Stein’s unbiased risk estimator for the average mean squared error that remains valid under heterogeneous sample sizes and small counts, without relying on Gaussian approximations. For a class of machine-learning–assisted linear shrinkage estimators, we establish asymptotic optimality, regret bounds relative to the oracle, and valid inference. The estimators satisfy a reporting-consistency property and are guaranteed to improve upon the maximum likelihood estimator and any single predictive model under mild conditions. More broadly, the framework offers a computationally tractable method for combining information pooling with covariate adjustment across large collections of binomial problems. Empirical applications to discrimination detection, education outcomes, and innovation rates show that the proposed estimators deliver stable and interpretable gains in accuracy in settings with heterogeneous sample sizes, small samples, or small binomial means.

For future work, we plan to incorporate Bayes estimators into the shrinkage framework and to extend the analysis to settings with selection decisions \citep{chen2025compound}, which will require a more sophisticated bias calculation for our risk estimator. In addition, if most $n_i$ are much larger than $2$, we can also add higher-order polynomials of $Y_i$ as covariates.

\clearpage
\bibliography{reference}

\clearpage
\appendix

\section{Additional Proofs for Section~\ref{sec:SURE}}

\subsection{Additional Proofs for Section~\ref{sec:SURE}}
\begin{proof}[Proof of Proposition \ref{prop:stein}]
  By definition,
  \begin{align*}
    \E[(n - Y)g(Y + 1)] & = \sum_{i=0}^{n-1} (n - i)g(i+1) \frac{n!}{i!(n - i)!}\theta^{i}(1 - \theta)^{n-i}\\
                        & = \sum_{i=0}^{n-1} g(i+1) \frac{n!}{i!(n - i - 1)!}\theta^{i}(1 - \theta)^{n-i}\\
                        & = \sum_{i=0}^{n-1} (i+1)g(i+1) \frac{n!}{(i + 1)!(n - i - 1)!}\theta^{i}(1 - \theta)^{n-i}\\
                        & = \frac{1 - \theta}{\theta}\sum_{j=1}^{n} jg(j) \frac{n!}{j!(n - j)!}\theta^{j}(1 - \theta)^{n-j}\\
                        & = \frac{1 - \theta}{\theta}\E[Yg(Y)].
  \end{align*}
\end{proof}

\noindent Recall that 
for any function $h$ on $\{0, \ldots, n\}$, 
$$\T_1 h(y; n) := \mathbf{1}(y > 0)\sum_{j=0}^{n-y}h(y + j)(-1)^{j}\frac{(n-y)!}{(n-y-j)!}\frac{y!}{(y+j)!},$$
$$\T_2 h(y; n) := h(y) - \mathbf{1}(y < n)\sum_{j=0}^{y}h(y - j)(-1)^{j}\frac{y!}{(y-j)!}\frac{(n - y)!}{(n-y+j)!},$$
$$\Delta h := \sum_{j=0}^{n}h(j)(-1)^{j}\com{n}{j}.$$

\begin{proof}[Proof of Theorem \ref{thm:better}]
Note that we can rewrite $\T_1 h$ and $\T_2 h$ as
\begin{align*}
  \T_1 h(y; n) &= (-1)^y \mathbf{1}(y > 0)\com{n}{y}^{-1}\sum_{k=y}^{n}h(k)(-1)^{k}\com{n}{k},
\end{align*}
and
\begin{align*}
  \T_2 h(y; n) &= h(y) - (-1)^y \mathbf{1}(y < n)\com{n}{y}^{-1}\sum_{k=0}^{y}h(k)(-1)^{k}\com{n}{k}\\
  & = h(n)\mathbf{1}(y = n) - (-1)^y \mathbf{1}(y < n)\com{n}{y}^{-1}\sum_{k=0}^{y-1}h(k)(-1)^{k}\com{n}{k}
\end{align*}
Recall the definition of $\Delta h$, we have
\[\T_1 h(y; n) - \T_2 h(y; n) = \left\{
    \begin{array}{ll}
      0 & y = 0 \text{ or }n\\
      (-1)^y \com{n}{y}^{-1} (\Delta h)& 0 < y < n
    \end{array}
  \right.\]
Then
\begin{align*}
  &\E[(\T_1 h(Y; n) - \T_2 h(Y; n))\mathbf{1}(Y > a)] = (\Delta h)\sum_{y=a+1}^{n-1} (-1)^y \theta^y (1 - \theta)^{n-y}.
\end{align*}
As a result, 
\begin{align*}
  &\E[\T h(Y; n, a)] - \theta \E[h(Y)]\\
  & = \E[\T_2 h(Y; n)] - \theta \E[h(Y)] + \E[(\T_1 h(Y; n) - \T_2 h(Y; n))\mathbf{1}(Y > a)]\\
  & = - (1 - \theta) \theta^{n}(-1)^{n+1}(\Delta h) + \E[(\T_1 h(Y; n) - \T_2 h(Y; n))\mathbf{1}(Y > a)]\\
  & = (\Delta h)\left\{- (1 - \theta) \theta^{n}(-1)^{n+1} + \sum_{y=a+1}^{n-1} (-1)^y \theta^y (1 - \theta)^{n-y}\right\}\\
  & = (\Delta h)\left\{- (1 - \theta) \theta^{n}(-1)^{n+1} + (1 - \theta)^{n}\sum_{y=a+1}^{n-1} \lb \frac{-\theta}{1-\theta}\rb^{y}\right\}\\
  & = (\Delta h)\left\{- (1 - \theta) \theta^{n}(-1)^{n+1} + (1 - \theta)^{n}\cdot \lb \frac{-\theta}{1-\theta}\rb^{a+1} \frac{1 - \lb \frac{-\theta}{1-\theta}\rb^{n-a-1}}{1 + \frac{\theta}{1- \theta}}\right\}\\
  & = (\Delta h)(-1)^{a+1}\theta^{a+1} (1 - \theta)^{n-a}.
\end{align*}
Hence \eqref{eq:bias:robust} follows. 

Next, to show statement (i) of Theorem~\ref{thm:better}, note that \eqref{eq:bias:robust} implies that 
\begin{equation}\label{eq:bias:robust:abs}
\left|\theta \E[h(Y)] - \E[\T h(Y; n, \lfloor n/2\rfloor)]\right|= \left|\theta^{\lfloor n/2\rfloor+1} (1 - \theta)^{n-\lfloor n/2\rfloor} (\Delta h)\right|.
\end{equation}
On the one hand, when $n=2k$ where $k\geq1$, then 
$$\left|\theta^{\lfloor n/2\rfloor+1} (1 - \theta)^{n-\lfloor n/2\rfloor} (\Delta h)\right|=|\theta^{k+1}(1-\theta)^k(\Delta h)|\leq{4}^{-k}|\Delta h|=2^{-n}|\Delta h|,$$
where we use the fact that $\theta(1-\theta)\leq1/4$ in the inequality above. On the other hand, when $n=2k-1$ where $k\geq1$, then $\lfloor n/2\rfloor=k-1$, and 
$$\left|\theta^{\lfloor n/2\rfloor+1} (1 - \theta)^{n-\lfloor n/2\rfloor} (\Delta h)\right|=|\theta^{k}(1-\theta)^k(\Delta h)|\leq{4}^{-k}|\Delta h|=2^{-n}|\Delta h|.$$

Lastly, to show statement (ii) of Theorem~\ref{thm:better}, note that for any $x\in [0, 1]$,
  \[(1 - x)^{n} = \sum_{j=0}^{n}x^j(-1)^{j}\com{n}{j}.\]
  For any $r \in \{0, \ldots, n-1\}$, taking $r$-th derivatives with respect to $x$ yields that
  \[\frac{n!}{(n-r)!}(1 - x)^{n-r} = \sum_{j=r}^{n}\frac{j!}{(j-r)!}x^{j-r}(-1)^{j}\com{n}{j}.\]
  Letting $x = 1$, we obtain that for any $r\in\{0,\ldots,n-1\}$,
  \[\sum_{j=r}^{n}\frac{j!}{(j-r)!}(-1)^{j}\com{n}{j}=0.\]
  Clearly, given any polynomial $h$ of degree less than $n$, there exist real numbers $w_r$ for $r\in\{0,1,\ldots,n-1\}$, such that for every integer $j\in\{0,1,\ldots,n\}$,
  $$h(j) = \sum_{r=0}^{n-1} w_r\frac{j!}{(j-r)!}.$$
  Consequently,
  $$\Delta h = \sum_{j=0}^{n}h(j)(-1)^{j}\com{n}{j}=\sum_{r=0}^{n-1} w_r\sum_{j=r}^{n}\frac{j!}{(j-r)!}(-1)^{j}\com{n}{j}=0.$$
  Therefore, $\Delta h = 0$ when $h$ is a polynomial of degree less than $n$. 
\end{proof}

\begin{proof}[Proof of Proposition~\ref{prop:approximate:sure}]
For any $i\in[N]$, recall that $k(i)$ is the fold that the $i$-th sample belongs to. Given any $k\neq k(i)$, for any $j\in\mathcal{I}_{k}$, define $\hat{g}_i^{-k}(\cdot)$ as the ML estimator for $g(\cdot)$ computed on the dataset excluding fold $k$ and replacing $Y_i$ with its independent copy $Y_i^{(2)}$. Define 
$$\bar{Y}_{-i}:=\frac{\sum_{j=1,j\neq i}^NY_j}{\sum_{j=1}^Nn_j}.$$  
Define 
$$\begin{array}{rl}
\tilde\theta_i^{\mathrm{o}}(\boldsymbol{\lambda})\!\!\!\!&\displaystyle:=\lambda_1\frac{Y_i}{n_i} + (1-\lambda_1)\bar{Y}_{-i}\\
&\displaystyle\quad+\lambda_2\left(\hat{g}^{-k(i)}(\mathbf{X}_i)-\frac{\sum_{j\in\mathcal{I}_{k(i)}}n_j\hat{g}^{-k(i)}(\mathbf{X}_j)+\sum_{k\neq k(i)}^K\sum_{j\in\mathcal{I}_k}n_j\hat{g}_{i}^{-k}(\mathbf{X}_{j})}{\sum_{j=1}^Nn_j}\right).
\end{array}$$
Thus by definition
\begin{equation}\label{eq:diff:one-sample:estimators}
\begin{array}{rl}
\hat{\theta}_i^{\mathrm{o}}(\boldsymbol{\lambda})- \tilde{\theta}_i^{\mathrm{o}}(\boldsymbol{\lambda})\!\!\!\!&\displaystyle=\frac{(1-\lambda_1)Y_i}{\sum_{j=1}^Nn_j}-\frac{\lambda_2\sum_{k\neq k(i)}^K\sum_{j\in\mathcal{I}_k}n_j\{\hat{g}^{-k}(\mathbf{X}_j)-\hat{g}_{i}^{-k}(\mathbf{X}_j)\}}{\sum_{j=1}^Nn_j}
\end{array}
\end{equation}
Define the term on the right hand side of \eqref{eq:diff:one-sample:estimators} as 
\begin{equation}\label{eq:Delta:lambda:one-sample}
\Delta_i^{\mathrm{o}}(\boldsymbol{\lambda}):=\frac{(1-\lambda_1)Y_i}{\sum_{j=1}^Nn_j}-\frac{\lambda_2\sum_{k\neq k(i)}^K\sum_{j\in\mathcal{I}_k}n_j\{\hat{g}^{-k}(\mathbf{X}_j)-\hat{g}_{i}^{-k}(\mathbf{X}_j)\}}{\sum_{j=1}^Nn_j}.
\end{equation}
So 
\begin{equation}\label{eq:theta:hat:lambda}
\hat{\theta}_i^{\mathrm{o}}(\boldsymbol{\lambda})=\tilde{\theta}_i^{\mathrm{o}}(\boldsymbol{\lambda})+\Delta_i^{\mathrm{o}}(\boldsymbol{\lambda}).
\end{equation}
For any $i\in[N]$, define
$$\mathbf{Y}_{-i}^{\mathrm{o}}:=\{Y_1,\ldots,Y_{i-1},Y_{i+1},\ldots,Y_N\},$$ 
so $\mathbf{Y}_{-i}^{\mathrm{o}}$ is the collection of one-sample observations excluding the $i$-th observation. Note that $Y_i\indep\hat{g}^{-k(i)}(\mathbf{X}_j)$ and $Y_i\indep\hat{g}_{i}^{-k}(\mathbf{X}_j)$ for any $k\neq k(i)$. Conditioning on $\mathbf{Y}_{-i}^{\mathrm{o}}$, $\tilde{\theta}_i^{\mathrm{o}}(\boldsymbol{\lambda})$ is affine in $Y_i$. So when $n_i\geq2$, Theorem~\ref{thm:better} implies that 
\begin{equation}\label{eq:conditional:unbiased:one-sample}
\mathbb{E}\left[\mathcal{T}\tilde{\theta}_i^{\mathrm{o}}(\boldsymbol{\lambda})|\mathbf{Y}_{-i}^{\mathrm{o}}\right]=\theta_i^{\mathrm{o}}\mathbb{E}\left[\tilde{\theta}_i^{\mathrm{o}}(\boldsymbol{\lambda})|\mathbf{Y}_{-i}^{\mathrm{o}}\right].
\end{equation}
Furthermore, using definition of $\mathcal{T}_1$ \eqref{eq:T1}, and recall that when we apply functional $\mathcal{T}$ to $\hat{\theta}_i^{\mathrm{o}}(\boldsymbol{\lambda})$, we fix the grand mean $(\sum_{i=1}^NY_i)/(\sum_{i=1}^Nn_i)$ and the ML model outputs $\hat{g}^{-k(j)}(\mathbf{X}_j)$, $\forall j\in[N]$. Thus 
\begin{equation}\label{eq:one-sample:T1}
\begin{array}{rl}
\mathcal{T}_1\hat{\theta}_i^{\mathrm{o}}(\boldsymbol{\lambda})\!\!\!\!&\displaystyle=\mathcal{T}_1\tilde\theta_i^{\mathrm{o}}(\boldsymbol{\lambda})+\Delta_i^{\mathrm{o}}(\boldsymbol{\lambda})\cdot\mathbf{1}\{Y_i>0\}\sum_{j=0}^{n_i-Y_i}(-1)^{j}\frac{(n_i-Y_i)!}{(n_i-Y_i-j)!}\frac{Y_i!}{(Y_i+j)!}\\
&\displaystyle=_{(1)}\mathcal{T}_1\tilde\theta_i^{\mathrm{o}}(\boldsymbol{\lambda})+\mathbf{1}\{Y_i>0\}\frac{Y_i}{n_i}\cdot\Delta_i^{\mathrm{o}}(\boldsymbol{\lambda}),
\end{array}
\end{equation}
where equality (1) of \eqref{eq:one-sample:T1} follows directly from Lemma~\ref{lemma:combinatorics}. Similarly, using definition of $\mathcal{T}_2$ \eqref{eq:T2}, we have 
\begin{equation}\label{eq:one-sample:T2}
\begin{array}{rl}
\mathcal{T}_2\hat{\theta}_i^{\mathrm{o}}(\boldsymbol{\lambda})\!\!\!\!&\displaystyle=\mathcal{T}_2\tilde\theta_i^{\mathrm{o}}(\boldsymbol{\lambda})+\Delta_i^{\mathrm{o}}(\boldsymbol{\lambda})-\Delta_i^{\mathrm{o}}(\boldsymbol{\lambda})\cdot\mathbf{1}(Y_i < n_i)\sum_{j=0}^{Y_i}(-1)^{j}\frac{Y_i!}{(Y_i-j)!}\frac{(n_i - Y_i)!}{(n_i-Y_i+j)!}\\
&\displaystyle=_{(2)}\mathcal{T}_2\tilde\theta_i^{\mathrm{o}}(\boldsymbol{\lambda})+\Delta_i^{\mathrm{o}}(\boldsymbol{\lambda})\mathbf{1}\{Y_i < n_i\}\left[1-\frac{n_i-Y_i}{n_i}\right]\\
&\displaystyle=\mathcal{T}_2\tilde\theta_i^{\mathrm{o}}(\boldsymbol{\lambda})+\mathbf{1}\{Y_i < n_i\}\frac{Y_i}{n_i}\cdot\Delta_i^{\mathrm{o}}(\boldsymbol{\lambda}).
\end{array}
\end{equation}
So following from \eqref{eq:T}, 
\begin{equation}\label{eq:one-sample:T}
\begin{array}{rl}
    \mathcal{T}\hat{\theta}_i^{\mathrm{o}}(\boldsymbol{\lambda})&\displaystyle=\mathcal{T}\tilde{\theta}_i^{\mathrm{o}}(\boldsymbol{\lambda})+\mathbf{1}\{Y_i>\lfloor n_i/2\rfloor\}\frac{Y_i}{n_i}\cdot\Delta_i^{\mathrm{o}}(\boldsymbol{\lambda})+\mathbf{1}\{Y_i\leq\lfloor n_i/2\rfloor\}\frac{Y_i}{n_i}\cdot\Delta_i^{\mathrm{o}}(\boldsymbol{\lambda})\\
    &\displaystyle=\mathcal{T}\tilde{\theta}_i^{\mathrm{o}}(\boldsymbol{\lambda})+\frac{Y_i}{n_i}\cdot\Delta_i^{\mathrm{o}}(\boldsymbol{\lambda}).
\end{array}
\end{equation}
Then \eqref{eq:one-sample:T} implies that
\begin{equation}\label{eq:one-sample:T:theta:hat}
\begin{array}{rl}
\mathbb{E}\left[\mathcal{T}\hat{\theta}_i^{\mathrm{o}}(\boldsymbol{\lambda})\big|\mathbf{Y}_{-i}^{\mathrm{o}}\right]\!\!\!\!&\displaystyle=\mathbb{E}\left[\mathcal{T}\tilde{\theta}_i^{\mathrm{o}}(\boldsymbol{\lambda})\big|\mathbf{Y}_{-i}^{\mathrm{o}}\right]+\mathbb{E}\left[\frac{Y_i}{n_i}\cdot\Delta_i^{\mathrm{o}}(\boldsymbol{\lambda})\big|\mathbf{Y}_{-i}^{\mathrm{o}}\right]\\
&\displaystyle=_{(1)}\theta_i^{\mathrm{o}}\mathbb{E}\left[\tilde{\theta}_i^{\mathrm{o}}(\boldsymbol{\lambda})|\mathbf{Y}_{-i}^{\mathrm{o}}\right]+\mathbb{E}\left[\frac{Y_i}{n_i}\cdot\Delta_i^{\mathrm{o}}(\boldsymbol{\lambda})\big|\mathbf{Y}_{-i}^{\mathrm{o}}\right]\\
&\displaystyle=_{(2)}\theta_i^{\mathrm{o}}\mathbb{E}\left[\hat{\theta}_i^{\mathrm{o}}(\boldsymbol{\lambda})-\Delta_i^{\mathrm{o}}(\boldsymbol{\lambda})\big|\mathbf{Y}_{-i}^{\mathrm{o}}\right]+\mathbb{E}\left[\frac{Y_i}{n_i}\cdot\Delta_i^{\mathrm{o}}(\boldsymbol{\lambda})\big|\mathbf{Y}_{-i}^{\mathrm{o}}\right]\\
&\displaystyle=\theta_i^{\mathrm{o}}\mathbb{E}\left[\hat{\theta}_i^{\mathrm{o}}(\boldsymbol{\lambda})\big|\mathbf{Y}_{-i}^{\mathrm{o}}\right]+\mathbb{E}\left[\left\{\frac{Y_i}{n_i}-\theta_i^{\mathrm{o}}\right\}\Delta_i^{\mathrm{o}}(\boldsymbol{\lambda})\big|\mathbf{Y}_{-i}^{\mathrm{o}}\right],
\end{array}
\end{equation}
where in \eqref{eq:one-sample:T:theta:hat}, (1) follows from \eqref{eq:conditional:unbiased:one-sample}, (2) follows from \eqref{eq:theta:hat:lambda}. Hence, using the fact that $|Y_i/n_i-\theta_i|\in[0,1]$, we have 
\begin{equation}\label{eq:expectation:diff:one-sample}
    \left|\mathbb{E}\left[\mathcal{T}\hat{\theta}_i^{\mathrm{o}}(\boldsymbol{\lambda})\big|\mathbf{Y}_{-i}^{\mathrm{o}}\right]-\theta_i^{\mathrm{o}}\mathbb{E}\left[\hat{\theta}_i^{\mathrm{o}}(\boldsymbol{\lambda})\big|\mathbf{Y}_{-i}^{\mathrm{o}}\right]\right|\leq\mathbb{E}\left[\left|\Delta_i^{\mathrm{o}}(\boldsymbol{\lambda})\right|\big|\mathbf{Y}_{-i}^{\mathrm{o}}\right]
\end{equation}
We now bound the right hand side of \eqref{eq:expectation:diff:one-sample}.
Since $2\leq n_j\leq\bar{n}$ for any $j\in[N]$, 
$$\left|\frac{(1-\lambda_1)Y_i}{\sum_{j=1}^Nn_j}\right|\leq\frac
{\bar{n}}{2N}.$$
Note that $\hat{g}^{-k}$ and $\hat{g}_i^{-k}$ are both between $0$ and $1$, and Assumption \ref{assump:one_sample} (b) implies that given any $i\in[N]$, for any $k\neq k(i)$, $\max_{j\in\mathcal{I}_k}|\hat{g}^{-k}(\mathbf{X}_j)-g(\mathbf{X}_j)|=\mathrm{o}_p(1)$, so using dominated convergence theorem, we have $\max_{j\in\mathcal{I}_k}\mathbb{E}\left[|\hat{g}^{-k}(\mathbf{X}_j)-g(\mathbf{X}_j)|\right]=\mathrm{o}(1)$. Since $\hat{g}_i^{-k}$ is computed by replacing $Y_i$ with its independent copy $Y_i^{(2)}$ and $\mathbf{X}_j$ are treated as fixed, so by definition, for any $k\neq k(i)$, we have  
$$\mathbb{E}\left[|\hat{g}_i^{-k}(\mathbf{X}_j)-g(\mathbf{X}_j)|\right]=\mathbb{E}\left[|\hat{g}^{-k}(\mathbf{X}_j)-g(\mathbf{X}_j)|\right],$$
implying that $\max_{j\in\mathcal{I}_k}\mathbb{E}[|\hat{g}_i^{-k}(\mathbf{X}_j)-g(\mathbf{X}_j)|]=\mathrm{o}(1).$
Thus by triangular inequality,
$$\begin{array}{rl}
&\displaystyle\quad\max_{j\in\mathcal{I}_k, k\neq k(i)}\mathbb{E}\left[|\hat{g}^{-k}(\mathbf{X}_j)-\hat{g}_{i}^{-k}(\mathbf{X}_j)|\right]\\
&\displaystyle\leq\max_{j\in\mathcal{I}_k, k\neq k(i)}\mathbb{E}\left[|\hat{g}^{-k}(\mathbf{X}_j)-g(\mathbf{X}_j)|\right]+\max_{j\in\mathcal{I}_k, k\neq k(i)}\mathbb{E}\left[|\hat{g}_i^{-k}(\mathbf{X}_j)-g(\mathbf{X}_j)|\right]\\
&\displaystyle=o(1).
\end{array}$$ 
So we have 
$$\begin{array}{rl}
&\displaystyle\quad\mathbb{E}\left[\left|\frac{\lambda_2\sum_{k\neq k(i)}^K\sum_{j\in\mathcal{I}_k}n_j\{\hat{g}^{-k}(\mathbf{X}_j)-\hat{g}_{i}^{-k}(\mathbf{X}_j)\}}{\sum_{j=1}^Nn_j}\right|\right]\\
&\displaystyle\leq\frac{|\lambda_2|\sum_{k\neq k(i)}^K\sum_{j\in\mathcal{I}_k}n_j\mathbb{E}[|\hat{g}^{-k}(\mathbf{X}_j)-\hat{g}_{i}^{-k}(\mathbf{X}_j)|]}{\sum_{j=1}^Nn_j}=\mathrm{o}(1)
\end{array}$$
Hence, \eqref{eq:Delta:lambda:one-sample} implies that 
\begin{equation}\label{eq:Delta:one-sample:bound}
\mathbb{E}[|\Delta_i^{\mathrm{o}}(\boldsymbol{\lambda})|]\leq\frac{\bar{n}}{2N}+\mathrm{o}(1).
\end{equation}
Thus \eqref{eq:expectation:diff:one-sample} and \eqref{eq:Delta:one-sample:bound} imply that 
\begin{equation}\label{eq:diff:expectation:abs}
\begin{array}{rl}
    \left|\mathbb{E}\left[\mathcal{T}\hat{\theta}_i^{\mathrm{o}}(\boldsymbol{\lambda})\right]-\theta_i^{\mathrm{o}}\mathbb{E}\left[\hat{\theta}_i^{\mathrm{o}}(\boldsymbol{\lambda})\right]\right|\!\!\!\!\!&\displaystyle=\left|\mathbb{E}\left\{\mathbb{E}\left[\mathcal{T}\hat{\theta}_i^{\mathrm{o}}(\boldsymbol{\lambda})\big|\mathbf{Y}_{-i}^{\mathrm{o}}\right]-\theta_i^{\mathrm{o}}\mathbb{E}\left[\hat{\theta}_i^{\mathrm{o}}(\boldsymbol{\lambda})\big|\mathbf{Y}_{-i}^{\mathrm{o}}\right]\right\}\right|\\
    &\displaystyle\leq\mathbb{E}\left\{\left|\mathbb{E}\left[\mathcal{T}\hat{\theta}_i^{\mathrm{o}}(\boldsymbol{\lambda})\big|\mathbf{Y}_{-i}^{\mathrm{o}}\right]-\theta_i^{\mathrm{o}}\mathbb{E}\left[\hat{\theta}_i^{\mathrm{o}}(\boldsymbol{\lambda})\big|\mathbf{Y}_{-i}^{\mathrm{o}}\right]\right|\right\}\\
    &\displaystyle\leq_{(1)}\mathbb{E}\left[\mathbb{E}\left[\left|\Delta_i^{\mathrm{o}}(\boldsymbol{\lambda})\right|\big|\mathbf{Y}_{-i}^{\mathrm{o}}\right]\right]\\
    &\displaystyle=\mathbb{E}\left[\Delta_i^{\mathrm{o}}(\boldsymbol{\lambda})\right]\leq_{(2)}\frac{\bar{n}}{2N}+\mathrm{o}(1),
\end{array}
\end{equation}
where in \eqref{eq:diff:expectation:abs} (1) follows from \eqref{eq:expectation:diff:one-sample} and (2) follows from \eqref{eq:Delta:one-sample:bound}. The upper bound in \eqref{eq:diff:expectation:abs} immediately implies that 
$$\left|\mathbb{E}\left[\hat{L}_{\mathrm{o}}(\boldsymbol{\lambda})\right]-L_{\mathrm{o}}(\boldsymbol{\lambda})\right|\leq\frac{\bar{n}}{N}+\mathrm{o}(1).$$
For the two-sample estimator, note that $\hat{\theta}_i^{\mathrm{t}}(\boldsymbol{\lambda})=\hat{\theta}_{i1}^{\mathrm{o}}(\boldsymbol{\lambda})-\hat{\theta}_{i2}^{\mathrm{o}}(\boldsymbol{\lambda})$. Define
$$\mathbf{Y}_{-i}^{\mathrm{t}}:=\{Y_{i\ell},\ldots,Y_{i-1,\ell},Y_{i+1,\ell},\ldots,Y_{N\ell}\}_{\ell\in\{1,2\}},$$
so $\mathbf{Y}_{-i}^{\mathrm{t}}$ is the two-sample observations excluding the $i$-th pair of observations $\{Y_{i1},Y_{i2}\}$. Following similar proof steps as for the one-sample case, we can show that there exists some absolute constant $\bar{C}$, such that almost surely we have
$$\left|\mathbb{E}\left[\mathcal{T}\hat{\theta}_i^{\mathrm{t}}(Y_{i1};n_{i1}|\boldsymbol{\lambda})|\mathbf{Y}_{-i}^{\mathrm{t}},Y_{i2}\right]-\theta_{i1}^{\mathrm{t}}\mathbb{E}\left[\hat{\theta}_i^{\mathrm{t}}(\boldsymbol{\lambda})|\mathbf{Y}_{-i}^{\mathrm{t}},Y_{i2}\right]\right|\leq\bar{n}/N+\mathrm{o}(1),$$
and 
$$\left|\mathbb{E}\left[\mathcal{T}\hat{\theta}_i^{\mathrm{t}}(Y_{i2};n_{i2}|\boldsymbol{\lambda})|\mathbf{Y}_{-i}^{\mathrm{t}},Y_{i1}\right]-\theta_{i2}^{\mathrm{t}}\mathbb{E}\left[\hat{\theta}_i^{\mathrm{t}}(\boldsymbol{\lambda})|\mathbf{Y}_{-i}^{\mathrm{t}},Y_{i1}\right]\right|\leq\bar{n}/N+\mathrm{o}(1),$$
Hence by definition almost surely we have 
$$\left|\mathbb{E}\left[\hat{L}_{\mathrm{t}}(\boldsymbol{\lambda})\right]-L_{\mathrm{t}}(\boldsymbol{\lambda})\right|\leq4\bar{n}/N+\mathrm{o}(1).$$
Hence we have proved the results.
\end{proof}

\subsection{Explicit Expressions for Approximate SUREs}\label{appendix:explicit:SURE:approx}
\subsubsection{One-Sample Approximate SURE}
Denote
\begin{equation}
  \label{eq:ell_iu_1}
  \hat{\ell}_{i, \u}^{(1)} := \hat{\theta}_i^{\mathrm{o}}(Y_i; n_i)^2 - 2\mathbf{1}(Y_i > 0)\sum_{j=0}^{n_i-Y_i}\hat{\theta}_i^{\mathrm{o}}(Y_i + j; n_i)(-1)^{j}\frac{(n_i-Y_i)!}{(n_i-Y_i-j)!}\frac{Y_i!}{(Y_i+j)!},
\end{equation}
and
\begin{equation}
  \label{eq:ell_iu_2}
  \hat{\ell}_{i, \u}^{(2)} := \hat{\theta}_i^{\mathrm{o}}(Y_i; n_i)^2 - 2 \hat{\theta}_i^{\mathrm{o}}(Y_i; n_i) + 2\mathbf{1}(Y_i < n_i)\sum_{j=0}^{Y_i}\hat{\theta}_i^{\mathrm{o}}(Y_i-j; n_i)(-1)^{j}\frac{Y_i!}{(Y_i-j)!}\frac{(n_i - Y_i)!}{(n_i-Y_i+j)!}.
\end{equation}
Expanding \eqref{eq:one-sample:SURE:lambda}, the explicit expression for the one-sample approximate SURE is 
\begin{equation}\label{eq:SURE:one-sample}
    \hat{L}_{\mathrm{o}}(\boldsymbol{\lambda})=\frac{1}{N}\sum_{i=1}^{N}\hat{\theta}_i^{\mathrm{o}}(Y_{i} ; n_i)^2-\frac{1}{N}\sum_{i=1}^{N}\hat{\ell}_{i, \u}^{(1)}\mathbf{1}\{Y_{i}>\lfloor n_{i}/2\rfloor\}+\frac{1}{N}\sum_{i=1}^{N}\hat{\ell}_{i, \u}^{(2)}\mathbf{1}\{Y_{i}\leq\lfloor n_{i}/2\rfloor\}.
\end{equation}

\subsubsection{Two-Sample Approximate SURE}
Define $\hat{\ell}_{i, \u}^{1,(1)}$ and $\hat{\ell}_{i, \u}^{1,(2)}$ as
\begin{equation}
  \label{eq:ell_iu_1:sample1}
  \hat{\ell}_{i, \u}^{1,(1)}:=2\mathbf{1}(Y_{i1} > 0)\!\!\!\sum_{j=0}^{n_{i1}-Y_{i1}}\!\!\!\hat{\theta}_i^{\mathrm{t}}(Y_{i1} + j; n_{i1}|\boldsymbol{\lambda})(-1)^{j}\frac{(n_{i1}-Y_{i1})!}{(n_{i1}-Y_{i1}-j)!}\frac{Y_{i1}!}{(Y_{i1}+j)!},
\end{equation}
\begin{equation}
  \label{eq:ell_iu_2:sample1}
  \begin{array}{rl}
  \hat{\ell}_{i, \u}^{1,(2)} \!\!\!\!\!&\displaystyle := 2\hat{\theta}_i^{\mathrm{t}}(Y_{i1}; n_{i1}|\boldsymbol{\lambda})\\
  &\quad \displaystyle -2\mathbf{1}(Y_{i1} < {n_{i1}})\sum_{j=0}^{Y_{i1}}\hat{\theta}_i^{\mathrm{t}}(Y_{i1}-j; n_{i1}|\boldsymbol{\lambda})(-1)^{j}\frac{Y_{i1}!}{(Y_{i1}-j)!}\frac{(n_{i1} - Y_{i1})!}{(n_{i1}-Y_{i1}+j)!}.
\end{array}
\end{equation}
Define $\hat{\ell}_{i, \u}^{2,(1)}$ and $\hat{\ell}_{i, \u}^{2,(2)}$ as
\begin{equation}
  \label{eq:ell_iu_1:sample2}
  \hat{\ell}_{i, \u}^{2,(1)}:=2\mathbf{1}(Y_{i2} > 0)\!\!\!\sum_{j=0}^{n_{i2}-Y_{i2}}\!\!\hat{\theta}_i^{\mathrm{t}}(Y_{i2} + j; n_{i2}|\boldsymbol{\lambda})(-1)^{j}\frac{(n_{i2}-Y_{i2})!}{(n_{i2}-Y_{i2}-j)!}\frac{Y_{i2}!}{(Y_{i2}+j)!},
\end{equation}
\begin{equation}
  \label{eq:ell_iu_2:sample2}
  \begin{array}{rl}
  \hat{\ell}_{i, \u}^{2,(2)} \!\!\!\!\!&\displaystyle:= 2 \hat{\theta}_i^{\mathrm{t}}(Y_{i2}; n_{i2}|\boldsymbol{\lambda})\\
  &\quad \displaystyle -2\mathbf{1}(Y_{i2} < {n_{i2}})\sum_{j=0}^{Y_{i2}}\hat{\theta}_i^{\mathrm{t}}(Y_{i2}-j; n_{i2}|\boldsymbol{\lambda})(-1)^{j}\frac{Y_{i2}!}{(Y_{i2}-j)!}\frac{(n_{i2} - Y_{i2})!}{(n_{i2}-Y_{i2}+j)!}. 
  \end{array}
\end{equation}
Expanding \eqref{eq:two-sample:SURE:lambda}, the explicit expression for the two-sample approximate SURE is 
\begin{equation}\label{eq:L2:risk:estimator}
\begin{array}{rl}
    \hat{L}_{\mathrm{t}}(\boldsymbol{\lambda})\!\!\!\!&=\displaystyle\frac{1}{N}\sum_{i=1}^{N}\hat{\theta}_i^{\mathrm{t}}(\boldsymbol{\lambda})^2-\frac{1}{N}\left[\sum_{i=1}^{N}\hat{\ell}_{i, \u}^{1,(1)}\mathbf{1}\{Y_{i1}>\lfloor n_{i1}/2\rfloor\}+\sum_{i=1}^{N}\hat{\ell}_{i, \u}^{1,(2)}\mathbf{1}\{Y_{i1}\leq\lfloor n_{i1}/2\rfloor\}\right]\\   
    \\
    &\displaystyle\quad\quad\quad +\frac{1}{N}\left[\sum_{i=1}^{N}\hat{\ell}_{i, \u}^{2,(1)}\mathbf{1}\{Y_{i2}>\lfloor n_{i2}/2\rfloor\}+\sum_{i=1}^{N}\hat{\ell}_{i, \u}^{2,(2)}\mathbf{1}\{Y_{i2}\leq\lfloor n_{i2}/2\rfloor\}\right].
\end{array}
\end{equation}

\section{Proofs for Quadratic Forms of Objectives Functions}
\begin{proof}[Proof of Proposition~\ref{prop:quadratic}]
The proposition follows directly according to Lemma~\ref{lemma:coeff:one-sample}, Lemma~\ref{lemma:one-sample:objective:quadratic}, Lemma~\ref{lemma:coeff:two-sample:SURE} and Lemma~\ref{lemma:two-sample:objective:quadratic}.
\end{proof}

\subsection{Lemmas of Quadratic Function Forms for One-Sample Case}
Define 
\begin{equation}\label{eq:notations:one-sample}
\begin{array}{rcl}
&\displaystyle\bar{Y}:=\frac{\sum_{k=1}^NY_k}{\sum_{k=1}^Nn_k},\ \bar{\theta}:=\frac{\sum_{k=1}^Nn_k\theta_k}{\sum_{k=1}^Nn_k},&\\
\\
&\displaystyle g_i=g(\mathbf{X}_i):=\mathbb{E} \left[ \left. \frac{Y_i}{n_i} \, \right|\, \mathbf{X}_i \right],\ \hat{g}:=\frac{\sum_{k=1}^Nn_k\hat{g}_k(X_k)}{\sum_{k=1}^Nn_k},\ \bar{g}:=\frac{\sum_{i=1}^Nn_ig(\mathbf{X}_i)}{\sum_{i=1}^Nn_i}.&
\end{array}
\end{equation}
$$\boldsymbol{\beta}_i:=\left(\frac{Y_i}{n_i}-\frac{\sum_{j=1}^NY_j}{\sum_{j=1}^Nn_j},\hat{g}_i(\mathbf{X}_i)-\frac{\sum_{j=1}^Nn_j\hat{g}_j(\mathbf{X}_{j})}{\sum_{j=1}^Nn_j}\right)^T,$$
$$\boldsymbol{\beta}_i(Y_i+j):=\left(\frac{Y_i+j}{n_i}-\frac{\sum_{j=1}^NY_j}{\sum_{j=1}^Nn_j},\hat{g}_i(\mathbf{X}_i)-\frac{\sum_{j=1}^Nn_j\hat{g}_j(\mathbf{X}_{j})}{\sum_{j=1}^Nn_j}\right)^T,$$
$$\boldsymbol{\beta}_i(Y_i-j):=\left(\frac{Y_i-j}{n_i}-\frac{\sum_{j=1}^NY_j}{\sum_{j=1}^Nn_j},\hat{g}_i(\mathbf{X}_i)-\frac{\sum_{j=1}^Nn_j\hat{g}_j(\mathbf{X}_{j})}{\sum_{j=1}^Nn_j}\right)^T.$$
$$\begin{array}{rcl}
\overline{\boldsymbol{\beta}}_i=\overline{\boldsymbol{\beta}}_i(Y_i):&=&\displaystyle\left(\frac{Y_i}{n_i}-\bar{\theta},g(\mathbf{X}_i)-\bar{g}\right)^T,\\
\overline{\boldsymbol{\beta}}_i(Y_i+j):&=&\displaystyle\left(\frac{Y_i+j}{n_i}-\bar{\theta},g(\mathbf{X}_i)-\bar{g}\right)^T,\\
\overline{\boldsymbol{\beta}}_i(Y_i-j):&=&\displaystyle\left(\frac{Y_i-j}{n_i}-\bar{\theta},g(\mathbf{X}_i)-\bar{g}\right)^T,
\end{array}$$
where $\displaystyle\frac{Y_i+j}{n_i},\frac{Y_i-j}{n_i}, g_i, \hat{g}_i, \bar{g}, \hat{g}\in[0,1]$. Let $\mathbf{\Delta}_i=(\bar{\theta}-\bar{Y},\hat{g}_i-g_i+\bar{g}-\hat{g})$, then 
\begin{equation}\label{def:beta}
\begin{array}{rcl}
\boldsymbol\beta_i&=&\overline{\boldsymbol{\beta}}_i+\mathbf{\Delta}_i\\
\boldsymbol\beta_i(Y_i+j)&=&\overline{\boldsymbol{\beta}}_i(Y_i+j)+\mathbf{\Delta}_i\\
\boldsymbol\beta_i(Y_i-j)&=&\overline{\boldsymbol{\beta}}_i(Y_i-j)+\mathbf{\Delta}_i.
\end{array}
\end{equation}
It is easy to see that $\max\{\|\boldsymbol{\bar\beta}_i\|_{\infty},\|\boldsymbol{\bar\beta}_i(Y_i+j)\|_{\infty},\|\boldsymbol{\bar\beta}_i(Y_i-j)\|_{\infty}\}\leq2$. 

\begin{lemma}\label{lemma:coeff:one-sample}
$\hat{L}_{\mathrm{o}}(\boldsymbol{\lambda})=\boldsymbol{\lambda}^T\mathbf{C}_{N,2}\boldsymbol{\lambda}+\mathbf{C}_{N,1}^T\boldsymbol{\lambda}+\mathbf{C}_0$, where $\boldsymbol{\lambda}=(\lambda_1,\lambda_2)^T$, $C_0$ is a constant not related to $\boldsymbol{\lambda}$, and 
\begin{equation}\label{eq:coefficient:quadratic:ML}
    \mathbf{C}_{N,2}=\frac{1}{N}\sum_{i=1}^N\boldsymbol{\beta}_i\boldsymbol{\beta}_i^T,
\end{equation}
\begin{equation}\label{eq:coeff:linear:one-sample}
\begin{array}{rl}
\mathbf{C}_{N,1}\!\!\!\!\!\!&\displaystyle=\frac{2\sum_{j=1}^NY_j}{\sum_{j=1}^Nn_j}\left\{\frac{1}{N}\sum_{i=1}^N\boldsymbol{\beta}_i\right\}-\frac{1}{N}\sum_{i=1}^N2\mathbf{1}(Y_{i} > \lfloor n_i/2\rfloor)\!\!\!\sum_{j=0}^{n_{i}-Y_{i}}\!\!\boldsymbol{\beta}_i(Y_i+j)(-1)^{j}\frac{\binom{n_i-Y_i}{j}}{\binom{Y_i+j}{j}}\\
&\quad\displaystyle-\frac{1}{N}\sum_{i=1}^N2\boldsymbol{\beta}_i\mathbf{1}\{Y_i\leq \lfloor n_i/2\rfloor\}+\frac{1}{N}\sum_{i=1}^N2\mathbf{1}(Y_{i}\leq \lfloor n_i/2\rfloor)\sum_{j=0}^{Y_{i}}\boldsymbol{\beta}_i(Y_i-j)(-1)^{j}\frac{\binom{Y_i}{j}}{\binom{n_i-Y_i+j}{j}}.
\end{array}
\end{equation} 
\end{lemma}
\begin{proof}[Proof of Lemma \ref{lemma:coeff:one-sample}]
Let $\boldsymbol{\lambda}=(\lambda_1,\lambda_2)^T$,
According to \eqref{eq:SURE:one-sample}, the one-sample SURE using estimator \eqref{eq:two-sample:estimator:lambda:one-sample} can be written as 
\begin{equation}\label{eq:SURE:one-sample:ML}
\begin{array}{rl}
    \hat{L}_{\mathrm{o}}(\boldsymbol{\lambda})\!\!\!\!\!\!&\displaystyle=\frac{1}{N}\sum_{i=1}^N\left(\boldsymbol{\beta}_i^T\boldsymbol{\lambda}+\frac{\sum_{j=1}^NY_j}{\sum_{j=1}^Nn_j}\right)^2-\hat{\ell}_{i,u}^{(1)}\mathbf{1}\{Y_i>\lfloor n_i/2\rfloor\}-\hat{\ell}_{i,u}^{(2)}\mathbf{1}\{Y_{i}\leq\lfloor n_i/2\rfloor\}\\
    &\displaystyle=\boldsymbol{\lambda}^T\left[\frac{1}{N}\sum_{i=1}^N\boldsymbol{\beta}_i\boldsymbol{\beta}_i^T\right]\boldsymbol{\lambda}+\frac{2\sum_{j=1}^NY_j}{\sum_{j=1}^Nn_j}\boldsymbol{\lambda}^T\left\{\frac{1}{N}\sum_{i=1}^N\boldsymbol{\beta}_i\right\}+\left(\frac{\sum_{j=1}^NY_j}{\sum_{j=1}^Nn_j}\right)^2\\
    &\displaystyle\quad-\frac{1}{N}\sum_{i=1}^N2\mathbf{1}(Y_{i} > \lfloor n_i/2\rfloor)\\
    &\displaystyle\quad\quad\quad\quad\times\sum_{j=0}^{n_{i}-Y_{i}}\!\!\!\left(\boldsymbol{\beta}_i(Y_i+j)^T\boldsymbol{\lambda}+\frac{\sum_{j=1}^NY_j}{\sum_{j=1}^Nn_j}\right)(-1)^{j}\frac{(n_{i}-Y_{i})!}{(n_{i}-Y_{i}-j)!}\frac{Y_{i}!}{(Y_{i}+j)!}\\
    &\displaystyle\quad-\frac{1}{N}\sum_{i=1}^N2\left(\boldsymbol{\beta}_i^T\boldsymbol{\lambda}+\frac{\sum_{j=1}^NY_j}{\sum_{j=1}^Nn_j}\right)\mathbf{1}\{Y_i\leq \lfloor n_i/2\rfloor\}\\
    &\displaystyle\quad+\frac{1}{N}\sum_{i=1}^N2\mathbf{1}(Y_{i}\leq \lfloor n_i/2\rfloor)\\
    &\displaystyle\quad\quad\quad\quad\times\sum_{j=0}^{Y_{i}}\left(\boldsymbol{\beta}_i(Y_i-j)^T\boldsymbol{\lambda}+\frac{\sum_{j=1}^NY_j}{\sum_{j=1}^Nn_j}\right)(-1)^{j}\frac{Y_{i}!}{(Y_{i}-j)!}\frac{(n_{i} - Y_{i})!}{(n_{i}-Y_{i}+j)!},
\end{array}
\end{equation}
so the coefficient for quadratic term is 
\begin{equation}\label{eq:coeff:quadratic:one-sample}
    \mathbf{C}_{N,2}=\frac{1}{N}\sum_{i=1}^N\boldsymbol{\beta}_i\boldsymbol{\beta}_i^T,
\end{equation}
and the coefficient for the first-order term is 
\begin{equation}\label{eq:coefficient:linear:ML}
\begin{array}{rl}
    \mathbf{C}_{N,1}\!\!\!\!\!\!&\displaystyle=\frac{2\sum_{j=1}^NY_j}{\sum_{j=1}^Nn_j}\left\{\frac{1}{N}\sum_{i=1}^N\boldsymbol{\beta}_i\right\}\\
    &\displaystyle\quad-\frac{1}{N}\sum_{i=1}^N2\mathbf{1}(Y_{i} > \lfloor n_i/2\rfloor)\!\!\!\sum_{j=0}^{n_{i}-Y_{i}}\!\!\!\boldsymbol{\beta}_i(Y_i+j)(-1)^{j}\frac{(n_{i}-Y_{i})!}{(n_{i}-Y_{i}-j)!}\frac{Y_{i}!}{(Y_{i}+j)!}\\
    &\quad\displaystyle-\frac{1}{N}\sum_{i=1}^N2\boldsymbol{\beta}_i\mathbf{1}\{Y_i\leq \lfloor n_i/2\rfloor\}\\
    &\quad\displaystyle+\frac{1}{N}\sum_{i=1}^N2\mathbf{1}(Y_{i}\leq \lfloor n_i/2\rfloor)\sum_{j=0}^{Y_{i}}\boldsymbol{\beta}_i(Y_i-j)(-1)^{j}\frac{Y_{i}!}{(Y_{i}-j)!}\frac{(n_{i} - Y_{i})!}{(n_{i}-Y_{i}+j)!}\\
    &\displaystyle=\frac{2\sum_{j=1}^NY_j}{\sum_{j=1}^Nn_j}\left\{\frac{1}{N}\sum_{i=1}^N\boldsymbol{\beta}_i\right\}-\frac{1}{N}\sum_{i=1}^N2\mathbf{1}(Y_{i} > \lfloor n_i/2\rfloor)\!\!\!\sum_{j=0}^{n_{i}-Y_{i}}\!\!\!\boldsymbol{\beta}_i(Y_i+j)(-1)^{j}\frac{\binom{n_i-Y_i}{j}}{\binom{Y_i+j}{j}}\\
    &\quad\displaystyle-\frac{1}{N}\sum_{i=1}^N2\boldsymbol{\beta}_i\mathbf{1}\{Y_i\leq \lfloor n_i/2\rfloor\}+\frac{1}{N}\sum_{i=1}^N2\mathbf{1}(Y_{i}\leq \lfloor n_i/2\rfloor)\sum_{j=0}^{Y_{i}}\boldsymbol{\beta}_i(Y_i-j)(-1)^{j}\frac{\binom{Y_i}{j}}{\binom{n_i-Y_i+j}{j}}
\end{array}
\end{equation}
\end{proof}

\begin{lemma}\label{lemma:one-sample:objective:quadratic}
$L_{\mathrm{o}}(\boldsymbol{\lambda})=\boldsymbol{\lambda}^T\mathbf{C}_{2}\boldsymbol{\lambda}+\mathbf{C}_{1}^T\boldsymbol{\lambda}+C_0^*$, where $\boldsymbol{\lambda}=(\lambda_1,\lambda_2)^T$, $C_0^*$ is a constant not related to $\boldsymbol{\lambda}$, and 
\begin{equation}\label{eq:coefficient:quadratic:ML:objective}
    \mathbf{C}_{2}=\frac{1}{N}\sum_{i=1}^N\mathbb{E}\left[\boldsymbol{\beta}_i\boldsymbol{\beta}_i^T\right],
\end{equation}
\begin{equation}\label{eq:coeff:linear:one-sample:objective}
\mathbf{C}_{1}=\frac{2}{N}\sum_{i=1}^N\mathbb{E}\left[\left(\bar{Y}-\theta_i^{\mathrm{o}}\right)\boldsymbol{\beta}_i\right].
\end{equation} 
\end{lemma}
\begin{proof}[Proof of Lemma~\ref{lemma:one-sample:objective:quadratic}]
The lemma is straightforward to prove by expanding \eqref{eq:L2:one-sample:lambda} with the definition of $\hat{\theta}_i^{\mathrm{o}}(\boldsymbol{\lambda})$ defined as \eqref{eq:one-sample:estimator:cross-fitted}. 
\end{proof}

\subsection{Lemmas of Quadratic Function Forms for Two-Sample Case}
Denote 
{\small$$\boldsymbol{\beta}_{i1}:=\begin{pmatrix}
\displaystyle\frac{Y_{i1}}{n_{i1}}-\frac{\sum_{j=1}^NY_{j1}}{\sum_{j=1}^Nn_{j1}}\\
\displaystyle\hat{g}_{i1}(\mathbf{X}_{i1})-\frac{\sum_{j=1}^Nn_{j1}\hat{g}_{j1}(\mathbf{X}_{j1})}{\sum_{j=1}^Nn_{j1}}
\end{pmatrix},\ \boldsymbol{\beta}_{i2}:=\begin{pmatrix}
\displaystyle\frac{Y_{i2}}{n_{i2}}-\frac{\sum_{j=1}^NY_{j2}}{\sum_{j=1}^Nn_{j2}}\\
\displaystyle\hat{g}_{i2}(\mathbf{X}_{i2})-\frac{\sum_{j=1}^Nn_{j2}\hat{g}_{j2}(\mathbf{X}_{j2})}{\sum_{j=1}^Nn_{j2}}
\end{pmatrix},$$}
{\small$$\boldsymbol{\beta}_{i1}(Y_{i1}+j):=\begin{pmatrix}
\displaystyle\frac{Y_{i1}+j}{n_{i1}}-\frac{\sum_{j=1}^NY_{j1}}{\sum_{j=1}^Nn_{j1}}\\
\displaystyle\hat{g}_{i1}(\mathbf{X}_{i1})-\frac{\sum_{j=1}^Nn_{j1}\hat{g}_{j1}(\mathbf{X}_{j1})}{\sum_{j=1}^Nn_{j1}}
\end{pmatrix},\ \boldsymbol{\beta}_{i2}(Y_{i2}+j):=\begin{pmatrix}
\displaystyle\frac{Y_{i2}+j}{n_{i2}}-\frac{\sum_{j=1}^NY_{j2}}{\sum_{j=1}^Nn_{j2}}\\
\displaystyle\hat{g}_{2}(\mathbf{X}_{i2})-\frac{\sum_{j=1}^Nn_{j2}\hat{g}_{j2}(\mathbf{X}_{j2})}{\sum_{j=1}^Nn_{j2}}
\end{pmatrix},$$}
{\small$$\boldsymbol{\beta}_{i1}(Y_{i1}-j):=\begin{pmatrix}
\displaystyle\frac{Y_{i1}-j}{n_{i1}}-\frac{\sum_{j=1}^NY_{j1}}{\sum_{j=1}^Nn_{j1}}\\
\displaystyle\hat{g}_{1}(\mathbf{X}_{i1})-\frac{\sum_{j=1}^Nn_{j1}\hat{g}_{j1}(\mathbf{X}_{j1})}{\sum_{j=1}^Nn_{j1}}
\end{pmatrix},\ \boldsymbol{\beta}_{i2}(Y_{i2}-j):=\begin{pmatrix}
\displaystyle\frac{Y_{i2}-j}{n_{i2}}-\frac{\sum_{j=1}^NY_{j2}}{\sum_{j=1}^Nn_{j2}}\\
\displaystyle\hat{g}_{2}(\mathbf{X}_{i2})-\frac{\sum_{j=1}^Nn_{j2}\hat{g}_{j2}(\mathbf{X}_{j2})}{\sum_{j=1}^Nn_{j2}}
\end{pmatrix}.$$}
\begin{equation}\label{eq:two-sample:average:Y}
    \bar{Y}_1=\frac{\sum_{j=1}^NY_{j1}}{\sum_{j=1}^Nn_{j1}},\quad \bar{Y}_2=\frac{\sum_{j=1}^NY_{j2}}{\sum_{j=1}^Nn_{j2}}.
\end{equation}
$$\bar{Y}_1=\frac{\sum_{j=1}^NY_{j1}}{\sum_{j=1}^Nn_{j1}},\ \bar{Y}_2=\frac{\sum_{j=1}^NY_{j2}}{\sum_{j=1}^Nn_{j2}}, \ \bar{\theta}_1=\frac{\sum_{j=1}^Nn_{j1}\theta_{j1}}{\sum_{j=1}^Nn_{j1}},\ \bar{\theta}_2=\frac{\sum_{j=1}^Nn_{j2}\theta_{j2}}{\sum_{j=1}^Nn_{j2}},$$
$$g_{1i}:=g_1(\mathbf{X}_i)=\mathbb{E}\left[\frac{Y_{i1}}{n_{i1}}\bigg|\mathbf{X}_{i1}\right],\ \bar{g}_1=\frac{\sum_{i=1}^Nn_{i1}g_1(\mathbf{X}_i)}{\sum_{i=1}^Nn_{i1}},\ \hat{g}_1=\frac{\sum_{j=1}^Nn_{j1}\hat{g}_{1}(\mathbf{X}_{j1})}{\sum_{j=1}^Nn_{j1}}.$$
$$g_{2i}:=g_2(\mathbf{X}_i)=\mathbb{E}\left[\frac{Y_{i2}}{n_{i2}}\bigg|\mathbf{X}_{i2}\right],\ \bar{g}_2=\frac{\sum_{i=1}^Nn_{i2}g_2(\mathbf{X}_i)}{\sum_{i=1}^Nn_{i2}},\ \hat{g}_2=\frac{\sum_{j=1}^Nn_{j2}\hat{g}_{2}(\mathbf{X}_{j2})}{\sum_{j=1}^Nn_{j2}}.$$
$$\Delta\hat{g}_i=\left[\hat{g}_{i1}(\mathbf{X}_{i1})-\frac{\sum_{j=1}^Nn_{j1}\hat{g}_{j1}(\mathbf{X}_{j1})}{\sum_{j=1}^Nn_{j1}}\right]-\left[\hat{g}_{i2}(\mathbf{X}_{i2})-\frac{\sum_{j=1}^Nn_{j2}\hat{g}_{j2}(\mathbf{X}_{j2})}{\sum_{j=1}^Nn_{j2}}\right].$$

\begin{lemma}\label{lemma:coeff:two-sample:SURE}
$\hat{L}_{\mathrm{t}}(\boldsymbol{\lambda})=\boldsymbol{\lambda}^T\mathbf{D}_{N,2}\boldsymbol{\lambda}+\mathbf{D}_{N,1}^T\boldsymbol{\lambda}+D_0$, where $\boldsymbol{\lambda}=(\lambda_1,\lambda_2)^T$, $D_0$ is a constant matrix not related to $\boldsymbol{\lambda}$, 
\begin{equation}\label{eq:coefficient:quadratic:ML:two-sample}
    \mathbf{D}_{N,2}=\frac{1}{N}\sum_{i=1}^N(\boldsymbol{\beta}_{i1}-\boldsymbol{\beta}_{i2})(\boldsymbol{\beta}_{i1}-\boldsymbol{\beta}_{i2})^T,
\end{equation}
\begin{equation}\label{eq:D:N:1}
\begin{array}{rl}
\mathbf{D}_{N,1}\!\!\!\!\!&\displaystyle=2(\bar{Y}_1-\bar{Y}_2)\left\{\frac{1}{N}\sum_{i=1}^N(\boldsymbol{\beta}_{i1}-\boldsymbol{\beta}_{i2})\right\}\\
    &\displaystyle\quad-\frac{1}{N}\sum_{i=1}^N2\mathbf{1}(Y_{i1} > \lfloor n_{i1}/2\rfloor)\!\!\!\sum_{j=0}^{n_{i1}-Y_{i1}}\!\!\!(\boldsymbol{\beta}_{i1}(Y_{i1}+j)-\boldsymbol{\beta}_{i2})(-1)^{j}\frac{(n_{i1}-Y_{i1})!}{(n_{i1}-Y_{i1}-j)!}\frac{Y_{i1}!}{(Y_{i1}+j)!}\\
    &\quad\displaystyle-\frac{1}{N}\sum_{i=1}^N2(\boldsymbol{\beta}_{i1}-\boldsymbol{\beta}_{i2})\mathbf{1}\{Y_{i1}\leq \lfloor n_{i1}/2\rfloor\}\\
    &\quad\displaystyle+\frac{1}{N}\sum_{i=1}^N2\mathbf{1}(Y_{i1}\leq \lfloor n_{i1}/2\rfloor)\sum_{j=0}^{Y_{i1}}(\boldsymbol{\beta}_{i1}(Y_{i1}-j)-\boldsymbol{\beta}_{i2})(-1)^{j}\frac{Y_{i1}!}{(Y_{i1}-j)!}\frac{(n_{i1} - Y_{i1})!}{(n_{i1}-Y_{i1}+j)!}\\
    &\displaystyle\quad+\frac{1}{N}\sum_{i=1}^N2\mathbf{1}(Y_{i2} > \lfloor n_{i2}/2\rfloor)\!\!\!\sum_{j=0}^{n_{i2}-Y_{i2}}\!\!\!(\boldsymbol{\beta}_{i1}-\boldsymbol{\beta}_{i2}(Y_{i2}+j))(-1)^{j}\frac{(n_{i2}-Y_{i2})!}{(n_{i2}-Y_{i2}-j)!}\frac{Y_{i2}!}{(Y_{i2}+j)!}\\
    &\quad\displaystyle+\frac{1}{N}\sum_{i=1}^N2(\boldsymbol{\beta}_{i1}-\boldsymbol{\beta}_{i2})\mathbf{1}\{Y_{i2}\leq \lfloor n_{i2}/2\rfloor\}\\
    &\quad\displaystyle-\frac{1}{N}\sum_{i=1}^N2\mathbf{1}(Y_{i2}\leq \lfloor n_{i2}/2\rfloor)\sum_{j=0}^{Y_{i2}}(\boldsymbol{\beta}_{i1}-\boldsymbol{\beta}_{i2}(Y_{i2}-j))(-1)^{j}\frac{Y_{i2}!}{(Y_{i2}-j)!}\frac{(n_{i2} - Y_{i2})!}{(n_{i2}-Y_{i2}+j)!}.
\end{array}
\end{equation}
\end{lemma}
\begin{proof}[Proof of Lemma \ref{lemma:coeff:two-sample:SURE}]
Let $\boldsymbol{\lambda}=(\lambda_1,\lambda_2)^T$, 
According to \eqref{eq:L2:risk:estimator},
\begin{equation}\label{eq:L:SURE:two-sample:ML}
\begin{array}{rl}
    &\quad\hat{L}_{\mathrm{t}}(\boldsymbol{\lambda})\\
    &\displaystyle=\frac{1}{N}\sum_{i=1}^N\left((\boldsymbol{\beta}_{i1}-\boldsymbol{\beta}_{i2})^T\boldsymbol{\lambda}+\bar{Y}_1-\bar{Y}_2\right)^2-\hat{\ell}_{i,u}^{1,(1)}\mathbf{1}\{Y_{i1}>\lfloor n_{i1}/2\rfloor\}-\hat{\ell}_{i,u}^{1,(2)}\mathbf{1}\{Y_{i1}\leq\lfloor n_{i1}/2\rfloor\}\\
    &\quad\quad\quad\quad\quad\displaystyle+\hat{\ell}_{i,u}^{2,(1)}\mathbf{1}\{Y_{i2}>\lfloor n_{i2}/2\rfloor\}+\hat{\ell}_{i,u}^{2,(2)}\mathbf{1}\{Y_{i2}\leq\lfloor n_{i2}/2\rfloor\}\\
    &\displaystyle=\boldsymbol{\lambda}^T\left[\frac{1}{N}\sum_{i=1}^N(\boldsymbol{\beta}_{i1}-\boldsymbol{\beta}_{i2})(\boldsymbol{\beta}_{i1}-\boldsymbol{\beta}_{i2})^T\right]\boldsymbol{\lambda}\\
    &\displaystyle\quad+2(\bar{Y}_1-\bar{Y}_2)\boldsymbol{\lambda}^T\left\{\frac{1}{N}\sum_{i=1}^N(\boldsymbol{\beta}_{i1}-\boldsymbol{\beta}_{i2})\right\}+(\bar{Y}_1-\bar{Y}_2)^2\\
    &\displaystyle\quad-\frac{1}{N}\sum_{i=1}^N2\mathbf{1}(Y_{i1} > \lfloor n_{i1}/2\rfloor)\\
    &\displaystyle\quad\quad\quad\quad\quad\times\sum_{j=0}^{n_{i1}-Y_{i1}}\!\!\!\left((\boldsymbol{\beta}_{i1}(Y_{i1}+j)-\boldsymbol{\beta}_{i2})^T\boldsymbol{\lambda}+(\bar{Y}_1-\bar{Y}_2)\right)(-1)^{j}\frac{(n_{i1}-Y_{i1})!}{(n_{i1}-Y_{i1}-j)!}\frac{Y_{i1}!}{(Y_{i1}+j)!}\\
    &\displaystyle\quad-\frac{1}{N}\sum_{i=1}^N2\left((\boldsymbol{\beta}_{i1}-\boldsymbol{\beta}_{i2})^T\boldsymbol{\lambda}+(\bar{Y}_{1}-\bar{Y}_2)\right)\mathbf{1}\{Y_{i1}\leq \lfloor n_{i1}/2\rfloor\}\\
    &\displaystyle\quad+\frac{1}{N}\sum_{i=1}^N2\mathbf{1}(Y_{i1}\leq \lfloor n_{i1}/2\rfloor)\\
    &\displaystyle\quad\quad\quad\quad\quad\times\sum_{j=0}^{Y_{i1}}\left((\boldsymbol{\beta}_{i1}(Y_{i1}-j)-\boldsymbol{\beta}_{i2})^T\boldsymbol{\lambda}+(\bar{Y}_1-\bar{Y}_2)\right)(-1)^{j}\frac{Y_{i1}!}{(Y_{i1}-j)!}\frac{(n_{i1} - Y_{i1})!}{(n_{i1}-Y_{i1}+j)!}\\
    &\displaystyle\quad+\frac{1}{N}\sum_{i=1}^N2\mathbf{1}(Y_{i2} > \lfloor n_{i2}/2\rfloor)\\
    &\displaystyle\quad\quad\quad\quad\quad\times\sum_{j=0}^{n_{i2}-Y_{i2}}\!\!\!\left((\boldsymbol{\beta}_{i1}-\boldsymbol{\beta}_{i2}(Y_{i2}+j))^T\boldsymbol{\lambda}+(\bar{Y}_1-\bar{Y}_2)\right)(-1)^{j}\frac{(n_{i2}-Y_{i2})!}{(n_{i2}-Y_{i2}-j)!}\frac{Y_{i2}!}{(Y_{i2}+j)!}\\
    &\displaystyle\quad+\frac{1}{N}\sum_{i=1}^N2\left((\boldsymbol{\beta}_{i1}-\boldsymbol{\beta}_{i2})^T\boldsymbol{\lambda}+(\bar{Y}_{1}-\bar{Y}_2)\right)\mathbf{1}\{Y_{i2}\leq \lfloor n_{i2}/2\rfloor\}\\
    &\displaystyle\quad-\frac{1}{N}\sum_{i=1}^N2\mathbf{1}(Y_{i2}\leq \lfloor n_{i2}/2\rfloor)\\
    &\displaystyle\quad\quad\quad\quad\quad\times\sum_{j=0}^{Y_{i2}}\left((\boldsymbol{\beta}_{i1}-\boldsymbol{\beta}_{i2}(Y_{i2}-j))^T\boldsymbol{\lambda}+(\bar{Y}_1-\bar{Y}_2)\right)(-1)^{j}\frac{Y_{i2}!}{(Y_{i2}-j)!}\frac{(n_{i2} - Y_{i2})!}{(n_{i2}-Y_{i2}+j)!}
\end{array}
\end{equation}
So the coefficient for the quadratic term is 
\begin{equation}\label{eq:coeff:quadratic:two-sample}
    \mathbf{D}_{N,2}:=\frac{1}{N}\sum_{i=1}^N(\boldsymbol{\beta}_{i1}-\boldsymbol{\beta}_{i2})(\boldsymbol{\beta}_{i1}-\boldsymbol{\beta}_{i2})^T,
\end{equation}
and the coefficient for the first-order term is 

\begin{equation}\label{eq:coefficient:linear:ML:two-sample}
\begin{array}{rl}
    \mathbf{D}_{N,1}\!\!\!\!\!\!&\displaystyle=2(\bar{Y}_1-\bar{Y}_2)\left\{\frac{1}{N}\sum_{i=1}^N(\boldsymbol{\beta}_{i1}-\boldsymbol{\beta}_{i2})\right\}\\
    &\displaystyle\quad-\frac{1}{N}\sum_{i=1}^N2\mathbf{1}(Y_{i1} > \lfloor n_{i1}/2\rfloor)\!\!\!\sum_{j=0}^{n_{i1}-Y_{i1}}\!\!\!(\boldsymbol{\beta}_{i1}(Y_{i1}+j)-\boldsymbol{\beta}_{i2})(-1)^{j}\frac{(n_{i1}-Y_{i1})!}{(n_{i1}-Y_{i1}-j)!}\frac{Y_{i1}!}{(Y_{i1}+j)!}\\
    &\quad\displaystyle-\frac{1}{N}\sum_{i=1}^N2(\boldsymbol{\beta}_{i1}-\boldsymbol{\beta}_{i2})\mathbf{1}\{Y_{i1}\leq \lfloor n_{i1}/2\rfloor\}\\
    &\quad\displaystyle+\frac{1}{N}\sum_{i=1}^N2\mathbf{1}(Y_{i1}\leq \lfloor n_{i1}/2\rfloor)\sum_{j=0}^{Y_{i1}}(\boldsymbol{\beta}_{i1}(Y_{i1}-j)-\boldsymbol{\beta}_{i2})(-1)^{j}\frac{Y_{i1}!}{(Y_{i1}-j)!}\frac{(n_{i1} - Y_{i1})!}{(n_{i1}-Y_{i1}+j)!}\\
    &\displaystyle\quad+\frac{1}{N}\sum_{i=1}^N2\mathbf{1}(Y_{i2} > \lfloor n_{i2}/2\rfloor)\!\!\!\sum_{j=0}^{n_{i2}-Y_{i2}}\!\!\!(\boldsymbol{\beta}_{i1}-\boldsymbol{\beta}_{i2}(Y_{i2}+j))(-1)^{j}\frac{(n_{i2}-Y_{i2})!}{(n_{i2}-Y_{i2}-j)!}\frac{Y_{i2}!}{(Y_{i2}+j)!}\\
    &\quad\displaystyle+\frac{1}{N}\sum_{i=1}^N2(\boldsymbol{\beta}_{i1}-\boldsymbol{\beta}_{i2})\mathbf{1}\{Y_{i2}\leq \lfloor n_{i2}/2\rfloor\}\\
    &\quad\displaystyle-\frac{1}{N}\sum_{i=1}^N2\mathbf{1}(Y_{i2}\leq \lfloor n_{i2}/2\rfloor)\sum_{j=0}^{Y_{i2}}(\boldsymbol{\beta}_{i1}-\boldsymbol{\beta}_{i2}(Y_{i2}-j))(-1)^{j}\frac{Y_{i2}!}{(Y_{i2}-j)!}\frac{(n_{i2} - Y_{i2})!}{(n_{i2}-Y_{i2}+j)!}.
\end{array}
\end{equation}
\end{proof}

\begin{lemma}\label{lemma:two-sample:objective:quadratic}
$L_{\mathrm{o}}(\boldsymbol{\lambda})=\boldsymbol{\lambda}^T\mathbf{D}_{2}\boldsymbol{\lambda}+\mathbf{D}_{1}^T\boldsymbol{\lambda}+D_0^*$, where $\boldsymbol{\lambda}=(\lambda_1,\lambda_2)^T$, $D_0^*$ is a constant not related to $\boldsymbol{\lambda}$, and 
\begin{equation}\label{eq:coefficient:quadratic:two-sample:objective}
    \mathbf{D}_{2}=\frac{1}{N}\sum_{i=1}^N\mathbb{E}\left[(\boldsymbol{\beta}_{i1}-\boldsymbol{\beta}_{i2})(\boldsymbol{\beta}_{i1}-\boldsymbol{\beta}_{i2})^T\right],
\end{equation}
\begin{equation}\label{eq:coeff:linear:two-sample:objective}
\mathbf{D}_{1}=\frac{2}{N}\sum_{i=1}^N\mathbb{E}\left[\left\{(\bar{Y}_1-\bar{Y}_2)-(\theta_{i1}^{\mathrm{t}}-\theta_{i2}^{\mathrm{t}})\right\}(\boldsymbol{\beta}_{i1}-\boldsymbol{\beta}_{i2})\right].
\end{equation} 
\end{lemma}
\begin{proof}[Proof of Lemma~\ref{lemma:one-sample:objective:quadratic}]
The lemma is straightforward to prove by expanding \eqref{eq:L2:two-sample:lambda} with the definition of $\hat{\theta}_i^{\mathrm{t}}(\boldsymbol\lambda) = \hat{\theta}_{i1}^{\mathrm{o}}(\boldsymbol{\lambda}) - \hat{\theta}_{i2}^{\mathrm{o}}(\boldsymbol{\lambda})$  defined as \eqref{eq:two-sample:estimator:lambda}. 
\end{proof}

\section{Proofs for Asymptotic Normality}
\begin{proof}[Proof of Theorem~\ref{thm:asymptotics}]
Theorem~\ref{thm:asymptotics} follows directly from Theorem~\ref{thm:one-sample:asymptotics} and Theorem~\ref{thm:CLT:two-sample}.
\end{proof}

\subsection{Proofs for Asymptotic Normality of One-Sample Case}\label{appendix:one-sample:asymptotis}
\begin{assumption}\label{ass:additional:one-sample}
Suppose the following statements hold:\\
\noindent (i) $\displaystyle\frac{1}{N}\sum_{i=1}^N\mathrm{Var}(Y_i)\rightarrow\sigma_Y^2$, $\displaystyle\frac{1}{N}\sum_{i=1}^N\mathbb{E}[Y_i]\rightarrow\mu_Y^*$, $\displaystyle\frac{1}{N}\sum_{i=1}^Nn_i\rightarrow\mu_n^*$, $\displaystyle\frac{1}{N}\sum_{i=1}^Nn_ig(\mathbf{X}_i)\rightarrow\mu_{gn}^*$,\\ 
$\displaystyle\frac{1}{N}\sum_{i=1}^N\theta_i^{\mathrm{o}}\rightarrow\mu_{\theta}^*$, $\displaystyle\frac{1}{N}\sum_{i=1}^N\left(\theta_i^{\mathrm{o}}\right)^2\rightarrow\sigma_{\theta}^2$, $\displaystyle\frac{1}{N}\sum_{i=1}^Ng(\mathbf{X}_i)\rightarrow\mu_g^*$, $\displaystyle\frac{1}{N}\sum_{i=1}^N\mathbb{E}\left[\mathbf{1}\{Y_i>\lfloor n_i/2\rfloor\}\right]\rightarrow\mu_I^*$, where $\sigma_Y^2$, $\sigma_{gn}^2$, $\mu_Y^*$, $\mu_n^*$, $\mu_{gn}^*$, $\mu_{\theta}^*$, $\sigma_{\theta}$, $\mu_g^*$, $\mu_I^*$ are all absolute constants.

\noindent (ii) $\displaystyle\frac{1}{N}\sum_{i=1}^N\mathrm{Cov}\left\{\tilde{\mathbf{Z}}_i\right\}\rightarrow\widetilde{\boldsymbol{\Sigma}}$, where $\widetilde{\boldsymbol{\Sigma}}\in\mathbb{R}^{6\times 6}$ is a positive definite matrix and $$\tilde{\mathbf{Z}}_i:=(\tilde{\zeta}_{i1},\tilde{\zeta}_{i2},\tilde{\zeta}_{i3},\tilde{\zeta}_{i4},\tilde{\zeta}_{i5},\tilde{\zeta}_{i6})^T\in\mathbb{R}^6,\quad i\in[N]$$
are i.n.i.d. vectors defined as \eqref{eq:ass:entry:cov:one-sample}:
\begin{equation}\label{eq:ass:entry:cov:one-sample}
\begin{array}{rcl}
&\displaystyle\Delta_{i}^{Y}:=Y_i-\mathbb{E}[Y_i],&\\
&\displaystyle\Delta_i^{gn}:=n_ig(\mathbf{X}_i)-\mathbb{E}[n_ig(\mathbf{X}_i)],&\\
&\displaystyle\tilde{\zeta}_{i1}=4\frac{\left(\mu_{\theta}^*-\frac{\mu_Y^*}{\mu_n^*}\right)}{\mu_n^*}\Delta_{i}^{Y}+2\left\{\frac{Y_i/n_i(1-Y_i/n_i)}{n_i-1}-\left(\frac{Y_i}{n_i}-\frac{\mu_Y^*}{\mu_n^*}\right)^2\right\},&\\
&\displaystyle\tilde{\zeta}_{i2}=2\frac{\left(\mu_g^*-\frac{\mu_{gn}^*}{\mu_n^*}\right)}{\mu_n^*}\Delta_{i}^{Y}-2\frac{\left(\frac{\mu_Y^*}{\mu_n^*}-\mu_I^*\right)}{\mu_n^*}\Delta_i^{gn}+2\left(\frac{\mu_Y^*}{\mu_n^*}-\mathbf{1}\{Y_i>\lfloor n_i/2\rfloor\}\right)\left(g(\mathbf{X}_i)-\frac{\mu_{gn}^*}{\mu_n^*}\right),&\\
&\displaystyle\tilde{\zeta}_{i3}=2\frac{\left(\mu_{\theta}^*-\frac{\mu_Y^*}{\mu_n^*}\right)}{\mu_n^*}\Delta_{i}^{Y}+\left(\frac{Y_i}{n_i}-\frac{\mu_Y^*}{\mu_n^*}\right)^2,&\\
&\displaystyle\tilde{\zeta}_{i4}=\tilde{\zeta}_{i5}=\frac{\left(\mu_g^*-\frac{\mu_{gn}^*}{\mu_n^*}\right)}{\mu_n^*}\Delta_{i}^{Y}+\left(\frac{Y_i}{n_i}-\frac{\mu_Y^*}{\mu_n^*}\right)\left(g(\mathbf{X}_i)-\frac{\mu_{gn}^*}{\mu_n^*}\right),&\\
&\displaystyle\tilde{\zeta}_{i6}=2\frac{\left(\frac{\mu_{gn}^*}{\mu_n^*}-\mu_g^*\right)}{\mu_n^*}\Delta_i^{gn}+\left(g(\mathbf{X}_i)-\frac{\mu_{gn}^*}{\mu_n^*}\right)^2.&
\end{array}
\end{equation}
\end{assumption}

\begin{theorem}\label{thm:one-sample:asymptotics}
Suppose Assumption \ref{assump:one_sample} and Assumption~\ref{ass:additional:one-sample} hold. Suppose $\boldsymbol{\lambda}_{\mathrm{o}}^*$ is unconstrained. Then 
$$\sqrt{N}\left(\hat{\boldsymbol{\lambda}}_{\mathrm{o}}-\boldsymbol{\lambda}_{\mathrm{o}}^*\right)\rightsquigarrow\mathcal{N}(\mathbf{0},\mathbf{V}),$$ 
where $\mathbf{V}\preceq\bar{C}\mathbf{I}_2$ for some absolute constant $\bar{C}$. 
\end{theorem}
\begin{proof}[Proof of Theorem \ref{thm:one-sample:asymptotics}]
Recall that $\boldsymbol{\mu}_{m,1}=\mathbb{E}[\mathbf{C}_{N,1}]$ and $\boldsymbol{\mu}_{m,2}=\mathbb{E}[\mathbf{C}_{N,2}]$. According to Lemma \ref{lemma:CLT:one-sample} we have  
$$\sqrt{N}\begin{pmatrix}
\mathbf{C}_{N,1}-\boldsymbol{\mu}_{m,1}\\
\mathrm{vec}(\mathbf{C}_{N,2}-\boldsymbol{\mu}_{m,2})
\end{pmatrix}\rightsquigarrow\mathcal{N}\left(\mathbf{0},\widetilde{\boldsymbol{\Sigma}}\right).$$
Define 
$$g(c_1,c_2)=-\frac{1}{2}c_2^{-1}c_1$$ 
for $c_2\in\mathbb{R}^{2\times 2}$, $c_1\in\mathbb{R}^{2\times 1}$. So $$g(\boldsymbol{\mu}_{m,1},\boldsymbol{\mu}_{m,2})=-\frac{1}{2}\boldsymbol{\mu}_{m,2}^{-1}\boldsymbol{\mu}_{m,1}.$$ 
For perturbations $(h_1,h_2)\in\mathbb{R}^2\times\mathbb{R}^{2\times 2}$ at $(\mu_1,\mu_2)\in\mathbb{R}^2\times\mathbb{R}^{2\times 2}$, denote $\boldsymbol{\lambda}=\mu_2^{-1}\mu_1$, and using the fact that $\mathrm{vec}(ABC)=(C^T\otimes A)\mathrm{vec}(B)$, then we have 
$$Dg_{(\mu_1,\mu_2)}(h_1,h_2)=\frac{1}{2}\mu_2^{-1}h_1-\frac{1}{2}\mu_2^{-1}h_2\mu_2^{-1}\mu_1=\frac{1}{2}\mu_2^{-1}h_1-\frac{1}{2}(\boldsymbol{\lambda}^T\otimes\mu_2^{-1})\mathrm{vec}(h_2).$$
Hence the $2\times 6$ Jacobian at $(\boldsymbol{\mu}_{m,1},\boldsymbol{\mu}_{m,2})$ that multiplies the stacked vector $(h_1,\mathrm{vec}(h_2))^T$ is 
$$\mathbf{J}=\begin{bmatrix}
\displaystyle\frac{1}{2}\boldsymbol{\mu}_{m,2}^{-1}&\displaystyle\frac{1}{4}\boldsymbol{\mu}_{m,1}^T\boldsymbol{\mu}_{m,2}^{-1}\otimes\boldsymbol{\mu}_{m,2}^{-1}
\end{bmatrix}.$$
Applying delta's method, we then have 
$$\sqrt{N}\left(\hat{\boldsymbol{\lambda}}_{\mathrm{o}}-\left\{-\frac{1}{2}\boldsymbol{\mu}_{m,2}^{-1}\boldsymbol{\mu}_{m,1}\right\}\right)\rightsquigarrow\mathcal{N}\left(\mathbf{0},\mathbf{J}\widetilde{\boldsymbol{\Sigma}}\mathbf{J}^T\right),$$
where $\mathbf{V}=\mathbf{J}\widetilde{\boldsymbol{\Sigma}}\mathbf{J}^T\preceq\bar{C}\mathbf{I}_2$ for some absolute constant $\bar{C}$. So Theorem \ref{thm:one-sample:asymptotics} follows from Lemma \ref{lemma:bias:lambda:one-sample} that
\begin{equation}
\sqrt{N}\left(\boldsymbol{\lambda}_{\mathrm{o}}^*-\left\{-\frac{1}{2}\boldsymbol{\mu}_{m,2}^{-1}\boldsymbol{\mu}_{m,1}\right\}\right)=\begin{pmatrix}
b_1\\
b_2
\end{pmatrix},
\end{equation}
where $b_1=\mathrm{o}_p(1)$ and $b_2=\mathrm{o}_p(1)$. 
\end{proof}

\subsubsection{Technical Lemmas for the One-Sample Asymptotic Normality Result}
\begin{lemma}\label{lemma:CLT:one-sample} Let $\mathrm{vec}(\mathbf{C}_{N,2})$ be the vector with the four column-wise entries stacked, $\boldsymbol{\mu}_{m,1}:=\mathbb{E}[\mathbf{C}_{N,1}]$ and $\boldsymbol{\mu}_{m,2}:=\mathbb{E}[\mathbf{C}_{N,2}]$. Suppose Assumption \ref{assump:one_sample} holds, then 
\begin{equation}\label{eq:CLT:main}
\sqrt{N}\begin{pmatrix}
\mathbf{C}_{N,1}-\boldsymbol{\mu}_{m,1}\\
\mathrm{vec}(\mathbf{C}_{N,2}-\boldsymbol{\mu}_{m,2})
\end{pmatrix}\rightsquigarrow\mathcal{N}\left(\mathbf{0},\widetilde{\boldsymbol{\Sigma}}\right),
\end{equation}
where $\widetilde{\boldsymbol{\Sigma}}\preceq\bar{c}\mathbf{I}_6$ for some absolute constant $\bar{c}$, and $\mathbf{I}_6$ is the $6$-by-$6$ identity matrix.
\end{lemma}
\begin{proof}[Proof of Lemma \ref{lemma:CLT:one-sample}]
For notational convenience, rewriting
\begin{equation}\label{eq:C:vector}
\begin{bmatrix}
\mathbf{C}_{N,1}\\
\mathrm{Vec}(\mathbf{C}_{N,1})
\end{bmatrix}=\frac{1}{N}\sum_{i=1}^N\boldsymbol{\xi}_i,
\end{equation}
where $\boldsymbol{\xi}_i\in\mathbb{R}^6$ and is the $i$-th summand vector related to $\begin{bmatrix}
\mathbf{C}_{N,1}\\
\mathrm{Vec}(\mathbf{C}_{N,1})
\end{bmatrix}$. Let $\mathbf{Z}_i=(n_i,Y_i,X_i)$, which are independent across $i$ but not necessarily identically distributed (i.n.i.d.). Denote 
$$p_i:=\frac{Y_i}{n_i},\quad v_i:=\frac{p_i(1-p_i)}{n_i-1},$$
and we denote the population-level constants as 
$$\bar{Y}^*:=\frac{\frac{1}{N}\sum_{i=1}^N\mathbb{E}[Y_i]}{\frac{1}{N}\sum_{i=1}^Nn_i},\quad  \bar{\theta}^*:=\frac{\frac{1}{N}\sum_{i=1}^Nn_i\theta_i}{\frac{1}{N}\sum_{i=1}^Nn_i}.$$
Denote the oracle vector $\tilde{\boldsymbol{\xi}}_i$ by freezing global average quantities at population-level values and removing the cross-fitted estimators $\hat{g}_i, \hat{g}$, where 
\begin{equation}\label{eq:xi:oracle}
\tilde{\boldsymbol{\xi}}_i:=\begin{pmatrix}
\displaystyle2\left\{v_i-\left(\frac{Y_i}{n_i}-\bar{Y}^*\right)^2\right\}\\
\displaystyle2\left(\bar{Y}^*-\mathbb{I}_i^{+}\right)\left(g_i-\bar{g}\right)\\
\displaystyle\left(\frac{Y_i}{n_i}-\bar{Y}^*\right)^2\\
\displaystyle\left(\frac{Y_i}{n_i}-\bar{Y}^*\right)(g_i-\bar{g})\\
\displaystyle\left(\frac{Y_i}{n_i}-\bar{Y}^*\right)(g_i-\bar{g})\\
\displaystyle(g_i-\bar{g})^2\\
\end{pmatrix}.
\end{equation}
Then $\tilde{\boldsymbol{\xi}}_i$ are independent across $i\in[N]$, and 
\begin{equation}\label{eq:CLT:one-sampl:decompose}
\frac{1}{N}\sum_{i=1}^N\left\{\boldsymbol{\xi}_i-\mathbb{E}[\boldsymbol{\xi}_i]\right\}=\mathbf{\Delta}_N-\mathbb{E}[\mathbf{\Delta}_N]+\frac{1}{N}\sum_{i=1}^N\left\{\tilde{\boldsymbol{\xi}}_i-\mathbb{E}[\tilde{\boldsymbol{\xi}}_i]\right\},
\end{equation}
where $\mathbf{\Delta}_{N}=(\Delta_{N1},\Delta_{N2},\Delta_{N3},\Delta_{N4},\Delta_{N5},\Delta_{N6})^T$, and 
\begin{equation}\label{eq:Delta:N:matrix}
\begin{array}{rcl}
\Delta_{N1}\!\!\!\!&=&\!\!\!\!\displaystyle(\bar{Y}-\bar{Y}^*)\left[\frac{4}{N}\sum_{i=1}^N\frac{Y_i}{n_i}-2(\bar{Y}+\bar{Y}^*)\right]\\
\\
\Delta_{N2}\!\!\!\!&=&\!\!\!\!\displaystyle2(\bar{Y}-\bar{Y}^*)\frac{1}{N}\sum_{i=1}^N(g_i-\bar{g})+\frac{2}{N}\sum_{i=1}^N\left(\bar{Y}-\frac{Y_i}{n_i}\right)(\hat{g}_i-g_i+\bar{g}-\hat{g})\\
\\
\Delta_{N3}\!\!\!\!&=&\!\!\!\!\displaystyle(\bar{Y}^*-\bar{Y})\left[\frac{2}{N}\sum_{i=1}^N\frac{Y_i}{n_i}-(\bar{Y}+\bar{Y}^*)\right]\\
\\
\Delta_{N4}\!\!\!\!&=&\!\!\!\!\displaystyle(\bar{Y}^*-\bar{Y})\frac{1}{N}\sum_{i=1}^N(g_i-\bar{g})-\bar{Y}\frac{1}{N}\sum_{i=1}^N(\hat{g}_i-g_i+\bar{g}-\hat{g})\\
\\
\Delta_{N5}\!\!\!\!&=&\!\!\!\!\displaystyle(\bar{Y}^*-\bar{Y})\frac{1}{N}\sum_{i=1}^N(g_i-\bar{g})-\bar{Y}\frac{1}{N}\sum_{i=1}^N(\hat{g}_i-g_i+\bar{g}-\hat{g})\\
\\
\Delta_{N6}\!\!\!\!&=&\!\!\!\!\displaystyle\frac{1}{N}\sum_{i=1}^N\{(\hat{g}_i-\hat{g})^2-(g_i-\bar{g})^2\}
\end{array}
\end{equation}
Firstly, note that 
\begin{equation}\label{eq:Y:bar:asymptotics}
\sqrt{N}\left(\bar{Y}-\bar{Y}^*\right)=\frac{\frac{1}{\sqrt{N}}\sum_{i=1}^N(Y_i-\mathbb{E}[Y_i])}{\frac{1}{N}\sum_{i=1}^Nn_i}.
\end{equation}
Note that under (iii) of Assumption \ref{assump:one_sample}, $\frac{1}{N}\sum_{i=1}^N\mathrm{Var}(Y_i)\rightarrow\sigma_Y^2$. Under (iii) of Assumption \ref{assump:one_sample}, $\frac{1}{N}\sum_{i=1}^Nn_i\rightarrow\mu_n^*$. Further note that $\{Y_i\}_{i\in[N]}$ are independent across $i\in[N]$. Hence according to \eqref{eq:Y:bar:asymptotics},  using Slutsky's theorem and Lindeberg-Feller \textit{Central Limit Theorem} (Lemma \ref{lemma:Lindeberg-Feller:CLT:multivariate}), we have 
\begin{equation}\label{eq:Y:bar:normality}
  \sqrt{N}(\bar{Y}-\bar{Y}^*)\rightsquigarrow\mathcal{N}(0,\bar{\sigma}_Y^2),  
\end{equation}
where $\bar{\sigma}_Y$ is a constant. 

\noindent Secondly, note that 
\begin{equation}\label{eq:g:diff:asymptotics}
\begin{array}{rl}
&\displaystyle\quad\frac{1}{\sqrt{N}}\sum_{i=1}^N\{(\hat{g}_i-g_i+\bar{g}-\hat{g})-\mathbb{E}[\hat{g}_i-g_i+\bar{g}-\hat{g}]\}\\
&\displaystyle=-\sqrt{N}\left[(\hat{g}-\bar{g})-\mathbb{E}(\hat{g}-\bar{g})\right]+\frac{1}{\sqrt{N}}\sum_{i=1}^N\left\{(\hat{g}(\mathbf{X}_i)-g(\mathbf{X}_i))-\mathbb{E}[(\hat{g}(\mathbf{X}_i)-g(\mathbf{X}_i))]\right\},
\end{array}
\end{equation}
where 
\begin{equation}\label{eq:g:hat:diff}
\begin{array}{rl}
\sqrt{N}\left(\hat{g}-\bar{g}\right)\!\!\!\!&\displaystyle=\frac{\sum_{i=1}^N\sqrt{N}n_i\{\hat{g}(\mathbf{X}_i)-g(\mathbf{X}_i)\}}{\sum_{i=1}^Nn_i}\\
&\displaystyle=\frac{\sum_{k=1}^K\sum_{i\in\textrm{Fold}(k)}n_i\sqrt{N}\{\hat{g}^{-k}(\mathbf{X}_i)-g(\mathbf{X}_i)\}}{\sum_{i=1}^Nn_i}\\
&\displaystyle=\frac{(N/K)\sum_{k=1}^K\frac{1}{(N/K)}\sum_{i\in\textrm{Fold}(k)}n_i\sqrt{N}\{\hat{g}^{-k}(\mathbf{X}_i)-g(\mathbf{X}_i)\}}{\sum_{i=1}^Nn_i}\\
&\displaystyle=\frac{\sum_{k=1}^K\frac{1}{\sqrt{N}}\sum_{i\in\textrm{Fold}(k)}n_i\{\hat{g}^{-k}(\mathbf{X}_i)-g(\mathbf{X}_i)\}}{\frac{1}{N}\sum_{i=1}^Nn_i}.
\end{array}
\end{equation}
Note that 
\begin{equation}\label{eq:diff:expectation:g}
\sqrt{N}\mathbb{E}\left[(\hat{g}-\bar{g})\right]-\frac{\sqrt{N}\sum_{i=1}^N\mathbb{E}[n_i\{\hat{g}(\mathbf{X}_i)-g(\mathbf{X}_i)\}]}{\sum_{i=1}^Nn_i}=0,
\end{equation}
Hence \eqref{eq:g:hat:diff} and \eqref{eq:diff:expectation:g} imply that

\begin{equation}\label{eq:g:hat:diff:centered}
\begin{array}{rl}
&\quad\sqrt{N}\left[(\hat{g}-\bar{g})-\mathbb{E}(\hat{g}-\bar{g})\right]\\
&\displaystyle=\frac{\displaystyle\sum_{k=1}^K\frac{1}{\sqrt{N}}\sum_{i\in\textrm{Fold}(k)}\{n_i(\hat{g}^{-k}(\mathbf{X}_i)-g(\mathbf{X}_i))-\mathbb{E}[n_i(\hat{g}^{-k}(\mathbf{X}_i)-g(\mathbf{X}_i))]\}}{\displaystyle\frac{1}{N}\sum_{i=1}^Nn_i}.
\end{array}
\end{equation}
Note that within each fold the samples are independent, hence 
$$\begin{array}{rl}
\displaystyle\mathrm{Var}\left[\frac{1}{\sqrt{N}}\sum_{i\in\textrm{Fold}(k)}n_i\{\hat{g}^{-k}(\mathbf{X}_i)-g(\mathbf{X}_i)\}\right]=\frac{1}{N}\sum_{i\in\textrm{Fold}(k)}\mathrm{Var}[n_i\{\hat{g}^{-k}(\mathbf{X}_i)-g(\mathbf{X}_i)\}],
\end{array}$$
where for each $i\in\textrm{Fold}(k)$, we have 
$$\begin{array}{rl}
\displaystyle\mathrm{Var}\left[n_i\{\hat{g}^{-k}(\mathbf{X}_i)-g(\mathbf{X}_i)\}\right]\!\!\!\!&\displaystyle\leq\mathbb{E}\left[n_i^2\{\hat{g}^{-k}(\mathbf{X}_i)-g(\mathbf{X}_i)\}^2\right]\\
&\displaystyle\leq\bar{n}^2\mathbb{E}\left[\{\hat{g}^{-k}(\mathbf{X}_i)-g(\mathbf{X}_i)\}^2\right]=_{(1)}\mathrm{o}(1),
\end{array}$$
where (1) follows from Assumption \ref{assump:one_sample}. Thus 
\begin{equation}\label{eq:g_hat:consistency}
\frac{1}{\sqrt{N}}\sum_{i\in\textrm{Fold}(k)}\{n_i(\hat{g}^{-k}(\mathbf{X}_i)-g(\mathbf{X}_i))-\mathbb{E}[n_i(\hat{g}^{-k}(\mathbf{X}_i)-g(\mathbf{X}_i))]\}=\mathrm{o}_p(1).
\end{equation}
Recall that $\displaystyle\frac{1}{N}\sum_{i=1}^Nn_i\overset{p}{\rightarrow}\mu_n^*$, hence according to \eqref{eq:g:hat:diff} and \eqref{eq:g:hat:diff:centered}, we have 
\begin{equation}\label{eq:g:hat:asymptotics}
    \sqrt{N}\{(\hat{g}-\bar{g})-\mathbb{E}(\hat{g}-\bar{g})\}=\mathrm{o}_p(1).
\end{equation}
Thirdly, note that 
\begin{equation}\label{eq:dif:g:hat:o(1)}
\begin{array}{rl}
&\quad\displaystyle\frac{1}{\sqrt{N}}\sum_{i=1}^N\{\hat{g}(\mathbf{X}_i)-g(\mathbf{X}_i)-\mathbb{E}[\hat{g}(\mathbf{X}_i)-g(\mathbf{X}_i)]\}\\
&\displaystyle=\sum_{k=1}^K\frac{1}{\sqrt{N}}\sum_{i\in\textrm{Fold}(k)}\{\hat{g}^{-k}(\mathbf{X}_i)-g(\mathbf{X}_i)-\mathbb{E}[\hat{g}^{-k}(\mathbf{X}_i)-g(\mathbf{X}_i)]\}=_{(b)}\mathrm{o}_p(1),
\end{array}
\end{equation}
where (b) holds following similar proving steps for \eqref{eq:g_hat:consistency}. So according to \eqref{eq:g:diff:asymptotics}, \eqref{eq:g:hat:asymptotics} and \eqref{eq:dif:g:hat:o(1)} we have 
\begin{equation}\label{eq:g:bar:normality}
    \frac{1}{\sqrt{N}}\sum_{i=1}^N\{(\hat{g}_i-g_i+\bar{g}-\hat{g})-\mathbb{E}(\hat{g}_i-g_i+\bar{g}-\hat{g})\}=\mathrm{o}_p(1)
\end{equation}
Fourthly, note that 
\begin{equation}\label{eq:Y:g:consistency}
\begin{array}{rl}
&\quad\displaystyle\sqrt{N}\frac{2}{N}\sum_{i=1}^N\left(\bar{Y}-\frac{Y_i}{n_i}\right)(\hat{g}_i-g_i+\bar{g}-\hat{g})\\
&\displaystyle=2\bar{Y}\frac{1}{\sqrt{N}}\sum_{i=1}^N(\hat{g}_i-g_i+\bar{g}-\hat{g})-\frac{2}{\sqrt{N}}\sum_{i=1}^N\frac{Y_i}{n_i}(\hat{g}_i-g_i+\bar{g}-\hat{g})\\
&\displaystyle=2\bar{Y}\frac{1}{\sqrt{N}}\sum_{i=1}^N(\hat{g}_i-g_i+\bar{g}-\hat{g})-\frac{2}{\sqrt{N}}\sum_{i=1}^N\frac{Y_i}{n_i}(\hat{g}_i-g_i)+2\sqrt{N}(\hat{g}-\bar{g})\frac{1}{N}\sum_{i=1}^N\frac{Y_i}{n_i}\\
&\displaystyle=2\sqrt{N}(\bar{Y}-\bar{Y}^*)\frac{1}{N}\sum_{i=1}^N(\hat{g}_i-g_i+\bar{g}-\hat{g})+2\bar{Y}^*\frac{1}{\sqrt{N}}\sum_{i=1}^N(\hat{g}_i-g_i+\bar{g}-\hat{g})\\
&\displaystyle\quad-\frac{2}{\sqrt{N}}\sum_{i=1}^N\frac{Y_i}{n_i}(\hat{g}_i-g_i)+2\sqrt{N}(\hat{g}-\bar{g})\frac{1}{N}\sum_{i=1}^N\frac{Y_i}{n_i}.
\end{array}
\end{equation}
Note that $\sqrt{N}(\bar{Y}-\bar{Y}^*)=\mathrm{O}_p(1)$ according to \eqref{eq:Y:bar:normality}, and according to (ii) of Assumption \ref{assump:one_sample} we have $\displaystyle\frac{1}{N}\sum_{i=1}^N(\hat{g}_i-g_i+\bar{g}-\hat{g})=\mathrm{o}_p(1)$, hence 
\begin{equation}\label{eq:vanish:cross}
    2\sqrt{N}(\bar{Y}-\bar{Y}^*)\frac{1}{N}\sum_{i=1}^N(\hat{g}_i-g_i+\bar{g}-\hat{g})=\mathrm{o}_p(1).
\end{equation}
Note that $\displaystyle\sqrt{N}(\bar{Y}-\bar{Y}^*)\frac{1}{N}\sum_{i=1}^N(\hat{g}_i-g_i+\bar{g}-\hat{g})$ is uniformly integrable by noting that 
$$|\hat{g}_i-g_i+\bar{g}-\hat{g}|\leq2$$ 
and 
$$\sup_{N}\mathbb{E}[|\sqrt{N}(\bar{Y}-\bar{Y}^*)|]\leq\sup_{N}\mathbb{E}[|\sqrt{N}(\bar{Y}-\bar{Y}^*)|^2]^{1/2}<\infty,$$ 
so \eqref{eq:vanish:cross} also implies that 
\begin{equation}\label{eq:expectation:vanish:cross}
    \mathbb{E}\left[\sqrt{N}(\bar{Y}-\bar{Y}^*)\frac{1}{N}\sum_{i=1}^N(\hat{g}_i-g_i+\bar{g}-\hat{g})\right]=\mathrm{o}(1)
\end{equation}
Additionally, note that 
\begin{equation}\label{eq:Y:g:cross:centered}
\begin{array}{rl}
&\displaystyle\quad\frac{1}{\sqrt{N}}\sum_{i=1}^N\left\{\frac{Y_i}{n_i}(\hat{g}(\mathbf{X}_i)-g(\mathbf{X}_i))-\mathbb{E}\left[\frac{Y_i}{n_i}(\hat{g}(\mathbf{X}_i)-g(\mathbf{X}_i))\right]\right\}\\
&\displaystyle=\sum_{k=1}^K\frac{1}{\sqrt{N}}\sum_{i\in\textrm{Fold}(k)}\left\{\frac{Y_i}{n_i}(\hat{g}^{-k}(\mathbf{X}_i)-g(\mathbf{X}_i))-\mathbb{E}\left[\frac{Y_i}{n_i}(\hat{g}^{-k}(\mathbf{X}_i)-g(\mathbf{X}_i))\right]\right\},
\end{array}
\end{equation}
where for each $k\in[K]$, $\displaystyle\frac{Y_i}{n_i}(\hat{g}^{-k}(\mathbf{X}_i)-g(\mathbf{X}_i))$ are independent across $i\in\textrm{Fold}(k)$, so 
\begin{equation}\label{eq:Y:g:cross:var}
\begin{array}{rl}
&\displaystyle\quad\mathrm{Var}\left[\frac{1}{\sqrt{N}}\sum_{i\in\textrm{Fold}(k)}\left\{\frac{Y_i}{n_i}(\hat{g}^{-k}(\mathbf{X}_i)-g(\mathbf{X}_i))-\mathbb{E}\left[\frac{Y_i}{n_i}(\hat{g}^{-k}(\mathbf{X}_i)-g(\mathbf{X}_i))\right]\right\}\right]\\
&\displaystyle=\frac{1}{N}\sum_{i\in\textrm{Fold}(k)}\mathrm{Var}\left\{\frac{Y_i}{n_i}(\hat{g}^{-k}(\mathbf{X}_i)-g(\mathbf{X}_i))\right\}\leq\frac{1}{N}\sum_{i\in\textrm{Fold}(k)}\mathbb{E}\left[\left|\frac{Y_i}{n_i}(\hat{g}^{-k}(\mathbf{X}_i)-g(\mathbf{X}_i))\right|^2\right]\\
&\displaystyle\leq\frac{1}{N}\sum_{i\in\textrm{Fold}(k)}\mathbb{E}\left[\left|\hat{g}^{-k}(\mathbf{X}_i)-g(\mathbf{X}_i)\right|^2\right]=\mathrm{o}(1).
\end{array}
\end{equation}
where the last inequality of \eqref{eq:Y:g:cross:var} uses the fact that $0\leq Y_i/n_i\leq 1$. Hence 
\begin{equation}\label{eq:Y:g:vanish:2}
    -\frac{2}{\sqrt{N}}\sum_{i=1}^N\left\{\frac{Y_i}{n_i}(\hat{g}(\mathbf{X}_i)-g(\mathbf{X}_i))-\mathbb{E}\left[\frac{Y_i}{n_i}(\hat{g}(\mathbf{X}_i)-g(\mathbf{X}_i))\right]\right\}=\mathrm{o}_p(1).
\end{equation}
Further, note that 
\begin{equation}\label{eq:asymptotics:g:Y}
\begin{array}{rl}
&\quad\displaystyle\sqrt{N}(\hat{g}-\bar{g})\frac{1}{N}\sum_{i=1}^N\frac{Y_i}{n_i}\\
&=\displaystyle\sqrt{N}(\hat{g}-\bar{g})\frac{1}{N}\sum_{i=1}^N\left(\frac{Y_i}{n_i}-\theta_i\right)+\left\{\frac{1}{N}\sum_{i=1}^N\theta_i\right\}\sqrt{N}(\hat{g}-\bar{g})=\mu_{\theta}^*\sqrt{N}(\hat{g}-\bar{g})+\mathrm{o}_p(1),
\end{array}
\end{equation}
where the last equality above follows since $\displaystyle\sqrt{N}(\hat{g}-\bar{g})\frac{1}{N}\sum_{i=1}^N\left(\frac{Y_i}{n_i}-\theta_i\right)=\mathrm{o}_p(1)$ and $\displaystyle\frac{1}{N}\sum_{i=1}^N\theta_i\rightarrow\mu_{\theta}^*$ according to (iii) of Assumption \ref{assump:one_sample}. Also note that $|\hat{g}-\bar{g}|\leq1$ and $$\sup_N\mathbb{E}\left[\left|\frac{1}{\sqrt{N}}\sum_{i=1}^N\left(\frac{Y_i}{n_i}-\theta_i\right)\right|\right]\leq\sup_N\mathbb{E}\left[\left|\frac{1}{\sqrt{N}}\sum_{i=1}^N\left(\frac{Y_i}{n_i}-\theta_i\right)\right|^2\right]^{1/2}<\infty,$$ thus $\displaystyle\sqrt{N}(\hat{g}-\bar{g})\frac{1}{N}\sum_{i=1}^N\left(\frac{Y_i}{n_i}-\theta_i\right)$ is uniformly integrable, so 
$$\mathbb{E}\left[\sqrt{N}(\hat{g}-\bar{g})\frac{1}{N}\sum_{i=1}^N\left(\frac{Y_i}{n_i}-\theta_i\right)\right]=\mathrm{o}(1).$$ 
Thus according to \eqref{eq:g:hat:asymptotics} and \eqref{eq:asymptotics:g:Y} we have 
\begin{equation}\label{eq:asymptotics:g:Y:vanish}
\begin{array}{rl}
&\displaystyle\quad\sqrt{N}(\hat{g}-\bar{g})\frac{1}{N}\sum_{i=1}^N\frac{Y_i}{n_i}-\mathbb{E}\left[\sqrt{N}(\hat{g}-\bar{g})\frac{1}{N}\sum_{i=1}^N\frac{Y_i}{n_i}\right]\\
&\displaystyle=\mu_{\theta}^*\left[\sqrt{N}(\hat{g}-\bar{g})-\mathbb{E}\{\sqrt{N}(\hat{g}-\bar{g})\}\right]+\mathrm{o}_p(1)\\
&\displaystyle=\mathrm{o}_p(1).
\end{array}
\end{equation}
Hence according to \eqref{eq:Y:g:consistency}, \eqref{eq:vanish:cross}, \eqref{eq:expectation:vanish:cross}, \eqref{eq:Y:g:vanish:2}, \eqref{eq:asymptotics:g:Y:vanish}, \eqref{eq:g:bar:normality}, and following (iii) of Assumption \ref{assump:one_sample} we have $\displaystyle\bar{Y}^*\rightarrow\frac{\mu_Y^*}{\mu_n^*}$ as $N\rightarrow\infty$, so we have 
\begin{equation}\label{eq:vanish:cross:Y:g}
    \frac{2}{\sqrt{N}}\sum_{i=1}^N\left(\bar{Y}-\frac{Y_i}{n_i}\right)(\hat{g}_i-g_i+\bar{g}-\hat{g})-\mathbb{E}\left[\frac{2}{\sqrt{N}}\sum_{i=1}^N\left(\bar{Y}-\frac{Y_i}{n_i}\right)(\hat{g}_i-g_i+\bar{g}-\hat{g})\right]=\mathrm{o}_p(1).
\end{equation}
Lastly, note that 
\begin{equation}\label{eq:g:square:diff}
\begin{array}{rl}
    &\quad\displaystyle\frac{1}{\sqrt{N}}\sum_{i=1}^N\{(\hat{g}_i-\hat{g})^2-(g_i-\bar{g})^2\}\\
    &\displaystyle=\frac{1}{\sqrt{N}}\sum_{i=1}^N[(\hat{g}_i-\hat{g})+(g_i-\bar{g})][\hat{g}_i-\hat{g}+\bar{g}-g_i]\\
    &\displaystyle=\frac{1}{\sqrt{N}}\sum_{i=1}^N(\hat{g}_i+g_i)[\hat{g}_i-\hat{g}+\bar{g}-g_i]-(\hat{g}+\bar{g})\frac{1}{\sqrt{N}}\sum_{i=1}^N[\hat{g}_i-\hat{g}+\bar{g}-g_i]\\
    &\displaystyle=\frac{1}{\sqrt{N}}\sum_{i=1}^N(\hat{g}_i+g_i-2\bar{g})[\hat{g}_i-\hat{g}+\bar{g}-g_i]+(\bar{g}-\hat{g})\frac{1}{\sqrt{N}}\sum_{i=1}^N[\hat{g}_i-\hat{g}+\bar{g}-g_i]\\
    &\displaystyle=\sum_{k=1}^K\frac{1}{\sqrt{N}}\sum_{i\in\textrm{Fold}(k)}(\hat{g}_i+g_i-2\bar{g})[\hat{g}_i-\hat{g}+\bar{g}-g_i]\\
    &\displaystyle\quad\quad\quad+(\bar{g}-\hat{g})\sum_{k=1}^K\frac{1}{\sqrt{N}}\sum_{i\in\textrm{Fold}(k)}\sum_{i=1}^N[\hat{g}_i-\hat{g}+\bar{g}-g_i]
\end{array}
\end{equation}
Following similar proof steps as those for \eqref{eq:Y:g:cross:centered}, \eqref{eq:Y:g:cross:var} and \eqref{eq:Y:g:vanish:2}, we have 
$$\mathrm{Var}\left[\frac{1}{\sqrt{N}}\sum_{i\in\textrm{Fold}(k)}(\hat{g}_i+g_i-2\bar{g})[\hat{g}_i-\hat{g}+\bar{g}-g_i]\right]=\mathrm{o}(1),$$
and note that $|\bar{g}-\hat{g}|\leq1$, so 
$$\mathrm{Var}\left[(\bar{g}-\hat{g})\sum_{k=1}^K\frac{1}{\sqrt{N}}\sum_{i\in\textrm{Fold}(k)}\sum_{i=1}^N[\hat{g}_i-\hat{g}+\bar{g}-g_i]\right]=\mathrm{o}(1),$$
thus \eqref{eq:g:square:diff} implies that 
\begin{equation}\label{eq:square:vanish}
\begin{array}{rl}
\displaystyle\mathrm{o}_p(1)\!\!\!\!&\displaystyle=\frac{1}{\sqrt{N}}\sum_{i=1}^N[(\hat{g}_i-\hat{g})+(g_i-\bar{g})][\hat{g}_i-\hat{g}+\bar{g}-g_i]\\
&\quad\quad\quad\quad\quad\displaystyle-\mathbb{E}\left[\frac{1}{\sqrt{N}}\sum_{i=1}^N\{(\hat{g}_i-\hat{g})+(g_i-\bar{g})\}\{\hat{g}_i-\hat{g}+\bar{g}-g_i\}\right].
\end{array}
\end{equation}
Hence, \eqref{eq:g:square:diff} and \eqref{eq:square:vanish} imply that 
\begin{equation}\label{eq:square:CLT}
    \frac{1}{\sqrt{N}}\sum_{i=1}^N\{(\hat{g}_i-\hat{g})^2-(g_i-\bar{g})^2\}-\mathbb{E}\left[(\hat{g}_i-\hat{g})^2-(g_i-\bar{g})^2\right]=\mathrm{o}_p(1).
\end{equation}
Recall from \eqref{eq:C:vector} and \eqref{eq:CLT:one-sampl:decompose} that 
$$\begin{bmatrix}
\mathbf{C}_{N,1}\\
\mathrm{Vec}(\mathbf{C}_{N,1})
\end{bmatrix}-\begin{bmatrix}
\mathbb{E}[\mathbf{C}_{N,1}]\\
\mathbb{E}\left[\mathrm{Vec}(\mathbf{C}_{N,1})\right]
\end{bmatrix}=\mathbf{\Delta}_N-\mathbb{E}[\mathbf{\Delta}_N]+\frac{1}{N}\sum_{i=1}^N\left\{\tilde{\boldsymbol{\xi}}_i-\mathbb{E}[\tilde{\boldsymbol{\xi}}_i]\right\},$$
where $\tilde{\boldsymbol{\xi}}_i$ are i.n.i.d. across $i\in[N]$. Given any $\mathbf{t}=(t_1,t_2,t_3,t_4,t_5,t_6)\in\mathbb{R}^6$, according to \eqref{eq:CLT:one-sampl:decompose} and \eqref{eq:Delta:N:matrix}, we have
$$\begin{array}{rl}
&\quad\displaystyle\mathbf{t}^T\left\{\sqrt{N}\left\{\mathbf{\Delta}_N-\mathbb{E}[\mathbf{\Delta}_N]\right\}+\frac{1}{\sqrt{N}}\sum_{i=1}^N\left\{\tilde{\boldsymbol{\xi}}_i-\mathbb{E}[\tilde{\boldsymbol{\xi}}_i]\right\}\right\}\\
&\displaystyle=\sqrt{N}(\bar{Y}-\bar{Y}^*)\Bigg\{t_1\left[\frac{4}{N}\sum_{i=1}^N\frac{Y_i}{n_i}-2(\bar{Y}+\bar{Y}^*)\right]+\frac{2t_2}{N}\sum_{i=1}^N(g_i-\bar{g})\\
&\displaystyle\quad\quad\quad\quad\quad\quad\quad\ +t_3\left[\frac{2}{N}\sum_{i=1}^N\frac{Y_i}{n_i}-(\bar{Y}+\bar{Y}^*)\right]+\frac{t_4+t_5}{N}\sum_{i=1}^N(g_i-\bar{g})\Bigg\}\\
&\displaystyle\quad+t_2\frac{2}{\sqrt{N}}\sum_{i=1}^N\left(\bar{Y}-\frac{Y_i}{n_i}\right)(\hat{g}_i-g_i+\bar{g}-\hat{g})-2\sqrt{N}(\bar{g}-\bar{g})t_2\left[\bar{Y}^*-\frac{1}{N}\sum_{i=1}^N\mathbb{I}_i^{+}\right]\\
&\displaystyle\quad-\bar{Y}(t_4+t_5)\frac{1}{\sqrt{N}}\sum_{i=1}^N(\hat{g}_i-g_i+\bar{g}-\hat{g})+\frac{t_6}{\sqrt{N}}\sum_{i=1}^N\{(\hat{g}_i-\hat{g})^2-(g_i-\bar{g})^2\}\\
&\displaystyle\quad+\frac{1}{\sqrt{N}}\sum_{i=1}^N\mathbf{t}^T\left\{\tilde{\boldsymbol{\xi}}_i-\mathbb{E}[\tilde{\boldsymbol{\xi}}_i]\right\}-\mathbf{t}^T\mathbb{E}[\sqrt{N}\mathbf{\Delta}_N]+\mathrm{o}_p(1).
\end{array}$$
Using (iii) of Assumption \ref{assump:one_sample}, Law of Large Numbers,  Slutsky's theorem, we get
$$\displaystyle\frac{1}{N}\sum_{i=1}^N\frac{Y_i}{n_i}\overset{p}{\rightarrow}\mu_{\theta}^*, \ \displaystyle\bar{Y}\overset{p}{\rightarrow}\frac{\mu_Y^*}{\mu_n^*}, \ \displaystyle\bar{Y}^*\overset{p}{\rightarrow}\frac{\mu_Y^*}{\mu_n^*},\  \displaystyle\frac{1}{N}\sum_{i=1}^N(g_i-\bar{g})\overset{p}{\rightarrow}\mu_g^*-\frac{\mu_{gn}^*}{\mu_n^*},\ \displaystyle\frac{1}{N}\sum_{i=1}^N\mathbb{I}_i^{+}\overset{p}{\rightarrow}\mu_I^*.$$ 
Then we have 
\begin{equation}\label{eq:Z:N:equiv}
Z_N(t)-\mathbf{t}^T\left\{\sqrt{N}\left\{\mathbf{\Delta}_N-\mathbb{E}[\mathbf{\Delta}_N]\right\}+\frac{1}{\sqrt{N}}\sum_{i=1}^N\left\{\tilde{\boldsymbol{\xi}}_i-\mathbb{E}[\tilde{\boldsymbol{\xi}}_i]\right\}\right\}=\mathrm{o}_p(1),
\end{equation}
where 
\begin{equation}\label{eq:Z:N:t}
\begin{array}{rl}
Z_N(t)\!\!\!\!&\displaystyle=\sqrt{N}(\bar{Y}\!-\!\bar{Y}^*)\left\{\left(4\mu_{\theta}^*\!-\!4\frac{\mu_Y^*}{\mu_n^*}\right)\!t_1\!+\!\left(2\mu_{\theta}^*\!-\!2\frac{\mu_Y^*}{\mu_n^*}\right)\!t_3\!+\!(2t_2\!+\!t_4\!+\!t_5)\!\left(\mu_g^*\!-\!\frac{\mu_{gn}^*}{\mu_n^*}\right)\right\}\\
&\displaystyle\quad-(t_4+t_5)\frac{\mu_Y^*}{\mu_n^*}\frac{1}{\sqrt{N}}\sum_{i=1}^N\{(\hat{g}_i-g_i+\bar{g}-\hat{g})-\mathbb{E}[\hat{g}_i-g_i+\bar{g}-\hat{g}]\}\\
&\displaystyle\quad+\frac{t_6}{\sqrt{N}}\sum_{i=1}^N\left(\{(\hat{g}_i-\hat{g})^2-(g_i-\bar{g})^2\}-\mathbb{E}[(\hat{g}_i-\hat{g})^2-(g_i-\bar{g})^2]\right)\\
&\displaystyle\quad+\frac{1}{\sqrt{N}}\sum_{i=1}^N\mathbf{t}^T\left\{\tilde{\boldsymbol{\xi}}_i-\mathbb{E}[\tilde{\boldsymbol{\xi}}_i]\right\}\\
&=\displaystyle\mathbf{t}^T\boldsymbol{\zeta}_N'=\mathbf{t}^T\begin{pmatrix}
    \zeta_{N1}'\\
    \zeta_{N2}'\\
    \zeta_{N3}'\\
    \zeta_{N4}'\\
    \zeta_{N5}'\\
    \zeta_{N6}'
\end{pmatrix},
\end{array}
\end{equation}
where 
$$\zeta_{N1}'=4\left(\mu_{\theta}^*-\frac{\mu_Y^*}{\mu_n^*}\right)\frac{\displaystyle\frac{1}{\sqrt{N}}\sum_{i=1}^N(Y_i-\mathbb{E}[Y_i])}{\displaystyle\frac{1}{N}\sum_{i=1}^Nn_i}+\frac{1}{\sqrt{N}}\sum_{i=1}^N\{\tilde{\boldsymbol{\xi}}_{i1}-\mathbb{E}[\tilde{\boldsymbol{\xi}}_{i1}]\},$$
$$\begin{array}{rl}
\zeta_{N2}'\!\!\!\!&\displaystyle=\frac{1}{\sqrt{N}}\sum_{i=1}^N\{\tilde{\boldsymbol{\xi}}_{i2}-\mathbb{E}[\tilde{\boldsymbol{\xi}}_{i2}]\}+2\left(\mu_g^*-\frac{\mu_{gn}^*}{\mu_n^*}\right)\frac{\displaystyle\frac{1}{\sqrt{N}}\sum_{i=1}^N(Y_i-\mathbb{E}[Y_i])}{\displaystyle\frac{1}{N}\sum_{i=1}^Nn_i},
\end{array}$$
$$\zeta_{N3}'=2\left(\mu_{\theta}^*-\frac{\mu_Y^*}{\mu_n^*}\right)\frac{\displaystyle\frac{1}{\sqrt{N}}\sum_{i=1}^N(Y_i-\mathbb{E}[Y_i])}{\displaystyle\frac{1}{N}\sum_{i=1}^Nn_i}+\frac{1}{\sqrt{N}}\sum_{i=1}^N\{\tilde{\boldsymbol{\xi}}_{i3}-\mathbb{E}[\tilde{\boldsymbol{\xi}}_{i3}]\},$$
$$\begin{array}{rl}
\zeta_{N4}'=\zeta_{N5}'\!\!\!\!&\displaystyle=\left(\mu_g^*-\frac{\mu_{gn}^*}{\mu_n^*}\right)\frac{\displaystyle\frac{1}{\sqrt{N}}\sum_{i=1}^N(Y_i-\mathbb{E}[Y_i])}{\displaystyle\frac{1}{N}\sum_{i=1}^Nn_i}\\
&\displaystyle\quad-\frac{\mu_Y^*}{\mu_n^*}\sum_{k=1}^K\frac{1}{\sqrt{N}}\!\!\sum_{i\in\textrm{Fold}(k)}\left[\left(\hat{g}^{-k}(\mathbf{X}_i)-g(\mathbf{X}_i)\right)-\mathbb{E}[\hat{g}^{-k}(\mathbf{X}_i)-g(\mathbf{X}_i)]\right]\\
&\displaystyle\quad-\frac{\mu_Y^*}{\mu_n^*}\left[\mathbb{E}\left[\frac{\frac{1}{N}\sum_{i=1}^Nn_i\hat{g}^{-k}(\mathbf{X}_i)}{\frac{1}{N}\sum_{i=1}^Nn_i}\right]-\frac{\frac{1}{N}\sum_{i=1}^Nn_i\hat{g}^{-k}(\mathbf{X}_i)}{\frac{1}{N}\sum_{i=1}^Nn_i}\right]\\
&\displaystyle\quad+\frac{1}{\sqrt{N}}\sum_{i=1}^N\{\tilde{\boldsymbol{\xi}}_{i4}-\mathbb{E}[\tilde{\boldsymbol{\xi}}_{i4}]\}.
\end{array}$$
\begin{equation}\label{eq:zeta:N:6:prime}
\zeta_{N6}'=\frac{1}{\sqrt{N}}\sum_{i=1}^N\{\tilde{\boldsymbol{\xi}}_{i6}-\mathbb{E}[\tilde{\boldsymbol{\xi}}_{i6}]\}+\mathrm{o}_p(1),
\end{equation}
where \eqref{eq:zeta:N:6:prime} follows from \eqref{eq:square:CLT}. Using (iii) of Assumption \ref{assump:one_sample}, Law of Large Numbers and Slutsky's theorem, we have 
\begin{equation}\label{eq:zeta:convergence}
\begin{array}{rcl}
&\zeta_{N1}'-\zeta_{N1}=\mathrm{o}_p(1),\ \zeta_{N2}'-\zeta_{N2}=\mathrm{o}_p(1),\ \zeta_{N3}'-\zeta_{N3}=\mathrm{o}_p(1),\ \zeta_{N4}'-\zeta_{N4}=\mathrm{o}_p(1)&\\
&\zeta_{N5}'-\zeta_{N5}=\mathrm{o}_p(1),\ \zeta_{N6}'-\zeta_{N6}=\mathrm{o}_p(1),&
\end{array}
\end{equation}
where 
\begin{equation}\label{eq:zeta:N:1}
\zeta_{N1}=\frac{1}{\sqrt{N}}\sum_{i=1}^N\left[4\frac{\left(\mu_{\theta}^*-\frac{\mu_Y^*}{\mu_n^*}\right)}{\mu_n^*}(Y_i-\mathbb{E}[Y_i])+\left\{\tilde{\boldsymbol{\xi}}_{i1}-\mathbb{E}[\tilde{\boldsymbol{\xi}}_{i1}]\right\}\right],
\end{equation}
\begin{equation}\label{eq:zeta:N:2}
\begin{array}{rl}
\displaystyle\zeta_{N2}=\frac{1}{\sqrt{N}}\sum_{i=1}^N\Bigg[\!\!\!\!&\displaystyle2\frac{\left(\mu_g^*-\frac{\mu_{gn}^*}{\mu_n^*}\right)}{\mu_n^*}(Y_i-\mathbb{E}[Y_i])+\{\tilde{\boldsymbol{\xi}}_{i2}-\mathbb{E}[\tilde{\boldsymbol{\xi}}_{i2}]\}\Bigg],
\end{array}
\end{equation}
\begin{equation}\label{eq:zeta:N:3}
\zeta_{N3}=\frac{1}{\sqrt{N}}\sum_{i=1}^N\left[2\frac{\left(\mu_{\theta}^*-\frac{\mu_Y^*}{\mu_n^*}\right)}{\mu_n^*}(Y_i-\mathbb{E}[Y_i])+\{\tilde{\boldsymbol{\xi}}_{i3}-\mathbb{E}[\tilde{\boldsymbol{\xi}}_{i3}]\}\right],    
\end{equation}
\begin{equation}\label{eq:zeta:N:4,5}
\begin{array}{rl}
\zeta_{N4}=\zeta_{N5}\!\!\!\!&\displaystyle=\frac{1}{\sqrt{N}}\sum_{i=1}^N\left[\frac{\left(\mu_g^*-\frac{\mu_{gn}^*}{\mu_n^*}\right)}{\mu_n^*}(Y_i-\mathbb{E}[Y_i])+\{\tilde{\boldsymbol{\xi}}_{i4}-\mathbb{E}[\tilde{\boldsymbol{\xi}}_{i4}]\}\right]\\
&\displaystyle\quad-\sum_{k=1}^K\frac{1}{\sqrt{N}}\sum_{i\in\textrm{Fold}(k)}\frac{\mu_Y^*}{\mu_n^*}\left\{\left(\hat{g}^{-k}(\mathbf{X}_i)-g(\mathbf{X}_i)\right)-\mathbb{E}\left[\hat{g}^{-k}(\mathbf{X}_i)-g(\mathbf{X}_i)\right]\right\},
\end{array}
\end{equation}
\begin{equation}\label{eq:zeta:N:6}
\begin{array}{rl}
\displaystyle\zeta_{N6}=\frac{1}{\sqrt{N}}\sum_{i=1}^N\{\tilde{\boldsymbol{\xi}}_{i6}-\mathbb{E}[\tilde{\boldsymbol{\xi}}_{i6}]\},
\end{array}
\end{equation}
Thus according to \eqref{eq:xi:oracle}, \eqref{eq:Z:N:equiv}, \eqref{eq:Z:N:t}, \eqref{eq:zeta:convergence}, \eqref{eq:zeta:N:1}, \eqref{eq:zeta:N:2}, \eqref{eq:zeta:N:3}, \eqref{eq:zeta:N:4,5}, \eqref{eq:zeta:N:6}, given any $\mathbf{t}\in\mathbb{R}^6$, using (iii) of Assumption \ref{assump:one_sample} and Slutsky's theorem, we have 
\begin{equation}\label{eq:Z:equiv:i.n.i.d.}
    \mathbf{t}^T\left\{\sqrt{N}\left\{\mathbf{\Delta}_N-\mathbb{E}[\mathbf{\Delta}_N]\right\}+\frac{1}{\sqrt{N}}\sum_{i=1}^N\left\{\tilde{\boldsymbol{\xi}}_i-\mathbb{E}[\tilde{\boldsymbol{\xi}}_i]\right\}\right\}-\mathbf{t}^T\widetilde{\mathbf{Z}}_N=\mathrm{o}_p(1),
\end{equation}
where 
\begin{equation}\label{eq:Z:tilde:N}
\widetilde{\mathbf{Z}}_N=\frac{1}{\sqrt{N}}\sum_{i=1}^N\{\tilde{\boldsymbol{\zeta}}_i-\mathbb{E}[\tilde{\boldsymbol{\zeta}}_i]\}+\sum_{k=1}^K\frac{1}{\sqrt{N}}\sum_{i\in\textrm{Fold}(k)}\tilde{\boldsymbol{\delta}}_i,   
\end{equation}
$$\tilde{\boldsymbol{\delta}}_i=\begin{pmatrix}
0\\
0\\
0\\
\displaystyle(\mu_Y^*/\mu_n^*)\{\hat{g}^{-k}(\mathbf{X}_i)-g(\mathbf{X}_i)-\mathbb{E}[\hat{g}^{-k}(\mathbf{X}_i)-g(\mathbf{X}_i)]\}\\
\displaystyle(\mu_Y^*/\mu_n^*)\{\hat{g}^{-k}(\mathbf{X}_i)-g(\mathbf{X}_i)-\mathbb{E}[\hat{g}^{-k}(\mathbf{X}_i)-g(\mathbf{X}_i)]\}\\
0
\end{pmatrix},$$
and 
$$\tilde{\boldsymbol{\zeta}}_i=(\tilde{\zeta}_{i1},\tilde{\zeta}_{i2},\tilde{\zeta}_{i3},\tilde{\zeta}_{i4},\tilde{\zeta}_{i5},\tilde{\zeta}_{i6})^T,$$
such that $\tilde{\boldsymbol{\zeta}}_i$ are i.n.i.d. across $i\in[N]$. According to \eqref{eq:dif:g:hat:o(1)}, we have 
\begin{equation}\label{eq:vanish:delta}
\left\|\sum_{k=1}^K\frac{1}{\sqrt{N}}\sum_{i\in\textrm{Fold}(k)}\tilde{\boldsymbol{\delta}}_i\right\|_{\infty}=\mathrm{o}_p(1).
\end{equation}
Each entry of $\tilde{\boldsymbol{\zeta}}_i$ is defined as 
$$\tilde{\zeta}_{i1}=4\frac{\left(\mu_{\theta}^*-\frac{\mu_Y^*}{\mu_n^*}\right)}{\mu_n^*}\left\{Y_i-\mathbb{E}[Y_i]\right\}+2\left\{\frac{Y_i/n_i(1-Y_i/n_i)}{n_i-1}-\left(\frac{Y_i}{n_i}-\frac{\mu_Y^*}{\mu_n^*}\right)^2\right\},$$
$$\begin{array}{rl}
\tilde{\zeta}_{i2}\!\!\!\!&\displaystyle=2\frac{\left(\mu_g^*-\frac{\mu_{gn}^*}{\mu_n^*}\right)}{\mu_n^*}(Y_i-\mathbb{E}[Y_i])+2\left(\frac{\mu_Y^*}{\mu_n^*}-\mathbf{1}\{Y_i>\lfloor n_i/2\rfloor\}\right)\left(g(\mathbf{X}_i)-\frac{\mu_{gn}^*}{\mu_n^*}\right),
\end{array}$$
$$\tilde{\zeta}_{i3}=2\frac{\left(\mu_{\theta}^*-\frac{\mu_Y^*}{\mu_n^*}\right)}{\mu_n^*}(Y_i-\mathbb{E}[Y_i])+\left(\frac{Y_i}{n_i}-\frac{\mu_Y^*}{\mu_n^*}\right)^2,$$
$$\tilde{\zeta}_{i4}=\tilde{\zeta}_{i5}=\frac{\left(\mu_g^*-\frac{\mu_{gn}^*}{\mu_n^*}\right)}{\mu_n^*}(Y_i-\mathbb{E}[Y_i])+\left(\frac{Y_i}{n_i}-\frac{\mu_Y^*}{\mu_n^*}\right)\left(g(\mathbf{X}_i)-\frac{\mu_{gn}^*}{\mu_n^*}\right),$$
$$\tilde{\zeta}_{i6}=\left(g(\mathbf{X}_i)-\frac{\mu_{gn}^*}{\mu_n^*}\right)^2.$$
According to (iv) of Assumption \ref{assump:one_sample}, we have 
$$\frac{1}{N}\sum_{i=1}^N\mathrm{Cov}\left\{\tilde{\boldsymbol{\zeta}}_i\right\}\rightarrow\widetilde{\boldsymbol{\Sigma}}.$$
Note that according to (i) of Assumption \ref{assump:one_sample}, and the fact that $Y_i/n_i, g(\mathbf{X}_i)\in[0,1]$, then $|\tilde{\zeta}_{i\ell}|\leq c_0$ for some absolute constant $c_0$, hence there must exists an absolute constant $\bar{c}$, such that 
$$\widetilde{\boldsymbol{\Sigma}}\preceq\bar{c}\mathbf{I}_6,$$
where $\mathbf{I}_6$ is the $6$-by-$6$ identity matrix. Additionally, note that 
$$\frac{1}{N}\sum_{i=1}^N\mathbb{E}\left[\left\|\tilde{\boldsymbol{\zeta}}_i-\mathbb{E}\left[\tilde{\boldsymbol{\zeta}}_i\right]\right\|_2^2\mathbf{1}\left\{\left\|\tilde{\boldsymbol{\zeta}}_i-\mathbb{E}\left[\tilde{\boldsymbol{\zeta}}_i\right]\right\|_2>\sqrt{N}\epsilon\right\}\right]\rightarrow0,\ \ \forall \epsilon>0.$$
Recall that $\tilde{\boldsymbol{\zeta}}_i$ are independent across $i\in[N]$. Hence by Lemma \ref{lemma:Lindeberg-Feller:CLT:multivariate}, we have 
\begin{equation}\label{eq:orcale:CLT}
\frac{1}{\sqrt{N}}\sum_{i=1}^N\left\{\tilde{\boldsymbol{\zeta}}_i-\mathbb{E}\left[\tilde{\boldsymbol{\zeta}}_i\right]\right\}\rightsquigarrow\mathcal{N}\left(\mathbf{0},\widetilde{\boldsymbol{\Sigma}}\right).
\end{equation}
Thus, according to \eqref{eq:C:vector}, \eqref{eq:CLT:one-sampl:decompose}, \eqref{eq:Z:equiv:i.n.i.d.}, \eqref{eq:Z:tilde:N}, \eqref{eq:vanish:delta}, \eqref{eq:orcale:CLT}, given any $\mathbf{t}\in\mathbb{R}^6$, we have 
$$\mathbf{t}^T\sqrt{N}\begin{pmatrix}
\mathbf{C}_{N,1}-\boldsymbol{\mu}_{m,1}\\
\mathrm{vec}(\mathbf{C}_{N,2}-\boldsymbol{\mu}_{m,2})
\end{pmatrix}\rightsquigarrow\mathcal{N}\left(0,\mathbf{t}^T\widetilde{\boldsymbol{\Sigma}}\mathbf{t}\right).$$
Thus \eqref{eq:CLT:main} holds using Cramer-Wold theorem (Lemma \ref{lemma:cramer-wold}).
\end{proof}

\begin{lemma}\label{lemma:bias:lambda:one-sample}
Let  $\boldsymbol{\mu}_{m,1}=\mathbb{E}[\mathbf{C}_{N,1}]$ and $\boldsymbol{\mu}_{m,2}=\mathbb{E}[\mathbf{C}_{N,2}]$. Then    $$\sqrt{N}\left(\boldsymbol{\lambda}_{\mathrm{o}}^*-\left\{-\frac{1}{2}\boldsymbol{\mu}_{m,2}^{-1}\boldsymbol{\mu}_{m,1}\right\}\right)=\begin{pmatrix}
    b_1\\
    b_2
\end{pmatrix},$$
where $b_1=\mathrm{o}_p(1)$ and $b_2=\mathrm{o}_p(1)$. 
\end{lemma}
\begin{proof}[Proof of Lemma \ref{lemma:bias:lambda:one-sample}]
Recall from \eqref{eq:two-sample:estimator:lambda:one-sample} that for $\boldsymbol{\lambda}=(\lambda_1,\lambda_2)^T$, the estimator is 
$$\hat{\theta}_i^{\mathrm{o}}(\boldsymbol{\lambda}) = \lambda_1\frac{Y_i}{n_i} + (1-\lambda_1)\frac{\sum_{i}^{N}Y_i}{\sum_{i}^{N}n_i}+\lambda_2\left(\hat{g}(\mathbf{X}_i)-\frac{\sum_{j=1}^Nn_j\hat{g}(\mathbf{X}_{j})}{\sum_{j=1}^Nn_j}\right).$$
Plugging $\hat{\theta}_i^{\mathrm{o}}(\boldsymbol{\lambda})$ into \eqref{eq:unweighted:one-sample}, recall the notations 
$$\bar{Y}=\frac{\sum_{i}^{N}Y_i}{\sum_{i}^{N}n_i},\quad\hat{g}=\frac{\sum_{j=1}^Nn_j\hat{g}(\mathbf{X}_{j})}{\sum_{j=1}^Nn_j},\quad \boldsymbol{\beta}_i=\begin{pmatrix}
\displaystyle \frac{Y_i}{n_i}-\bar{Y}\\
\displaystyle \hat{g}(\mathbf{X}_i)-\hat{g}
\end{pmatrix}.$$ 
By first-order condition, 
\begin{equation}\label{eq:lambda:star:one-sample}
\begin{array}{rl}
\boldsymbol{\lambda}_{\mathrm{o}}^*\!\!\!&\displaystyle=-\left\{\frac{1}{N}\sum_{i=1}^N\mathbb{E}\left[\boldsymbol{\beta}_i\boldsymbol{\beta}_i^T\right]\right\}^{-1}\left\{\frac{1}{N}\sum_{i=1}^N\mathbb{E}\left[\left(\bar{Y}-\theta_i^{\mathrm{o}}\right)\boldsymbol{\beta}_i\right]\right\}\\
&\displaystyle=-\mathbb{E}[\mathbf{C}_{N,2}]^{-1}\left\{\frac{1}{N}\sum_{i=1}^N\mathbb{E}\left[\left(\bar{Y}-\theta_i^{\mathrm{o}}\right)\boldsymbol{\beta}_i\right]\right\}=-\boldsymbol{\mu}_{m,2}^{-1}\left\{\frac{1}{N}\sum_{i=1}^N\mathbb{E}\left[\left(\bar{Y}-\theta_i^{\mathrm{o}}\right)\boldsymbol{\beta}_i\right]\right\},
\end{array}
\end{equation}
where the second equality above follows from \eqref{eq:coefficient:quadratic:ML}. Additionally, following \eqref{eq:coeff:linear:one-sample}, we have 
$$\begin{array}{rl}
\displaystyle\frac{1}{2}\boldsymbol{\mu}_{m,1}\!\!\!\!&\displaystyle=\frac{1}{N}\sum_{i=1}^N\mathbb{E}\left[\bar{Y}\boldsymbol{\beta}_i\right]-\frac{1}{N}\sum_{i=1}^N\mathbb{E}\left[\mathbf{1}(Y_{i} > \lfloor n_i/2\rfloor)\!\!\!\sum_{j=0}^{n_{i}-Y_{i}}\!\!\boldsymbol{\beta}_i(Y_i+j)(-1)^{j}\frac{\binom{n_i-Y_i}{j}}{\binom{Y_i+j}{j}}\right]\\
&\quad\displaystyle-\frac{1}{N}\sum_{i=1}^N\mathbb{E}\left[\boldsymbol{\beta}_i\mathbf{1}\{Y_i\leq \lfloor n_i/2\rfloor\}\right]\\
&\quad\displaystyle+\frac{1}{N}\sum_{i=1}^N\mathbb{E}\left[\mathbf{1}(Y_{i}\leq \lfloor n_i/2\rfloor)\sum_{j=0}^{Y_{i}}\boldsymbol{\beta}_i(Y_i-j)(-1)^{j}\frac{\binom{Y_i}{j}}{\binom{n_i-Y_i+j}{j}}\right].
\end{array}$$
Thus 
\begin{equation}\label{eq:bias:one-sample}
\begin{array}{rl}
&\quad\displaystyle\frac{1}{N}\sum_{i=1}^N\mathbb{E}\left[\left(\bar{Y}-\theta_i^{\mathrm{o}}\right)\boldsymbol{\beta}_i\right]-\frac{1}{2}\boldsymbol{\mu}_{m,1}\\
&\displaystyle=-\frac{1}{N}\sum_{i=1}^N\theta_i^{\mathrm{o}}\mathbb{E}\left[\boldsymbol{\beta}_i\right]+\frac{1}{N}\sum_{i=1}^N\mathbb{E}\left[\mathbf{1}(Y_{i} > \lfloor n_i/2\rfloor)\!\!\!\sum_{j=0}^{n_{i}-Y_{i}}\!\!\boldsymbol{\beta}_i(Y_i+j)(-1)^{j}\frac{\binom{n_i-Y_i}{j}}{\binom{Y_i+j}{j}}\right]\\
&\displaystyle\quad+\frac{1}{N}\sum_{i=1}^N\mathbb{E}\left[\boldsymbol{\beta}_i\mathbf{1}\{Y_i\leq \lfloor n_i/2\rfloor\}\right]\\
&\displaystyle\quad-\frac{1}{N}\sum_{i=1}^N\mathbb{E}\left[\mathbf{1}(Y_{i}\leq \lfloor n_i/2\rfloor)\sum_{j=0}^{Y_{i}}\boldsymbol{\beta}_i(Y_i-j)(-1)^{j}\frac{\binom{Y_i}{j}}{\binom{n_i-Y_i+j}{j}}\right].
\end{array}
\end{equation}
Thus the first element of $\displaystyle\frac{1}{N}\sum_{i=1}^N\mathbb{E}\left[\left(\bar{Y}-\theta_i^{\mathrm{o}}\right)\boldsymbol{\beta}_i\right]-\frac{1}{2}\boldsymbol{\mu}_{m,1}\in\mathbb{R}^2$ is equal to 

\begin{equation}\label{eq:bias:one-sample:1}
\begin{array}{rl}
&\displaystyle-\frac{1}{N}\sum_{i=1}^N\theta_i^{\mathrm{o}}\mathbb{E}\left[\frac{Y_i}{n_i}-\bar{Y}\right]+\frac{1}{N}\sum_{i=1}^N\mathbb{E}\left[\mathbf{1}(Y_{i} > \lfloor n_i/2\rfloor)\!\!\!\sum_{j=0}^{n_{i}-Y_{i}}\!\!\left(\frac{Y_i+j}{n_i}-\bar{Y}\right)(-1)^{j}\frac{\binom{n_i-Y_i}{j}}{\binom{Y_i+j}{j}}\right]\\
&\displaystyle\quad+\frac{1}{N}\sum_{i=1}^N\mathbb{E}\left[\left(\frac{Y_i}{n_i}-\bar{Y}\right)\mathbf{1}\{Y_i\leq \lfloor n_i/2\rfloor\}\right]\\
&\displaystyle\quad-\frac{1}{N}\sum_{i=1}^N\mathbb{E}\left[\mathbf{1}(Y_{i}\leq \lfloor n_i/2\rfloor)\sum_{j=0}^{Y_{i}}\left(\frac{Y_i-j}{n_i}-\bar{Y}\right)(-1)^{j}\frac{\binom{Y_i}{j}}{\binom{n_i-Y_i+j}{j}}\right].
\end{array}
\end{equation}
The second element of $\displaystyle\frac{1}{N}\sum_{i=1}^N\mathbb{E}\left[\left(\bar{Y}-\theta_i^{\mathrm{o}}\right)\boldsymbol{\beta}_i\right]-\frac{1}{2}\boldsymbol{\mu}_{m,1}$ is equal to 

\begin{equation}\label{eq:bias:one-sample:2}
\begin{array}{rl}
&\displaystyle-\frac{1}{N}\sum_{i=1}^N\theta_i^{\mathrm{o}}\mathbb{E}\left[\hat{g}(\mathbf{X}_i)-\hat{g}\right]+\frac{1}{N}\sum_{i=1}^N\mathbb{E}\left[\mathbf{1}(Y_{i} > \lfloor n_i/2\rfloor)\!\!\!\sum_{j=0}^{n_{i}-Y_{i}}\!\!\left(\hat{g}(\mathbf{X}_i)-\hat{g}\right)(-1)^{j}\frac{\binom{n_i-Y_i}{j}}{\binom{Y_i+j}{j}}\right]\\
&\displaystyle\quad+\frac{1}{N}\sum_{i=1}^N\mathbb{E}\left[\left(\hat{g}(\mathbf{X}_i)-\hat{g}\right)\mathbf{1}\{Y_i\leq \lfloor n_i/2\rfloor\}\right]\\
&\displaystyle\quad-\frac{1}{N}\sum_{i=1}^N\mathbb{E}\left[\mathbf{1}(Y_{i}\leq \lfloor n_i/2\rfloor)\sum_{j=0}^{Y_{i}}\left(\hat{g}(\mathbf{X}_i)-\hat{g}\right)(-1)^{j}\frac{\binom{Y_i}{j}}{\binom{n_i-Y_i+j}{j}}\right].
\end{array}
\end{equation}
We first focus on the second element of $\displaystyle\frac{1}{N}\sum_{i=1}^N\mathbb{E}\left[\left(\bar{Y}-\theta_i^{\mathrm{o}}\right)\boldsymbol{\beta}_i\right]-\frac{1}{2}\boldsymbol{\mu}_{m,1}$. Recall that $g(\mathbf{X}_i)$ and $\displaystyle\bar{g}=\frac{\sum_{i=1}^Nn_ig(\mathbf{X}_i)}{\sum_{i=1}^Nn_i}$ are both constants, so 
\begin{equation}\label{eq:0:unbias:term:2}
\begin{array}{rl}
0\!\!\!\!&=\displaystyle-\frac{1}{N}\sum_{i=1}^N\theta_i^{\mathrm{o}}\mathbb{E}\left[g(\mathbf{X}_i)-\bar{g}\right]+\frac{1}{N}\sum_{i=1}^N\mathbb{E}\left[\mathbf{1}(Y_{i} > \lfloor n_i/2\rfloor)\!\!\!\sum_{j=0}^{n_{i}-Y_{i}}\!\!\left(g(\mathbf{X}_i)-\bar{g}\right)(-1)^{j}\frac{\binom{n_i-Y_i}{j}}{\binom{Y_i+j}{j}}\right]\\
&\displaystyle\quad+\frac{1}{N}\sum_{i=1}^N\mathbb{E}\left[\left(g(\mathbf{X}_i)-\bar{g}\right)\mathbf{1}\{Y_i\leq \lfloor n_i/2\rfloor\}\right]\\
&\displaystyle\quad-\frac{1}{N}\sum_{i=1}^N\mathbb{E}\left[\mathbf{1}(Y_{i}\leq \lfloor n_i/2\rfloor)\sum_{j=0}^{Y_{i}}\left(g(\mathbf{X}_i)-\bar{g}\right)(-1)^{j}\frac{\binom{Y_i}{j}}{\binom{n_i-Y_i+j}{j}}\right],
\end{array}
\end{equation}
which follows by setting $h(Y_i)=g(\mathbf{X}_i)-\bar{g}$ (a constant function) for each $i\in[N]$ in Lemma \ref{lemma:main:cor}. Denote 
$$\Delta g_i:=\{\hat{g}(\mathbf{X}_i)-g(\mathbf{X}_i)\}-\{\hat{g}-\bar{g}\},$$
\eqref{eq:bias:one-sample:2} and \eqref{eq:0:unbias:term:2} imply that the second element of $\displaystyle\frac{1}{N}\sum_{i=1}^N\mathbb{E}\left[\left(\bar{Y}-\theta_i^{\mathrm{o}}\right)\boldsymbol{\beta}_i\right]-\frac{1}{2}\boldsymbol{\mu}_{m,1}$ is equal to
\begin{equation}\label{eq:bias:one-sample:2:rewrite}
\begin{array}{rl}
&\displaystyle-\frac{1}{N}\sum_{i=1}^N\theta_i^{\mathrm{o}}\mathbb{E}\left[\Delta g_i\right]+\frac{1}{N}\sum_{i=1}^N\mathbb{E}\left[\mathbf{1}(Y_{i} > \lfloor n_i/2\rfloor)\Delta g_i\sum_{j=0}^{n_{i}-Y_{i}}(-1)^{j}\frac{\binom{n_i-Y_i}{j}}{\binom{Y_i+j}{j}}\right]\\
&\displaystyle\quad+\frac{1}{N}\sum_{i=1}^N\mathbb{E}\left[\Delta g_i\mathbf{1}\{Y_i\leq \lfloor n_i/2\rfloor\}\right]\\
&\displaystyle\quad-\frac{1}{N}\sum_{i=1}^N\mathbb{E}\left[\mathbf{1}(Y_{i}\leq \lfloor n_i/2\rfloor)\Delta g_i\sum_{j=0}^{Y_{i}}(-1)^{j}\frac{\binom{Y_i}{j}}{\binom{n_i-Y_i+j}{j}}\right].
\end{array}
\end{equation}
According to Lemma \ref{lemma:combinatorics}, 
$$\sum_{j=0}^{n_{i}-Y_{i}}(-1)^{j}\frac{\binom{n_i-Y_i}{j}}{\binom{Y_i+j}{j}}=\frac{Y_i}{n_i},\quad \sum_{j=0}^{Y_{i}}(-1)^{j}\frac{\binom{Y_i}{j}}{\binom{n_i-Y_i+j}{j}}=_{(a)}\frac{n_i-Y_i}{n_i},$$
where the last equality (a) follows by setting $Y=n_i-Y_i$ in Lemma \ref{lemma:combinatorics}. So taking both equalities into \eqref{eq:bias:one-sample:2:rewrite}, the second element of $\displaystyle\frac{1}{N}\sum_{i=1}^N\mathbb{E}\left[\left(\bar{Y}-\theta_i^{\mathrm{o}}\right)\boldsymbol{\beta}_i\right]-\frac{1}{2}\boldsymbol{\mu}_{m,1}$ is equal to 
\begin{equation}\label{eq:ele:second:one-sample}
\begin{array}{rl}
&\displaystyle\quad\frac{1}{N}\sum_{i=1}^N\mathbb{E}\left[\Delta g_i\left(\frac{Y_i}{n_i}-\theta_i^{\mathrm{o}}\right)\right]\\
&\displaystyle=\frac{1}{N}\sum_{i=1}^N\mathbb{E}\left[\{\hat{g}(\mathbf{X}_i)-g(\mathbf{X}_i)\}\left(\frac{Y_i}{n_i}-\theta_i^{\mathrm{o}}\right)\right]-\{\hat{g}-\bar{g}\}\frac{1}{N}\sum_{i=1}^N\mathbb{E}\left[\frac{Y_i}{n_i}-\theta_i^{\mathrm{o}}\right].
\end{array}
\end{equation}
Note that $\displaystyle\frac{1}{\sqrt{N}}\sum_{i=1}^N\mathbb{E}\left[\frac{Y_i}{n_i}-\theta_i^{\mathrm{o}}\right]=\mathrm{O}_p(1)$ according to (iii) of Assumption \ref{assump:one_sample} and Lemma \ref{lemma:Lindeberg-Feller:CLT:multivariate}, also according to statements (i) and (ii) of Assumption \ref{assump:one_sample}, $\hat{g}-\bar{g}=\mathrm{o}_p(1)$. Thus
$$\{\hat{g}-\bar{g}\}\frac{1}{\sqrt{N}}\sum_{i=1}^N\mathbb{E}\left[\frac{Y_i}{n_i}-\theta_i^{\mathrm{o}}\right]=\mathrm{o}_p(1).$$
Note that $\hat{g}(\cdot)$ is trained via cross-fitting, so $\hat{g}(\mathbf{X}_i)=\hat{g}^{-k(i)}(\mathbf{X}_i)$ is independent of $Y_i$, thus 
$$\mathbb{E}\left[\{\hat{g}(\mathbf{X}_i)-g(\mathbf{X}_i)\}\left(\frac{Y_i}{n_i}-\theta_i^{\mathrm{o}}\right)\right]=\mathbb{E}\left[\hat{g}(\mathbf{X}_i)-g(\mathbf{X}_i)\right]\mathbb{E}\left[\frac{Y_i}{n_i}-\theta_i^{\mathrm{o}}\right]=0.$$
Hence \eqref{eq:ele:second:one-sample} implies
\begin{equation}\label{eq:bias:second:one-sample}
\sqrt{N}\frac{1}{N}\sum_{i=1}^N\mathbb{E}\left[\Delta g_i\left(\frac{Y_i}{n_i}-\theta_i^{\mathrm{o}}\right)\right]=\mathrm{o}_p(1).
\end{equation}

\noindent We now focus on the first element of $\displaystyle\frac{1}{N}\sum_{i=1}^N\mathbb{E}\left[\left(\bar{Y}-\theta_i^{\mathrm{o}}\right)\boldsymbol{\beta}_i\right]-\frac{1}{2}\boldsymbol{\mu}_{m,1}$. Using Lemma \ref{lemma:main:cor}, 
\begin{equation}\label{eq:unbiased:one-sample}
\begin{array}{rl}
0\!\!\!\!&=\displaystyle-\frac{1}{N}\sum_{i=1}^N\theta_i^{\mathrm{o}}\mathbb{E}\left[\frac{Y_i}{n_i}\right]+\frac{1}{N}\sum_{i=1}^N\mathbb{E}\left[\mathbf{1}(Y_{i} > \lfloor n_i/2\rfloor)\!\!\!\sum_{j=0}^{n_{i}-Y_{i}}\!\!\left(\frac{Y_i+j}{n_i}\right)(-1)^{j}\frac{\binom{n_i-Y_i}{j}}{\binom{Y_i+j}{j}}\right]\\
&\displaystyle\quad+\frac{1}{N}\sum_{i=1}^N\mathbb{E}\left[\left(\frac{Y_i}{n_i}\right)\mathbf{1}\{Y_i\leq \lfloor n_i/2\rfloor\}\right]\\
&\displaystyle\quad-\frac{1}{N}\sum_{i=1}^N\mathbb{E}\left[\mathbf{1}(Y_{i}\leq \lfloor n_i/2\rfloor)\sum_{j=0}^{Y_{i}}\left(\frac{Y_i-j}{n_i}\right)(-1)^{j}\frac{\binom{Y_i}{j}}{\binom{n_i-Y_i+j}{j}}\right].
\end{array}
\end{equation}
Hence using \eqref{eq:unbiased:one-sample} for \eqref{eq:bias:one-sample:1}, we can see that \eqref{eq:bias:one-sample:1} is equal to
\begin{equation}\label{eq:bias:one-sample:1:rewrite}
\begin{array}{rl}
&\quad\displaystyle\frac{1}{N}\sum_{i=1}^N\theta_i^{\mathrm{o}}\mathbb{E}\left[\bar{Y}\right]-\frac{1}{N}\sum_{i=1}^N\mathbb{E}\left[\mathbf{1}(Y_{i} > \lfloor n_i/2\rfloor)\!\!\!\sum_{j=0}^{n_{i}-Y_{i}}\!\!\bar{Y}(-1)^{j}\frac{\binom{n_i-Y_i}{j}}{\binom{Y_i+j}{j}}\right]\\
&\displaystyle\quad\quad-\frac{1}{N}\sum_{i=1}^N\mathbb{E}\left[\bar{Y}\mathbf{1}\{Y_i\leq \lfloor n_i/2\rfloor\}\right]\\
&\displaystyle\quad\quad+\frac{1}{N}\sum_{i=1}^N\mathbb{E}\left[\mathbf{1}(Y_{i}\leq \lfloor n_i/2\rfloor)\sum_{j=0}^{Y_{i}}\bar{Y}(-1)^{j}\frac{\binom{Y_i}{j}}{\binom{n_i-Y_i+j}{j}}\right]\\
&=\displaystyle\frac{1}{N}\sum_{i=1}^N\theta_i^{\mathrm{o}}\mathbb{E}\left[\frac{\sum_{i}^{N}Y_i}{\sum_{i}^{N}n_i}\right]-\frac{1}{N}\sum_{i=1}^N\mathbb{E}\left[\mathbf{1}(Y_{i} > \lfloor n_i/2\rfloor)\!\!\!\sum_{j=0}^{n_{i}-Y_{i}}\frac{\sum_{i}^{N}Y_i}{\sum_{i}^{N}n_i}(-1)^{j}\frac{\binom{n_i-Y_i}{j}}{\binom{Y_i+j}{j}}\right]\\
&\displaystyle\quad\quad-\frac{1}{N}\sum_{i=1}^N\mathbb{E}\left[\frac{\sum_{i}^{N}Y_i}{\sum_{i}^{N}n_i}\mathbf{1}\{Y_i\leq \lfloor n_i/2\rfloor\}\right]\\
&\displaystyle\quad\quad+\frac{1}{N}\sum_{i=1}^N\mathbb{E}\left[\mathbf{1}(Y_{i}\leq \lfloor n_i/2\rfloor)\sum_{j=0}^{Y_{i}}\frac{\sum_{i}^{N}Y_i}{\sum_{i}^{N}n_i}(-1)^{j}\frac{\binom{Y_i}{j}}{\binom{n_i-Y_i+j}{j}}\right]\\
&\displaystyle=\frac{1}{\sum_{i}^{N}n_i}\Bigg\{\frac{1}{N}\sum_{i=1}^N\theta_i^{\mathrm{o}}\mathbb{E}\left[Y_i\right]-\frac{1}{N}\sum_{i=1}^N\mathbb{E}\left[\mathbf{1}(Y_{i} > \lfloor n_i/2\rfloor)\!\!\!\sum_{j=0}^{n_{i}-Y_{i}}Y_i(-1)^{j}\frac{\binom{n_i-Y_i}{j}}{\binom{Y_i+j}{j}}\right]\\
&\displaystyle\quad\quad-\frac{1}{N}\sum_{i=1}^N\mathbb{E}\left[Y_i\mathbf{1}\{Y_i\leq \lfloor n_i/2\rfloor\}\right]\\
&\displaystyle\quad\quad+\frac{1}{N}\sum_{i=1}^N\mathbb{E}\left[\mathbf{1}(Y_{i}\leq \lfloor n_i/2\rfloor)\sum_{j=0}^{Y_{i}}Y_i(-1)^{j}\frac{\binom{Y_i}{j}}{\binom{n_i-Y_i+j}{j}}\right]\Bigg\},
\end{array}
\end{equation}
where the last equality of \eqref{eq:bias:one-sample:1:rewrite} follows since given any $i\in[N]$ fixed, for any $k\neq i$, $Y_k$ is independent of $Y_i$, so 
\begin{equation}\label{eq:unbias:0:one-sample}
\begin{array}{rl}
&\displaystyle\quad\theta_i^{\mathrm{o}}\mathbb{E}\left[Y_k\right]-\mathbb{E}\left[\mathbf{1}(Y_{i} > \lfloor n_i/2\rfloor)\!\!\!\sum_{j=0}^{n_{i}-Y_{i}}Y_k(-1)^{j}\frac{\binom{n_i-Y_i}{j}}{\binom{Y_i+j}{j}}\right]-\mathbb{E}\left[Y_k\mathbf{1}\{Y_i\leq \lfloor n_i/2\rfloor\}\right]\\
&\displaystyle\quad\quad+\mathbb{E}\left[\mathbf{1}(Y_{i}\leq \lfloor n_i/2\rfloor)\sum_{j=0}^{Y_{i}}Y_k(-1)^{j}\frac{\binom{Y_i}{j}}{\binom{n_i-Y_i+j}{j}}\right]\\
&\displaystyle=\mathbb{E}[Y_k]\Bigg\{\theta_i^{\mathrm{o}}-\mathbb{E}\left[\mathbf{1}(Y_{i} > \lfloor n_i/2\rfloor)\!\!\!\sum_{j=0}^{n_{i}-Y_{i}}(-1)^{j}\frac{\binom{n_i-Y_i}{j}}{\binom{Y_i+j}{j}}\right]-\mathbb{E}[\mathbf{1}\{Y_i\leq\lfloor n_i/2\rfloor\}]\\
&\displaystyle\quad\quad\quad\quad+\mathbb{E}\left[\mathbf{1}(Y_{i}\leq \lfloor n_i/2\rfloor)\sum_{j=0}^{Y_{i}}(-1)^{j}\frac{\binom{Y_i}{j}}{\binom{n_i-Y_i+j}{j}}\right]\Bigg\}=_{(a)}0,
\end{array}
\end{equation}
and (a) of \eqref{eq:unbias:0:one-sample} follows by setting $h(Y_i)\equiv1$ in Lemma \ref{lemma:main:cor}. Further, by setting $h(Y_i)=Y_i$ for any $i\in[N]$ in (c) of Lemma \ref{lemma:main:cor}, we have 
\begin{equation}\label{eq:Y:val:one-sample}
\begin{array}{rl}
0\!\!\!\!&=\displaystyle\frac{1}{N}\sum_{i=1}^N\theta_i^{\mathrm{o}}\mathbb{E}\left[Y_i\right]-\frac{1}{N}\sum_{i=1}^N\mathbb{E}\left[\mathbf{1}(Y_{i} > \lfloor n_i/2\rfloor)\!\!\!\sum_{j=0}^{n_{i}-Y_{i}}(Y_i+j)(-1)^{j}\frac{\binom{n_i-Y_i}{j}}{\binom{Y_i+j}{j}}\right]\\
&\displaystyle\quad\quad-\frac{1}{N}\sum_{i=1}^N\mathbb{E}\left[Y_i\mathbf{1}\{Y_i\leq \lfloor n_i/2\rfloor\}\right]\\
&\displaystyle\quad\quad+\frac{1}{N}\sum_{i=1}^N\mathbb{E}\left[\mathbf{1}(Y_{i}\leq \lfloor n_i/2\rfloor)\sum_{j=0}^{Y_{i}}(Y_i-j)(-1)^{j}\frac{\binom{Y_i}{j}}{\binom{n_i-Y_i+j}{j}}\right]
\end{array}
\end{equation}
Using \eqref{eq:Y:val:one-sample} for \eqref{eq:bias:one-sample:1:rewrite}, the right hand side of \eqref{eq:bias:one-sample:1:rewrite} is equal to
\begin{equation}\label{eq:RHS:one-sample}
\begin{array}{rl}
&\!\!\!\!\!\!\!\!\displaystyle\frac{1}{N\sum_{i=1}^Nn_i}\!\sum_{i=1}^N\!\mathbb{E}\!\!\left[\mathbf{1}\{Y_i>\lfloor n_i/2\rfloor\}\sum_{j=0}^{n_i-Y_i}j(-1)^j\frac{\binom{n_i-Y_i}{j}}{\binom{Y_i+j}{j}}+\mathbf{1}\{Y_{i}\leq \lfloor n_i/2\rfloor\}\sum_{j=0}^{Y_{i}}j(-1)^{j}\frac{\binom{Y_i}{j}}{\binom{n_i-Y_i+j}{j}}\right]\\
&\displaystyle=_{(a)}\frac{1}{N\sum_{i=1}^Nn_i}\sum_{i=1}^N\mathbb{E}\left[-\mathbf{1}\{Y_i>\lfloor n_i/2\rfloor\}\frac{Y_i(n_i-Y_i)}{n_i(n_i-1)}-\mathbf{1}\{Y_{i}\leq \lfloor n_i/2\rfloor\}\frac{Y_i(n_i-Y_i)}{n_i(n_i-1)}\right]\\
&\displaystyle=-\frac{1}{N\sum_{i=1}^Nn_i}\sum_{i=1}^N\mathbb{E}\left[\frac{n_i-Y_i}{n_i}\frac{Y_i}{n_i-1}\right],
\end{array}
\end{equation}
where (a) of \eqref{eq:RHS:one-sample} follows since 
$$\sum_{j=0}^{n_i-Y_i}j(-1)^j\frac{\binom{n_i-Y_i}{j}}{\binom{Y_i+j}{j}}=_{(1)}-\frac{Y_i(n_i-Y_i)}{n_i(n_i-1)}=_{(2)}\sum_{j=0}^{Y_{i}}j(-1)^{j}\frac{\binom{Y_i}{j}}{\binom{n_i-Y_i+j}{j}},$$
where (1) in the last equality follows from Lemma \ref{lemma:combinatorics:2} and (2) follows by replacing $Y$ with $n_i-Y_i$ in Lemma \ref{lemma:combinatorics:2}. Hence according to \eqref{eq:bias:one-sample:1}, \eqref{eq:bias:one-sample:1:rewrite} and \eqref{eq:RHS:one-sample}, the first element of 
$$\displaystyle\frac{1}{N}\sum_{i=1}^N\mathbb{E}\left[\left(\bar{Y}-\theta_i^{\mathrm{o}}\right)\boldsymbol{\beta}_i\right]-\frac{1}{2}\boldsymbol{\mu}_{m,1}$$ 
is equal to 
\begin{equation}\label{eq:one-sample:bias:term:1}
    -\frac{1}{N\sum_{i=1}^Nn_i}\sum_{i=1}^N\mathbb{E}\left[\frac{n_i-Y_i}{n_i}\frac{Y_i}{n_i-1}\right].
\end{equation}
Since $Y_i\in[0,n_i]$ and $n_i\in[2,\bar{n}]$, we have 
\begin{equation}\label{eq:one-sample:bias:term:1:sqrt:N}
    \sqrt{N}\left|-\frac{1}{N\sum_{i=1}^Nn_i}\sum_{i=1}^N\mathbb{E}\left[\frac{n_i-Y_i}{n_i}\frac{Y_i}{n_i-1}\right]\right|\leq\frac{1}{2\sqrt{N}}.
\end{equation}
Thus according to \eqref{eq:bias:one-sample:2} - \eqref{eq:bias:second:one-sample} and \eqref{eq:one-sample:bias:term:1}, \eqref{eq:one-sample:bias:term:1:sqrt:N}, both elements of  
$$\sqrt{N}\left[\frac{1}{N}\sum_{i=1}^N\mathbb{E}\left[\left(\bar{Y}-\theta_i^{\mathrm{o}}\right)\boldsymbol{\beta}_i\right]-\frac{1}{2}\boldsymbol{\mu}_{m,1}\right]=\begin{pmatrix}
    b_1\\
    b_2
\end{pmatrix},$$
where $b_1=\mathrm{o}_p(1)$ and $b_2=\mathrm{o}_p(1)$. Thus the Lemma holds by recalling the formula for $\boldsymbol{\lambda}_{\mathrm{o}}^*$ in \eqref{eq:lambda:star:one-sample}, and noting that the elements of $\boldsymbol{\mu}_{m,2}=\mathbb{E}[\mathbf{C}_{N,2}]$ are of order $\mathrm{O}(1)$. 
\end{proof}

\subsection{Proofs for Asymptotic Normality of Two-Sample Case}\label{appendix:two-sample:case}
\begin{assumption}\label{ass:additional:two-sample}
Suppose the following statements hold:\\
\noindent (i) $\displaystyle\frac{1}{N}\sum_{i=1}^N\mathrm{Var}(Y_{i\ell})\rightarrow\sigma_{Y\ell}^2$,  $\displaystyle\frac{1}{N}\sum_{i=1}^N\mathbb{E}[Y_{i\ell}]\rightarrow\mu_{Y\ell}^*$, $\displaystyle\frac{1}{N}\sum_{i=1}^Nn_{i\ell}\rightarrow\mu_{n\ell}^*$, $\displaystyle\frac{1}{N}\sum_{i=1}^N\theta_{i\ell}^{\mathrm{t}}\rightarrow\mu_{\theta\ell}^*$, \\
$\displaystyle\frac{1}{N}\sum_{i=1}^N\left(\theta_{i\ell}^{\mathrm{t}}\right)^2\rightarrow\sigma_{\theta\ell}^2$, $\displaystyle\frac{1}{N}\sum_{i=1}^Ng_{\ell}(\mathbf{X}_{i\ell})\rightarrow\mu_{g\ell}^*$, $\displaystyle\frac{1}{N}\sum_{i=1}^N\mathbb{E}\left[\mathbf{1}\{Y_{i\ell}>\lfloor n_{i\ell}/2\rfloor\}\right]\rightarrow\mu_{I\ell}^*$,\\ $\displaystyle\frac{1}{N}\sum_{i=1}^N\mathbb{E}\left[\mathbf{1}\{Y_{i\ell}>\lfloor n_{i\ell}/2\rfloor\}\frac{Y_{i\ell}}{n_{i\ell}}\right]\rightarrow\mu_{IY,\ell}^*$, where $\sigma_{Y\ell}^2$, $\mu_{Y\ell}^*$, $\mu_{n\ell}^*$, $\mu_{gn,\ell}^*$, $\mu_{\theta\ell}^*$, $\sigma_{\theta\ell}$, $\mu_{g\ell}^*$, $\mu_{I\ell}^*$, $\mu_{IY,\ell}^*$ are all absolute constants.

\noindent (ii) $\displaystyle\frac{1}{N}\sum_{i=1}^N\mathrm{Cov}\left\{\overline{\mathbf{Z}}_i\right\}\rightarrow\overline{\boldsymbol{\Sigma}}$, where $\overline{\boldsymbol{\Sigma}}\in\mathbb{R}^{6\times 6}$ is a positive definite matrix, $\overline{\mathbf{Z}}_i$ are i.n.i.d. across $i\in[N]$ and each $\overline{\mathbf{Z}}_i=(\bar{\zeta}_{i1},\bar{\zeta}_{i2},\bar{\zeta}_{i3},\bar{\zeta}_{i4},\bar{\zeta}_{i5},\bar{\zeta}_{i6})^T\in\mathbb{R}^6$ is defined as \eqref{eq:ass:two-sample:cov:Z}:
\begin{equation}\label{eq:ass:two-sample:cov:Z}
\begin{array}{rcl}
&\displaystyle\bar{\zeta}_{i1}=2W_i^*\Delta Y_i^*+2I_{i1}v_{i1}+(2I_{i2}-1)v_{i2}+[(\mu_{\theta1}^*-\mu_{\theta2}^*)-2\kappa^*]\Delta_{i,Yn},&\\
&\displaystyle\bar{\zeta}_{i2}=2W_i^*\Gamma_i^*+\kappa_3^*\Delta_{i,Yn},&\\
&\displaystyle\bar{\zeta}_{i3}=\left[\left(\frac{Y_{i1}}{n_{i1}}-\frac{Y_{i2}}{n_{i2}}\right)-(\bar{Y}_1^*-\bar{Y}_2^*)\right]^2-2\kappa_2^*\Delta_{i,Yn},&\\
&\displaystyle\bar{\zeta}_{i4}=\left[\left(\frac{Y_{i1}}{n_{i1}}-\frac{Y_{i2}}{n_{i2}}\right)-(\bar{Y}_1^*-\bar{Y}_2^*)\right][(g_1(\mathbf{X}_{i1})-g_2(\mathbf{X}_{i2}))-(\bar{g}_1-\bar{g}_2)]-\kappa_3^*\Delta_{i,Yn},&\\
&\displaystyle\bar{\zeta}_{i5}=\left[\left(\frac{Y_{i1}}{n_{i1}}-\frac{Y_{i2}}{n_{i2}}\right)-(\bar{Y}_1^*-\bar{Y}_2^*)\right][(g_1(\mathbf{X}_{i1})-g_2(\mathbf{X}_{i2}))-(\bar{g}_1-\bar{g}_2)]-\kappa_3^*\Delta_{i,Yn},&\\
&\displaystyle\bar{\zeta}_{i6}=[(g_1(\mathbf{X}_{i1})-g_2(\mathbf{X}_{i2}))-(\bar{g}_1-\bar{g}_2)]^2,&
\end{array}
\end{equation}
where 
$$\bar{Y}_\ell^*=\frac{\sum_{i=1}^Nn_{i\ell}Y_{i\ell}}{\sum_{i=1}^Nn_{i\ell}},\quad \bar{g}_{\ell}=\frac{\sum_{i=1}^Nn_{i\ell}g_{\ell}(\mathbf{X}_{i\ell})}{\sum_{i=1}^Nn_{i\ell}},\ \ \forall \ell\in\{1,2\},$$
$$\Delta Y_i^*:=\left(\frac{Y_{i1}}{n_{i1}}-\frac{Y_{i2}}{n_{i2}}\right)-(\bar{Y}_1^*-\bar{Y}_2^*),\quad \Gamma_i^*:=(g_1(\mathbf{X}_{i1})-g_2(\mathbf{X}_{i2}))-(\bar{g}_1-\bar{g}_2),$$
$$v_{i1}=\frac{Y_{i1}(n_{i1}-Y_{i1})}{n_{i1}^2(n_{i1}-1)},\quad v_{i2}=\frac{Y_{i2}(n_{i2}-Y_{i2})}{n_{i2}^2(n_{i2}-1)},$$
$$I_{i1}:=\mathbf{1}\{Y_{i1}>\lfloor n_{i1}/2\rfloor\},\quad I_{i2}:=\mathbf{1}\{Y_{i2}>\lfloor n_{i2}/2\rfloor\},$$
$$W_i^*:=(\bar{Y}_1^*-\bar{Y}_2^*)+2I_{i1}\left(1-\frac{Y_{i1}}{n_{i1}}\right)+(2I_{i2}-1)\left(\frac{Y_{i2}}{n_{i2}}-1\right),$$
$$\Delta_{i,Yn}=\frac{Y_{i1}}{\mu_{n1}^*}-\frac{Y_{i2}}{\mu_{n2}^*},$$
$$\kappa^*=\frac{\mu_{Y1}^*}{\mu_{n1}^*}-\frac{\mu_{Y2}^*}{\mu_{n2}^*}+2(\mu_{I1}^*-\mu_{IY,1}^*)+2\mu_{IY,2}^*-\mu_{\theta2}^*-2\mu_{I2}^*+1,$$
$$\kappa_2^*=\left(\mu_{\theta1}^*-\mu_{\theta2}^*\right)-\left(\frac{\mu_{Y1}^*}{\mu_{n1}^*}-\frac{\mu_{Y2}^*}{\mu_{n2}^*}\right),\kappa_3^*=\left[(\mu_{g1}^*-\mu_{g2}^*)-\left(\frac{\mu_{gn,1}^*}{\mu_{n1}^*}-\frac{\mu_{gn,2}^*}{\mu_{n2}^*}\right)\right].$$
\end{assumption}

\begin{theorem}\label{thm:CLT:two-sample}
Suppose Assumption \ref{assump:two_sample} and Assumption~\ref{ass:additional:two-sample} hold. Suppose $\boldsymbol{\lambda}_{\mathrm{t}}^*$ is unconstrained. Then 
$$\sqrt{N}\left(\hat{\boldsymbol{\lambda}}_{\mathrm{t}}-\boldsymbol{\lambda}_{\mathrm{t}}^*\right)\rightsquigarrow\mathcal{N}(\mathbf{0},\overline{\mathbf{V}}),$$ 
where $\overline{\mathbf{V}}\preceq\bar{C}'\mathbf{I}_2$ for some absolute constant $\bar{C}'$.   
\end{theorem}
\begin{proof}[Proof of Theorem \ref{thm:CLT:two-sample}]
Recall that $\overline{\boldsymbol{\mu}}_{m,1}=\mathbb{E}[\mathbf{D}_{N,1}]$ and $\overline{\boldsymbol{\mu}}_{m,2}=\mathbb{E}[\mathbf{D}_{N,2}]$. According to Lemma \ref{lemma:CLT:two-sample} we have  
$$\sqrt{N}\begin{pmatrix}
\mathbf{D}_{N,1}-\overline{\boldsymbol{\mu}}_{m,1}\\
\mathrm{vec}(\mathbf{D}_{N,2}-\overline{\boldsymbol{\mu}}_{m,2})
\end{pmatrix}\rightsquigarrow\mathcal{N}\left(\mathbf{0},\overline{\boldsymbol{\Sigma}}\right).$$
Then the rest of the proof follows the same steps as in the proof of Theorem \ref{thm:one-sample:asymptotics} by applying Delta's method, and we can get 
$$\sqrt{N}\left(\hat{\boldsymbol{\lambda}}_{\mathrm{t}}-\left\{-\frac{1}{2}\overline{\boldsymbol{\mu}}_{m,2}^{-1}\overline{\boldsymbol{\mu}}_{m,1}\right\}\right)\rightsquigarrow\mathcal{N}(\mathbf{0},\overline{\mathbf{V}}),$$ 
where $\overline{\mathbf{V}}\preceq\bar{C}'\mathbf{I}_2$ for some absolute constant $\bar{C}'$. Then Theorem \ref{thm:CLT:two-sample} holds by following similar proof steps as in Lemma \ref{lemma:bias:lambda:one-sample} so that 
$$\sqrt{N}\left(\boldsymbol{\lambda}_{\mathrm{t}}^*-\left\{-\frac{1}{2}\overline{\boldsymbol{\mu}}_{m,2}^{-1}\overline{\boldsymbol{\mu}}_{m,1}\right\}\right)=\begin{pmatrix}
    a_1\\
    a_2
\end{pmatrix},$$
where $a_1=\mathrm{o}_p(1)$ and $a_2=\mathrm{o}_p(1)$. 
\end{proof}

\subsubsection{Technical Lemmas for the Two-Sample Asymptotic Normality Result}
\begin{lemma}\label{lemma:CLT:two-sample} Let $\mathrm{vec}(\mathbf{D}_{N,2})$ be the vector with the four column-wise entries of $\mathbf{D}_{N,2}$ stacked. Let $\overline{\boldsymbol{\mu}}_{m,1}=\mathbb{E}[\mathbf{D}_{N,1}]$ and $\overline{\boldsymbol{\mu}}_{m,2}=\mathbb{E}[\mathrm{vec}(\mathbf{D}_{N,2})]$. Suppose Assumption \ref{assump:two_sample} holds, then 
\begin{equation}\label{eq:CLT:main:two-sample}
\sqrt{N}\begin{pmatrix}
\mathbf{D}_{N,1}-\overline{\boldsymbol{\mu}}_{m,1}\\
\mathrm{vec}(\mathbf{D}_{N,2}-\overline{\boldsymbol{\mu}}_{m,2})
\end{pmatrix}\rightsquigarrow\mathcal{N}\left(\mathbf{0},\overline{\boldsymbol{\Sigma}}\right),
\end{equation}
where $\overline{\boldsymbol{\Sigma}}\preceq\tilde{C}\mathbf{I}_6$ for some absolute constant $\tilde{C}$, and $\mathbf{I}_6$ is the $6$-by-$6$ identity matrix.
\end{lemma}
\begin{proof}[Proof of Lemma \ref{lemma:CLT:two-sample}]
Firstly, according to \eqref{eq:coeff:quadratic:two-sample}, we have 
\begin{equation}\label{eq:D:N:2:expand}
\begin{array}{rl}
    \mathrm{vec}\left(\mathbf{D}_{N,2}\right)\!\!\!\!&\displaystyle=\frac{1}{N}\sum_{i=1}^N\mathrm{vec}\left\{(\boldsymbol{\beta}_{i1}-\boldsymbol{\beta}_{i2})(\boldsymbol{\beta}_{i1}-\boldsymbol{\beta}_{i2})^T\right\}\\
    &\displaystyle=\begin{pmatrix}
        \displaystyle\frac{1}{N}\sum_{i=1}^N\left[\left(\frac{Y_{i1}}{n_{i1}}-\frac{Y_{i2}}{n_{i2}}\right)-(\bar{Y}_1-\bar{Y}_2)\right]^2\\
        \displaystyle\frac{1}{N}\sum_{i=1}^N\left[\left(\frac{Y_{i1}}{n_{i1}}-\frac{Y_{i2}}{n_{i2}}\right)-(\bar{Y}_1-\bar{Y}_2)\right][(\hat{g}_{i1}-\hat{g}_{i2})-(\hat{g}_1-\hat{g}_2)]\\
        \displaystyle\frac{1}{N}\sum_{i=1}^N\left[\left(\frac{Y_{i1}}{n_{i1}}-\frac{Y_{i2}}{n_{i2}}\right)-(\bar{Y}_1-\bar{Y}_2)\right][(\hat{g}_{i1}-\hat{g}_{i2})-(\hat{g}_1-\hat{g}_2)]\\
        \displaystyle\frac{1}{N}\sum_{i=1}^N[(\hat{g}_{i1}-\hat{g}_{i2})-(\hat{g}_1-\hat{g}_2)]^2
    \end{pmatrix}.
\end{array}
\end{equation}
Denote
$$\Delta Y_i=\left(\frac{Y_{i1}}{n_{i1}}-\frac{Y_{i2}}{n_{i2}}\right)-(\bar{Y}_1-\bar{Y}_2),\quad \Gamma_i=(\hat{g}_{i1}-\hat{g}_{i2})-(\hat{g}_1-\hat{g}_2),$$
$$\bar{Y}_\ell^*=\frac{\sum_{i=1}^N\mathbb{E}[n_{i\ell}Y_{i\ell}]}{\sum_{i=1}^Nn_{i\ell}},\quad \bar{g}_{\ell}=\frac{\sum_{i=1}^N\mathbb{E}[n_{i\ell}g_{\ell}(\mathbf{X}_{i\ell})]}{\sum_{i=1}^Nn_{i\ell}},\ \ \forall \ell\in\{1,2\},$$
$$\Delta Y_i^*:=\left(\frac{Y_{i1}}{n_{i1}}-\frac{Y_{i2}}{n_{i2}}\right)-(\bar{Y}_1^*-\bar{Y}_2^*),\quad \Gamma_i^*:=(g_{i1}-g_{i2})-(\bar{g}_1-\bar{g}_2),$$
$$v_{i1}=\frac{Y_{i1}(n_{i1}-Y_{i1})}{n_{i1}^2(n_{i1}-1)},\quad v_{i2}=\frac{Y_{i2}(n_{i2}-Y_{i2})}{n_{i2}^2(n_{i2}-1)},$$
$$I_{i1}:=\mathbf{1}\{Y_{i1}>\lfloor n_{i1}/2\rfloor\},\quad I_{i2}:=\mathbf{1}\{Y_{i2}>\lfloor n_{i2}/2\rfloor\},$$
$$W_i:=(\bar{Y}_1-\bar{Y}_2)+2I_{i1}\left(1-\frac{Y_{i1}}{n_{i1}}\right)+(2I_{i2}-1)\left(\frac{Y_{i2}}{n_{i2}}-1\right),$$
$$W_i^*:=(\bar{Y}_1^*-\bar{Y}_2^*)+2I_{i1}\left(1-\frac{Y_{i1}}{n_{i1}}\right)+(2I_{i2}-1)\left(\frac{Y_{i2}}{n_{i2}}-1\right),$$
$$\Delta g_i=[(\hat{g}_{i1}-g_{i1})-(\hat{g}_{i2}-g_{i2})]-[(\hat{g}_1-\bar{g}_1)-(\hat{g}_2-\bar{g}_2)].$$
Hence, according to Lemma \ref{lemma:D:N,1:rewritten}, we can rewrite 
\begin{equation}\label{eq:two-sample:D:rewrite}
\begin{pmatrix}
    \mathbf{D}_{N,1}\\
    \mathrm{vec}\left(\mathbf{D}_{N,2}\right)
\end{pmatrix}=\frac{1}{N}\sum_{i=1}^N\{\tilde{\boldsymbol{\xi}}_i+\boldsymbol{\delta}_i\},
\end{equation}
where  
\begin{equation}\label{eq:xi:two-sample:inid}
    \tilde{\boldsymbol{\xi}}_i=\begin{pmatrix}
        \displaystyle 2W_i^*\Delta Y_i^*+2I_{i1}v_{i1}+(2I_{i2}-1)v_{i2}\\
        \displaystyle 2W_i^*\Gamma_i^*\\
        \displaystyle\left[\left(\frac{Y_{i1}}{n_{i1}}-\frac{Y_{i2}}{n_{i2}}\right)-(\bar{Y}_1^*-\bar{Y}_2^*)\right]^2\\
        \displaystyle\left[\left(\frac{Y_{i1}}{n_{i1}}-\frac{Y_{i2}}{n_{i2}}\right)-(\bar{Y}_1^*-\bar{Y}_2^*)\right][(g_{i1}-g_{i2})-(\bar{g}_1-\bar{g}_2)]\\
        \displaystyle\left[\left(\frac{Y_{i1}}{n_{i1}}-\frac{Y_{i2}}{n_{i2}}\right)-(\bar{Y}_1^*-\bar{Y}_2^*)\right][(g_{i1}-g_{i2})-(\bar{g}_1-\bar{g}_2)]\\
        \displaystyle[(g_{i1}-g_{i2})-(\bar{g}_1-\bar{g}_2)]^2
    \end{pmatrix},
\end{equation}
and 
\begin{equation}\label{eq:delta:two-sample:remain}
    \boldsymbol{\delta}_i=\begin{pmatrix}
        [(\bar{Y}_1-\bar{Y}_2)-(\bar{Y}_1^*-\bar{Y}_2^*)](\Delta Y_i^*-W_i)\\
        \\
        [(\bar{Y}_1-\bar{Y}_2)-(\bar{Y}_1^*-\bar{Y}_2^*)]\Gamma_i^*+(\Delta g_i)W_i\\
        \\
        \displaystyle-[(\bar{Y}_1-\bar{Y}_2)-(\bar{Y}_1^*-\bar{Y}_2^*)]\left[2\left(\frac{Y_{i1}}{n_{i1}}-\frac{Y_{i2}}{n_{i2}}\right)-(\bar{Y}_1-\bar{Y}_2)-(\bar{Y}_1^*-\bar{Y}_2^*)\right]\\
        \\
        \begin{array}{rl}
        &\displaystyle-[(\bar{Y}_1-\bar{Y}_2)-(\bar{Y}_1^*-\bar{Y}_2^*)][(\hat{g}_{i1}-\hat{g}_{i2})-(\hat{g}_1-\hat{g}_2)]\\
        &\quad+\displaystyle\left[\left(\frac{Y_{i1}}{n_{i1}}-\frac{Y_{i2}}{n_{i2}}\right)-(\bar{Y}_1^*-\bar{Y}_2^*)\right]\Delta g_i
        \end{array}\\
        \\
        \begin{array}{rl}
        &\displaystyle-[(\bar{Y}_1-\bar{Y}_2)-(\bar{Y}_1^*-\bar{Y}_2^*)][(\hat{g}_{i1}-\hat{g}_{i2})-(\hat{g}_1-\hat{g}_2)]\\
        &\quad+\displaystyle\left[\left(\frac{Y_{i1}}{n_{i1}}-\frac{Y_{i2}}{n_{i2}}\right)-(\bar{Y}_1^*-\bar{Y}_2^*)\right]\Delta g_i
        \end{array}\\
        \\
        [(\hat{g}_{i1}-\hat{g}_{i2})-(\hat{g}_1-\hat{g}_2)+(g_{i1}-g_{i2})-(\bar{g}_1-\bar{g}_2)]\Delta g_i
    \end{pmatrix}.
\end{equation}
Note that 
\begin{equation}\label{eq:CLT:Y:diff}
\sqrt{N}\left[(\bar{Y}_1-\bar{Y}_2)-(\bar{Y}_1^*-\bar{Y}_2^*)\right]=\frac{\frac{1}{\sqrt{N}}\sum_{i=1}^N(Y_{i1}-\mathbb{E}[Y_{i1}])}{\frac{1}{N}\sum_{i=1}^Nn_{i1}}-\frac{\frac{1}{\sqrt{N}}\sum_{i=1}^N(Y_{i2}-\mathbb{E}[Y_{i2}])}{\frac{1}{N}\sum_{i=1}^Nn_{i2}}.
\end{equation}
Note that under (iii) of Assumption \ref{assump:two_sample}, for any $\ell\in\{1,2\}$, we have $\frac{1}{N}\sum_{i=1}^N\mathrm{Var}(Y_{i\ell})\rightarrow\sigma_{Y\ell}^2$,  
$\frac{1}{N}\sum_{i=1}^N\mathbb{E}[Y_{i\ell}]\rightarrow\mu_{Y\ell}^*$, and $\frac{1}{N}\sum_{i=1}^Nn_{i\ell}\rightarrow\mu_{n\ell}^*$. Further note that $\{Y_{i1}, Y_{i2}\}$ are independent across $i\in[N]$, so using Slutsky's theorem and Lindeberg-Feller Central Limit Theorem (Lemma \ref{lemma:Lindeberg-Feller:CLT:multivariate}), for some constant $\tilde{\sigma}_Y$ we have 
\begin{equation}\label{eq:CLT:Y:two-sample}
    \sqrt{N}\left[(\bar{Y}_1-\bar{Y}_2)-(\bar{Y}_1^*-\bar{Y}_2^*)\right]\rightsquigarrow\mathcal{N}\left(0,\tilde{\sigma}_Y^2\right).
\end{equation}
Note that $|\Delta Y_i^*-W_i|$, $|\Gamma_i^*|$, $\displaystyle\left|2\left(\frac{Y_{i1}}{n_{i1}}-\frac{Y_{i2}}{n_{i2}}\right)-(\bar{Y}_1-\bar{Y}_2)-(\bar{Y}_1^*-\bar{Y}_2^*)\right|$, $|(\hat{g}_{i1}-\hat{g}_{i2})-(\hat{g}_1-\hat{g}_2)|$ are uniformly bounded by some constant $c_0>0$, so we have 
$$\begin{array}{rl}
&\quad\displaystyle\sup_N\mathbb{E}\left[\left|\sqrt{N}[(\bar{Y}_1-\bar{Y}_2)-(\bar{Y}_1^*-\bar{Y}_2^*)]\frac{1}{N}\sum_{i=1}^N(\Delta Y_i^*-W_i)\right|\right]\\
&\displaystyle\leq c_0\sup_N\mathbb{E}\left[\left|\sqrt{N}[(\bar{Y}_1-\bar{Y}_2)-(\bar{Y}_1^*-\bar{Y}_2^*)\right|\right]\\
&\displaystyle\leq c_0\sup_N\mathbb{E}\left[\left|\sqrt{N}[(\bar{Y}_1-\bar{Y}_2)-(\bar{Y}_1^*-\bar{Y}_2^*)\right|^2\right]^{1/2}<\infty,
\end{array}$$
so $\displaystyle\sqrt{N}[(\bar{Y}_1-\bar{Y}_2)-(\bar{Y}_1^*-\bar{Y}_2^*)]\frac{1}{N}\sum_{i=1}^N(\Delta Y_i^*-W_i)$ is uniformly integrable. Similarly, all of the following terms are also uniformly integrable:
$$\sqrt{N}[(\bar{Y}_1-\bar{Y}_2)-(\bar{Y}_1^*-\bar{Y}_2^*)]\frac{1}{N}\sum_{i=1}^N\Gamma_i^*,$$
$$-\sqrt{N}[(\bar{Y}_1-\bar{Y}_2)-(\bar{Y}_1^*-\bar{Y}_2^*)]\frac{1}{N}\sum_{i=1}^N\left[2\left(\frac{Y_{i1}}{n_{i1}}-\frac{Y_{i2}}{n_{i2}}\right)-(\bar{Y}_1-\bar{Y}_2)-(\bar{Y}_1^*-\bar{Y}_2^*)\right],$$
$$-\sqrt{N}[(\bar{Y}_1-\bar{Y}_2)-(\bar{Y}_1^*-\bar{Y}_2^*)]\frac{1}{N}\sum_{i=1}^N[(\hat{g}_{i1}-\hat{g}_{i2})-(\hat{g}_1-\hat{g}_2)].$$
Thus using (iii) of Assumption \ref{assump:two_sample} and Slutsky's theorem, we have 
\begin{equation}\label{eq:vanish:mean:Y}
\begin{array}{rcl}
&\displaystyle\mathbb{E}\left[\sqrt{N}[(\bar{Y}_1-\bar{Y}_2)-(\bar{Y}_1^*-\bar{Y}_2^*)]\frac{1}{N}\sum_{i=1}^N(\Delta Y_i^*-W_i)\right]=\mathrm{o}(1)&\\
&\displaystyle\mathbb{E}\left[\sqrt{N}[(\bar{Y}_1-\bar{Y}_2)-(\bar{Y}_1^*-\bar{Y}_2^*)]\frac{1}{N}\sum_{i=1}^N\Gamma_i^*\right]=\mathrm{o}(1)&\\
&\displaystyle\mathbb{E}\left[\sqrt{N}[(\bar{Y}_1-\bar{Y}_2)-(\bar{Y}_1^*-\bar{Y}_2^*)]\frac{1}{N}\sum_{i=1}^N\!\!\left[2\left(\frac{Y_{i1}}{n_{i1}}-\frac{Y_{i2}}{n_{i2}}\right)\!-\!(\bar{Y}_1-\bar{Y}_2)\!-\!(\bar{Y}_1^*-\bar{Y}_2^*)\right]\right]=\mathrm{o}(1)&\\
&\displaystyle\mathbb{E}\left[\sqrt{N}[(\bar{Y}_1-\bar{Y}_2)-(\bar{Y}_1^*-\bar{Y}_2^*)]\frac{1}{N}\sum_{i=1}^N[(\hat{g}_{i1}-\hat{g}_{i2})-(\hat{g}_1-\hat{g}_2)]\right]=\mathrm{o}(1).&
\end{array}
\end{equation}
Thus for any $t_1,t_2,t_3,t_4,t_5\in\mathbb{R}$, we have 
\begin{equation}\label{eq:Y:delta:equiv}
\begin{array}{rl}
&\displaystyle\quad\sqrt{N}\left[(\bar{Y}_1-\bar{Y}_2)-(\bar{Y}_1^*-\bar{Y}_2^*)\right]\\
&\displaystyle\quad\times\frac{1}{N}\sum_{i=1}^N\Bigg\{t_1(\Delta Y_i^*-W_i)+t_2\Gamma_i^*-t_3\left[2\left(\frac{Y_{i1}}{n_{i1}}-\frac{Y_{i2}}{n_{i2}}\right)-(\bar{Y}_1-\bar{Y}_2)-(\bar{Y}_1^*-\bar{Y}_2^*)\right]\\
&\displaystyle\quad\quad\quad\quad\quad-(t_4+t_5)[(\hat{g}_{i1}-\hat{g}_{i2})-(\hat{g}_1-\hat{g}_2)]\Bigg\}\\
&\displaystyle=_{(a)}\sqrt{N}\left[(\bar{Y}_1-\bar{Y}_2)-(\bar{Y}_1^*-\bar{Y}_2^*)\right]\\
&\displaystyle\quad\times\frac{1}{N}\sum_{i=1}^N\Bigg\{t_1(\Delta Y_i^*-W_i)+t_2\Gamma_i^*-t_3\left[2\left(\frac{Y_{i1}}{n_{i1}}-\frac{Y_{i2}}{n_{i2}}\right)-(\bar{Y}_1-\bar{Y}_2)-(\bar{Y}_1^*-\bar{Y}_2^*)\right]\\
&\displaystyle\quad\quad\quad\quad\quad-(t_4+t_5)[(g_{i1}-g_{i2})-(\bar{g}_1-\bar{g}_2)]\Bigg\}+\mathrm{o}_p(1)\\
&\displaystyle=_{(b)}\frac{1}{\sqrt{N}}\sum_{i=1}^N\left\{\frac{(Y_{i1}-\mathbb{E}[Y_{i1}])}{\mu_{n1}^*}-\frac{\mu_{Y1}^*}{\mu_{n1}^*}(n_{i1}-\mathbb{E}[n_{i1}])-\frac{(Y_{i2}-\mathbb{E}[Y_{i2}])}{\mu_{n2}^*}+\frac{\mu_{Y2}^*}{\mu_{n2}^*}(n_{i2}-\mathbb{E}[n_{i2}])\right\}\\
&\displaystyle\quad\times\Bigg\{t_1\left[(\mu_{\theta1}^*-\mu_{\theta2^*})-2\left(\frac{\mu_{Y1}^*}{\mu_{n1}^*}-\frac{\mu_{Y2}^*}{\mu_{n2}^*}\right)-2\mu_{I1}^*+2\mu_{IY,1}^*-2\mu_{IY,2}^*+\mu_{\theta2}*+2\mu_{I2}^*-1\right]\\
&\displaystyle\quad\quad\quad+t_2\left[(\mu_{g1}^*-\mu_{g2}^*)-\left(\frac{\mu_{gn,1}^*}{\mu_{n1}^*}-\frac{\mu_{gn,2}^*}{\mu_{n2}^*}\right)\right]-t_3\left[2\left(\mu_{\theta1}^*-\mu_{\theta2}^*\right)-2\left(\frac{\mu_{Y1}^*}{\mu_{n1}^*}-\frac{\mu_{Y2}^*}{\mu_{n2}^*}\right)\right]\\
&\displaystyle\quad\quad\quad-(t_4+t_5)\left[(\mu_{g1}^*-\mu_{g2}^*)-\left(\frac{\mu_{gn,1}^*}{\mu_{n1}^*}-\frac{\mu_{gn,2}^*}{\mu_{n2}^*}\right)\right]\Bigg\}+\mathrm{o}_p(1)\\
\end{array}
\end{equation}
where (a) uses condition (ii) of Assumption \ref{assump:two_sample}, (b) follows from \eqref{eq:CLT:Y:diff}, condition (iii) of Assumption \ref{assump:two_sample}, Law of Large Numbers, and Slutsky's theorem. Further, note that 
\begin{equation}\label{eq:W:g}
\begin{array}{rl}
&\displaystyle\quad\frac{1}{\sqrt{N}}\sum_{i=1}^NW_i\{\Delta g_i\}\\
&\displaystyle=\frac{1}{\sqrt{N}}\sum_{i=1}^NW_i[(\hat{g}_{i1}-g_{i1})-(\hat{g}_{i2}-g_{i2})]-[(\hat{g}_1-\bar{g}_1)-(\hat{g}_2-\bar{g}_2)]\frac{1}{\sqrt{N}}\sum_{i=1}^NW_i,
\end{array}
\end{equation}
where 
\begin{equation}\label{eq:W:g:decompose}
\begin{array}{rl}
&\displaystyle\quad\frac{1}{\sqrt{N}}\sum_{i=1}^NW_i[(\hat{g}_{i1}-g_{i1})-(\hat{g}_{i2}-g_{i2})]\\
&\displaystyle=\sum_{k=1}^K\frac{1}{\sqrt{N}}\sum_{i\in\textrm{Fold}(k)}^NW_i[(\hat{g}_{1}^{-k}(\mathbf{X}_{i1})-g_{1}(\mathbf{X}_{i1}))-(\hat{g}_{2}^{-k}(\mathbf{X}_{i2})-g_{2}(\mathbf{X}_{i2}))]\\
&\displaystyle=\sum_{k=1}^K\frac{1}{\sqrt{N}}\sum_{i\in\textrm{Fold}(k)}^NW_i^*[(\hat{g}_{1}^{-k}(\mathbf{X}_{i1})-g_{1}(\mathbf{X}_{i1}))-(\hat{g}_{2}^{-k}(\mathbf{X}_{i2})-g_{2}(\mathbf{X}_{i2}))]\\
&\displaystyle\quad+\sum_{k=1}^K\frac{1}{\sqrt{N}}\sum_{i\in\textrm{Fold}(k)}^N(W_i-W_i^*)[(\hat{g}_{1}^{-k}(\mathbf{X}_{i1})-g_{1}(\mathbf{X}_{i1}))-(\hat{g}_{2}^{-k}(\mathbf{X}_{i2})-g_{2}(\mathbf{X}_{i2}))]\\
\end{array}
\end{equation}
Since $W_i^*\left[(\hat{g}_{1}^{-k}(\mathbf{X}_{i1})-g_{1}(\mathbf{X}_{i1}))-(\hat{g}_{2}^{-k}(\mathbf{X}_{i2})-g_{2}(\mathbf{X}_{i2}))\right]$ are independent across $i\in\textrm{Fold}(k)$, so we have 
\begin{equation}\label{eq:W:g:1}
\begin{array}{rl}
&\displaystyle\quad\mathrm{Var}\left[\frac{1}{\sqrt{N}}\sum_{i\in\textrm{Fold}(k)}^NW_i^*\left[(\hat{g}_{1}^{-k}(\mathbf{X}_{i1})-g_{1}(\mathbf{X}_{i1}))-(\hat{g}_{2}^{-k}(\mathbf{X}_{i2})-g_{2}(\mathbf{X}_{i2}))\right]\right]\\
&\displaystyle=\frac{1}{N}\sum_{i\in\textrm{Fold}(k)}^N\mathrm{Var}\left\{W_i^*\left[(\hat{g}_{1}^{-k}(\mathbf{X}_{i1})-g_{1}(\mathbf{X}_{i1}))-(\hat{g}_{2}^{-k}(\mathbf{X}_{i2})-g_{2}(\mathbf{X}_{i2}))\right]\right\}\\
&\displaystyle\leq\frac{1}{N}\sum_{i\in\textrm{Fold}(k)}^N\mathbb{E}\left[\left|W_i^*\left[(\hat{g}_{1}^{-k}(\mathbf{X}_{i1})-g_{1}(\mathbf{X}_{i1}))-(\hat{g}_{2}^{-k}(\mathbf{X}_{i2})-g_{2}(\mathbf{X}_{i2}))\right]\right|^2\right]\\
&\displaystyle\leq_{(a)}\frac{c_0^2}{N}\sum_{i\in\textrm{Fold}(k)}^N\mathbb{E}\left[\left|\left[(\hat{g}_{1}^{-k}(\mathbf{X}_{i1})-g_{1}(\mathbf{X}_{i1}))-(\hat{g}_{2}^{-k}(\mathbf{X}_{i2})-g_{2}(\mathbf{X}_{i2}))\right]\right|^2\right]\\
&\displaystyle\leq_{(b)}\frac{c_0^2}{N}\sum_{i\in\textrm{Fold}(k)}^N2\left\{\mathbb{E}\left[|\hat{g}_{1}^{-k}(\mathbf{X}_{i1})-g_{1}(\mathbf{X}_{i1})|^2\right]+\mathbb{E}\left[|\hat{g}_{2}^{-k}(\mathbf{X}_{i2})-g_{2}(\mathbf{X}_{i2})|^2\right]\right\}\\
&=_{(c)}\mathrm{o}(1),
\end{array}
\end{equation}
where in \eqref{eq:W:g:1}, (a) follows since $|W_i^*|\leq c_0$ for some constant $c_0$, (b) follows from the inequality $(a-b)^2\leq 2(a^2+b^2)$, (c) follows from (ii) of Assumption \ref{assump:two_sample}. Additionally, 
\begin{equation}\label{eq:W:g:2}
\begin{array}{rl}
&\displaystyle\quad\sum_{k=1}^K\frac{1}{\sqrt{N}}\sum_{i\in\textrm{Fold}(k)}^N(W_i-W_i^*)[(\hat{g}_{1}^{-k}(\mathbf{X}_{i1})-g_{1}(\mathbf{X}_{i1}))-(\hat{g}_{2}^{-k}(\mathbf{X}_{i2})-g_{2}(\mathbf{X}_{i2}))]\\
&\displaystyle=\sqrt{N}\left[(\bar{Y}_1-\bar{Y}_2)-(\bar{Y}_1^*-\bar{Y}_2^*)\right]\frac{1}{N}\sum_{i=1}^N\left[(\hat{g}_{1}(\mathbf{X}_{i1})-g_{1}(\mathbf{X}_{i1}))-(\hat{g}_{2}(\mathbf{X}_{i2})-g_{2}(\mathbf{X}_{i2}))\right]\\
&=_{(a)}\mathrm{o}_p(1),
\end{array}
\end{equation}
where (a) of \eqref{eq:W:g:2} follows because $\sqrt{N}\left[(\bar{Y}_1-\bar{Y}_2)-(\bar{Y}_1^*-\bar{Y}_2^*)\right]=\mathrm{O}_p(1)$ due to \eqref{eq:CLT:Y:two-sample}, and 
$\displaystyle\frac{1}{N}\sum_{i=1}^N\left[(\hat{g}_{1}(\mathbf{X}_{i1})-g_{1}(\mathbf{X}_{i1}))-(\hat{g}_{2}(\mathbf{X}_{i2})-g_{2}(\mathbf{X}_{i2}))\right]=\mathrm{o}_p(1)$ due to condition (ii) of Assumption \ref{assump:two_sample}. Thus according to \eqref{eq:W:g:decompose} we have 
\begin{equation}\label{eq:W:g:vanish:1}
    \frac{1}{\sqrt{N}}\sum_{i=1}^N\left\{W_i[(\hat{g}_{i1}-g_{i1})-(\hat{g}_{i2}-g_{i2})]-\mathbb{E}\left(W_i[(\hat{g}_{i1}-g_{i1})-(\hat{g}_{i2}-g_{i2})]\right)\right\}=\mathrm{o}_p(1).
\end{equation}
Further, 
\begin{equation}\label{eq:g:w:equiv}
\begin{array}{rl}
\displaystyle[(\hat{g}_1-\bar{g}_1)-(\hat{g}_2-\bar{g}_2)]\frac{1}{\sqrt{N}}\sum_{i=1}^NW_i\!\!\!\!&\displaystyle=[(\hat{g}_1-\bar{g}_1)-(\hat{g}_2-\bar{g}_2)]\frac{1}{\sqrt{N}}\sum_{i=1}^NW_i^*\\
&\quad\displaystyle+\sqrt{N}\left[(\bar{Y}_1-\bar{Y}_2)-(\bar{Y}_1^*-\bar{Y}_2^*)\right][(\hat{g}_1-\bar{g}_1)-(\hat{g}_2-\bar{g}_2)]\\
&\displaystyle=_{(a)}[(\hat{g}_1-\bar{g}_1)-(\hat{g}_2-\bar{g}_2)]\frac{1}{\sqrt{N}}\sum_{i=1}^NW_i^*+\mathrm{o}_p(1),
\end{array}
\end{equation}
where (a) follows by noting that $\sqrt{N}\left[(\bar{Y}_1-\bar{Y}_2)-(\bar{Y}_1^*-\bar{Y}_2^*)\right]=\mathrm{O}_p(1)$ according to \eqref{eq:CLT:Y:two-sample}, $\left[(\hat{g}_1-\bar{g}_1)-(\hat{g}_2-\bar{g}_2)\right]=\mathrm{o}_p(1)$ according to (ii) of Assumption \ref{assump:two_sample}. Note that 
$$\begin{array}{rl}
&\displaystyle\quad\left[(\hat{g}_1-\bar{g}_1)-(\hat{g}_2-\bar{g}_2)\right]\frac{1}{\sqrt{N}}\sum_{i=1}^NW_i^*\\
&\displaystyle=\left[\frac{\frac{1}{\sqrt{N}}\sum_{i=1}^Nn_{i1}\{\hat{g}_{1}(\mathbf{X}_{i1})-g_{1}(\mathbf{X}_{i1})\}}{\frac{1}{N}\sum_{i=1}^Nn_{i1}}-\frac{\frac{1}{\sqrt{N}}\sum_{i=1}^Nn_{i2}\{\hat{g}_{2}(\mathbf{X}_{i2})-g_{2}(\mathbf{X}_{i2})\}}{\frac{1}{N}\sum_{i=1}^Nn_{i2}}\right]\frac{1}{N}\sum_{i=1}^NW_i^*.
\end{array}$$
For any $\ell\in\{0,1\}$,
$$\begin{array}{rl}
&\displaystyle\quad\frac{1}{\sqrt{N}}\sum_{i=1}^Nn_{i\ell}\{\hat{g}_{\ell}(\mathbf{X}_{i\ell})-g_{\ell}(\mathbf{X}_{i\ell})\}-\mathbb{E}\left[n_{i\ell}\{\hat{g}_{\ell}(\mathbf{X}_{i\ell})-g_{\ell}(\mathbf{X}_{i\ell})\}\right]\\
&\displaystyle=\sum_{k=1}^K\frac{1}{\sqrt{N}}\sum_{i\in\textrm{Fold}(k)}\{\hat{g}_{\ell}^{-k}(\mathbf{X}_{i\ell})-g_{\ell}(\mathbf{X}_{i\ell})\}-\mathbb{E}\left[n_{i\ell}\{\hat{g}_{\ell}^{-k}(\mathbf{X}_{i\ell})-g_{\ell}(\mathbf{X}_{i\ell})\}\right],
\end{array}$$
where 
$$\mathrm{Var}\left\{\frac{1}{\sqrt{N}}\sum_{i\in\textrm{Fold}(k)}\{\hat{g}_{\ell}^{-k}(\mathbf{X}_{i\ell})-g_{\ell}(\mathbf{X}_{i\ell})\}-\mathbb{E}\left[n_{i\ell}\{\hat{g}_{\ell}^{-k}(\mathbf{X}_{i\ell})-g_{\ell}(\mathbf{X}_{i\ell})\}\right]\right\}=\mathrm{o}(1).$$
So
\begin{equation}\label{eq:n:g:vanish}
    \frac{1}{\sqrt{N}}\sum_{i=1}^Nn_{i\ell}\{\hat{g}_{\ell}(\mathbf{X}_{i\ell})-g_{\ell}(\mathbf{X}_{i\ell})\}-\mathbb{E}\left[n_{i\ell}\{\hat{g}_{\ell}(\mathbf{X}_{i\ell})-g_{\ell}(\mathbf{X}_{i\ell})\}\right]=\mathrm{o}_p(1).
\end{equation}
According to condition (ii) of Assumption \ref{assump:two_sample}. Further note that for any $\ell\in\{1,2\}$, we have 
\begin{equation}\label{eq:g:diff:vanish:expectation}
\begin{array}{rl}
&\quad\displaystyle\mathbb{E}\left[\frac{\frac{1}{\sqrt{N}}\sum_{i=1}^Nn_{i\ell}\{\hat{g}_{\ell}(\mathbf{X}_{i\ell})-g_{\ell}(\mathbf{X}_{i\ell})\}}{\frac{1}{N}\sum_{i=1}^Nn_{i\ell}}\frac{1}{N}\sum_{i=1}^NW_i^*\right]\\
&\displaystyle\quad\quad-\frac{\frac{1}{\sqrt{N}}\sum_{i=1}^N\mathbb{E}[n_{i\ell}\{\hat{g}_{\ell}(\mathbf{X}_{i\ell})-g_{\ell}(\mathbf{X}_{i\ell})\}]}{\frac{1}{N}\sum_{i=1}^Nn_{i\ell}}\left\{\frac{1}{N}\sum_{i=1}^N\mathbb{E}[W_i^*]\right\}\\
&\displaystyle=\mathbb{E}\left[\frac{\frac{1}{\sqrt{N}}\sum_{i=1}^Nn_{i\ell}\{\hat{g}_{\ell}(\mathbf{X}_{i\ell})-g_{\ell}(\mathbf{X}_{i\ell})\}-\mathbb{E}[n_{i\ell}\{\hat{g}_{\ell}(\mathbf{X}_{i\ell})-g_{\ell}(\mathbf{X}_{i\ell})\}]}{\frac{1}{N}\sum_{i=1}^Nn_{i\ell}}\frac{1}{N}\sum_{i=1}^NW_i^*\right].
\end{array}
\end{equation}
Note that 
$$\frac{1}{N}\sum_{i=1}^Nn_{i\ell}\{\hat{g}_{\ell}(\mathbf{X}_{i\ell})-g_{\ell}(\mathbf{X}_{i\ell})\}-\mathbb{E}[n_{i\ell}\{\hat{g}_{\ell}(\mathbf{X}_{i\ell})-g_{\ell}(\mathbf{X}_{i\ell})\}]=\mathrm{o}_p(1)$$ 
according to \eqref{eq:n:g:vanish}. Also note that $\frac{1}{N}\sum_{i=1}^NW_i^*=\mathrm{O}_p(1)$, $\frac{1}{N}\sum_{i=1}^Nn_{i\ell}=\mathrm{O}_p(1)$, thus \eqref{eq:g:diff:vanish:expectation} implies that 
$$\begin{array}{rl}
&\displaystyle\quad\mathbb{E}\left[\frac{\frac{1}{\sqrt{N}}\sum_{i=1}^Nn_{i\ell}\{\hat{g}_{\ell}(\mathbf{X}_{i\ell})-g_{\ell}(\mathbf{X}_{i\ell})\}}{\frac{1}{N}\sum_{i=1}^Nn_{i\ell}}\frac{1}{N}\sum_{i=1}^NW_i^*\right]\\
&\displaystyle=\frac{\frac{1}{\sqrt{N}}\sum_{i=1}^N\mathbb{E}[n_{i\ell}\{\hat{g}_{\ell}(\mathbf{X}_{i\ell})-g_{\ell}(\mathbf{X}_{i\ell})\}]}{\frac{1}{N}\sum_{i=1}^Nn_{i\ell}}\left\{\frac{1}{N}\sum_{i=1}^N\mathbb{E}[W_i^*]\right\}+\mathrm{o}_p(1).
\end{array}$$
Hence we have 
\begin{equation}\label{eq:g:w:vanish:1}
\begin{array}{rl}
\displaystyle\left[(\hat{g}_1-\bar{g}_1)-(\hat{g}_2-\bar{g}_2)\right]\frac{1}{\sqrt{N}}\sum_{i=1}^NW_i^*-\mathbb{E}\left\{\left[(\hat{g}_1-\bar{g}_1)-(\hat{g}_2-\bar{g}_2)\right]\frac{1}{\sqrt{N}}\sum_{i=1}^NW_i^*\right\}=\mathrm{o}_p(1). 
\end{array}
\end{equation}
Thus \eqref{eq:W:g},\eqref{eq:W:g:vanish:1}, \eqref{eq:g:w:equiv}, \eqref{eq:g:w:vanish:1} imply that
\begin{equation}\label{eq:delta:W:g:asymptotics}
    \frac{1}{\sqrt{N}}\sum_{i=1}^N\left\{W_i\{\Delta g_i\}-\mathbb{E}\left[W_i\{\Delta g_i\}\right]\right\}=\mathrm{o}_p(1).
\end{equation}
Furthermore, note that 
\begin{equation}\label{eq:delta:Y:g:term}
\begin{array}{rl}
    &\quad\displaystyle\frac{1}{\sqrt{N}}\sum_{i=1}^N\left[\left(\frac{Y_{i1}}{n_{i1}}-\frac{Y_{i2}}{n_{i2}}\right)-(\bar{Y}_1^*-\bar{Y}_2^*)\right]\Delta g_i\\
    &\displaystyle=\frac{1}{\sqrt{N}}\sum_{i=1}^N\left[\left(\frac{Y_{i1}}{n_{i1}}-\frac{Y_{i2}}{n_{i2}}\right)-(\bar{Y}_1^*-\bar{Y}_2^*)\right][(\hat{g}_{i1}-g_{i1})-(\hat{g}_{i2}-g_{i2})]\\
    &\quad\displaystyle-[(\hat{g}_1-\bar{g}_1)-(\hat{g}_2-\bar{g}_2)]\frac{1}{\sqrt{N}}\sum_{i=1}^N\left[\left(\frac{Y_{i1}}{n_{i1}}-\frac{Y_{i2}}{n_{i2}}\right)-(\bar{Y}_1^*-\bar{Y}_2^*)\right].
\end{array}
\end{equation}
Using similar proof steps as for \eqref{eq:W:g:decompose} -- \eqref{eq:W:g:vanish:1}, we have 
\begin{equation}\label{eq:Y:g:vanish:1}
\begin{array}{rl}
&\displaystyle\frac{1}{\sqrt{N}}\sum_{i=1}^N\Bigg\{\left[\left(\frac{Y_{i1}}{n_{i1}}-\frac{Y_{i2}}{n_{i2}}\right)-(\bar{Y}_1^*-\bar{Y}_2^*)\right][(\hat{g}_{i1}-g_{i1})-(\hat{g}_{i2}-g_{i2})]\\
&\quad\quad\quad\quad\quad\displaystyle-\mathbb{E}\left(\left[\left(\frac{Y_{i1}}{n_{i1}}-\frac{Y_{i2}}{n_{i2}}\right)-(\bar{Y}_1^*-\bar{Y}_2^*)\right][(\hat{g}_{i1}-g_{i1})-(\hat{g}_{i2}-g_{i2})]\right)\Bigg\}=\mathrm{o}_p(1).
\end{array}
\end{equation}
Using similar proof steps as for \eqref{eq:W:g} -- \eqref{eq:g:w:vanish:1}, we have 
\begin{equation}\label{eq:g:Y:vanish:2}
\begin{array}{rl}
    \displaystyle[(\hat{g}_1-\bar{g}_1)-(\hat{g}_2-\bar{g}_2)]\frac{1}{\sqrt{N}}\sum_{i=1}^N\left[\left(\frac{Y_{i1}}{n_{i1}}-\frac{Y_{i2}}{n_{i2}}\right)-(\bar{Y}_1^*-\bar{Y}_2^*)\right]=\mathrm{o}_p(1),
\end{array}
\end{equation}
So \eqref{eq:delta:Y:g:term}, \eqref{eq:Y:g:vanish:1}, \eqref{eq:g:Y:vanish:2} imply that  
\begin{equation}\label{eq:delta:Y:g:equiv:asymptotics}
\begin{array}{rl}
    &\quad\displaystyle\frac{1}{\sqrt{N}}\sum_{i=1}^N\left[\left(\frac{Y_{i1}}{n_{i1}}-\frac{Y_{i2}}{n_{i2}}\right)-(\bar{Y}_1^*-\bar{Y}_2^*)\right]\Delta g_i-\mathbb{E}\left\{\left[\left(\frac{Y_{i1}}{n_{i1}}-\frac{Y_{i2}}{n_{i2}}\right)-(\bar{Y}_1^*-\bar{Y}_2^*)\right]\Delta g_i\right\}\\
    &\displaystyle=\mathrm{o}_p(1).
\end{array}
\end{equation}
Additionally, note that 
\begin{equation}\label{eq:delta:g:g:term}
\begin{array}{rl}
&\quad\displaystyle\frac{1}{\sqrt{N}}\sum_{i=1}^N[(\hat{g}_{i1}-\hat{g}_{i2})-(\hat{g}_1-\hat{g}_2)+(g_{i1}-g_{i2})-(\bar{g}_1-\bar{g}_2)]\Delta g_i\\
&\displaystyle=\frac{1}{\sqrt{N}}\sum_{i=1}^N[(\hat{g}_{i1}-\hat{g}_{i2})-(\hat{g}_1-\hat{g}_2)+(g_{i1}-g_{i2})-(\bar{g}_1-\bar{g}_2)][(\hat{g}_{i1}-g_{i1})-(\hat{g}_{i2}-g_{i2})]\\
&\quad\displaystyle-[(\hat{g}_1-\bar{g}_1)-(\hat{g}_2-\bar{g}_2)]\frac{1}{\sqrt{N}}\sum_{i=1}^N[(\hat{g}_{i1}-\hat{g}_{i2})-(\hat{g}_1-\hat{g}_2)+(g_{i1}-g_{i2})-(\bar{g}_1-\bar{g}_2)].
\end{array}
\end{equation}
Using similar proof steps as for \eqref{eq:W:g:decompose} -- \eqref{eq:W:g:vanish:1}, we have
\begin{equation}\label{eq:g:g:vanish:1}
\begin{array}{rl}
&\quad\displaystyle\frac{1}{\sqrt{N}}\sum_{i=1}^N[(\hat{g}_{i1}-\hat{g}_{i2})-(\hat{g}_1-\hat{g}_2)+(g_{i1}-g_{i2})-(\bar{g}_1-\bar{g}_2)][(\hat{g}_{i1}-g_{i1})-(\hat{g}_{i2}-g_{i2})]\\
&\displaystyle\quad\quad\quad\quad-\mathbb{E}\left\{[(\hat{g}_{i1}-\hat{g}_{i2})-(\hat{g}_1-\hat{g}_2)+(g_{i1}-g_{i2})-(\bar{g}_1-\bar{g}_2)][(\hat{g}_{i1}-g_{i1})-(\hat{g}_{i2}-g_{i2})]\right\}\\
&=\mathrm{o}_p(1).
\end{array}
\end{equation}
Using similar proof steps as for \eqref{eq:W:g:decompose} -- \eqref{eq:W:g:vanish:1}, we have 
\begin{equation}\label{eq:g:g:vanish:2}
\begin{array}{rl}
&\quad\displaystyle[(\hat{g}_1-\bar{g}_1)-(\hat{g}_2-\bar{g}_2)]\frac{1}{\sqrt{N}}\sum_{i=1}^N[(\hat{g}_{i1}-\hat{g}_{i2})-(\hat{g}_1-\hat{g}_2)+(g_{i1}-g_{i2})-(\bar{g}_1-\bar{g}_2)]\\
&\displaystyle=\mathrm{o}_p(1),
\end{array}
\end{equation}
Thus \eqref{eq:delta:g:g:term}, \eqref{eq:g:g:vanish:1}, \eqref{eq:g:g:vanish:2} imply that 
\begin{equation}\label{eq:delta:g:g:equiv:asymptotics}
\begin{array}{rl}
&\quad\displaystyle\frac{1}{\sqrt{N}}\sum_{i=1}^N\big\{[(\hat{g}_{i1}-\hat{g}_{i2})-(\hat{g}_1-\hat{g}_2)+(g_{i1}-g_{i2})-(\bar{g}_1-\bar{g}_2)]\Delta g_i\\
&\displaystyle\quad\quad\quad\quad\quad-\mathbb{E}\left([(\hat{g}_{i1}-\hat{g}_{i2})-(\hat{g}_1-\hat{g}_2)+(g_{i1}-g_{i2})-(\bar{g}_1-\bar{g}_2)]\Delta g_i\right)\big\}\\
&\displaystyle=\mathrm{o}_p(1).
\end{array}
\end{equation}
Furthermore, denote
$$\Delta_{i,Yn}=\frac{Y_{i1}}{\mu_{n1}^*}-\frac{Y_{i2}}{\mu_{n2}^*}.$$
So \eqref{eq:Y:delta:equiv} implies that 
\begin{equation}\label{eq:Y:delta:equiv:notation}
\begin{array}{rl}
&\displaystyle\quad\sqrt{N}\left[(\bar{Y}_1-\bar{Y}_2)-(\bar{Y}_1^*-\bar{Y}_2^*)\right]\\
&\displaystyle\quad\times\frac{1}{N}\sum_{i=1}^N\Bigg\{t_1(\Delta Y_i^*-W_i)+t_2\Gamma_i^*-t_3\left[2\left(\frac{Y_{i1}}{n_{i1}}-\frac{Y_{i2}}{n_{i2}}\right)-(\bar{Y}_1-\bar{Y}_2)-(\bar{Y}_1^*-\bar{Y}_2^*)\right]\\
&\displaystyle\quad\quad\quad\quad\quad-(t_4+t_5)[(\hat{g}_{i1}-\hat{g}_{i2})-(\hat{g}_1-\hat{g}_2)]\Bigg\}\\
&\displaystyle=\frac{1}{\sqrt{N}}\{\Delta_{i,Yn}-\mathbb{E}[\Delta_{i,Yn}]\}\{t_1[(\mu_{\theta1}^*-\mu_{\theta2}^*)-2\kappa^*]+t_2\kappa_3^*-2t_3\kappa_2^*-(t_4+t_5)\kappa_3^*\},
\end{array}
\end{equation}
where
$$\kappa^*=\frac{\mu_{Y1}^*}{\mu_{n1}^*}-\frac{\mu_{Y2}^*}{\mu_{n2}^*}+2(\mu_{I1}^*-\mu_{IY,1}^*)+2\mu_{IY,2}^*-\mu_{\theta2}^*-2\mu_{I2}^*+1,$$
$$\kappa_2^*=\left(\mu_{\theta1}^*-\mu_{\theta2}^*\right)-\left(\frac{\mu_{Y1}^*}{\mu_{n1}^*}-\frac{\mu_{Y2}^*}{\mu_{n2}^*}\right),\kappa_3^*=\left[(\mu_{g1}^*-\mu_{g2}^*)-\left(\frac{\mu_{gn,1}^*}{\mu_{n1}^*}-\frac{\mu_{gn,2}^*}{\mu_{n2}^*}\right)\right].$$

Further, according to \eqref{eq:two-sample:D:rewrite}, \eqref{eq:xi:two-sample:inid}, \eqref{eq:delta:two-sample:remain}, \eqref{eq:vanish:mean:Y}, \eqref{eq:Y:delta:equiv}, \eqref{eq:delta:W:g:asymptotics}, \eqref{eq:delta:Y:g:equiv:asymptotics}, \eqref{eq:delta:g:g:equiv:asymptotics}, given any $\mathbf{t}=(t_1,t_2,t_3,t_4,t_5,t_6)\in\mathbb{R}^6$,
\begin{equation}\label{eq:t:objective}
\mathbf{t}^T\left\{\frac{1}{\sqrt{N}}\sum_{i=1}^N\{\tilde{\boldsymbol{\xi}}_i-\mathbb{E}[\tilde{\boldsymbol{\xi}}_i]+\boldsymbol{\delta}_i-\mathbb{E}[\boldsymbol{\delta}_i]\}\right\}=\mathbf{t}^T\left\{\frac{1}{\sqrt{N}}\sum_{i=1}^N\left\{\overline{\boldsymbol{\zeta}}_i-\mathbb{E}[\overline{\boldsymbol{\zeta}}_i]\right\}\right\}+\mathrm{o}_p(1),
\end{equation}
where $\bar{\boldsymbol{\zeta}}_i=(\bar{\zeta}_{i1},\bar{\zeta}_{i2},\bar{\zeta}_{i3},\bar{\zeta}_{i4},\bar{\zeta}_{i5},\bar{\zeta}_{i6})^T$ are i.n.i.d. across $i$, such that 
$$\bar{\zeta}_{i1}=2W_i^*\Delta Y_i^*+2I_{i1}v_{i1}+(2I_{i2}-1)v_{i2}+[(\mu_{\theta1}^*-\mu_{\theta2}^*)-2\kappa^*]\Delta_{i,Yn},$$
$$\bar{\zeta}_{i2}=2W_i^*\Gamma_i^*+\kappa_3^*\Delta_{i,Yn},$$
$$\bar{\zeta}_{i3}=\left[\left(\frac{Y_{i1}}{n_{i1}}-\frac{Y_{i2}}{n_{i2}}\right)-(\bar{Y}_1^*-\bar{Y}_2^*)\right]^2-2\kappa_2^*\Delta_{i,Yn},$$
$$\bar{\zeta}_{i4}=\left[\left(\frac{Y_{i1}}{n_{i1}}-\frac{Y_{i2}}{n_{i2}}\right)-(\bar{Y}_1^*-\bar{Y}_2^*)\right][(g_{i1}-g_{i2})-(\bar{g}_1-\bar{g}_2)]-\kappa_3^*\Delta_{i,Yn},$$
$$\bar{\zeta}_{i5}=\left[\left(\frac{Y_{i1}}{n_{i1}}-\frac{Y_{i2}}{n_{i2}}\right)-(\bar{Y}_1^*-\bar{Y}_2^*)\right][(g_{i1}-g_{i2})-(\bar{g}_1-\bar{g}_2)]-\kappa_3^*\Delta_{i,Yn},$$
$$\bar{\zeta}_{i6}=[(g_{i1}-g_{i2})-(\bar{g}_1-\bar{g}_2)]^2.$$
According to (iv) of Assumption \ref{assump:two_sample}, for any $\mathbf{t}\in\mathbb{R}^6$ we have 
$$\mathbf{t}^T\left\{\frac{1}{\sqrt{N}}\sum_{i=1}^N\left\{\overline{\boldsymbol{\zeta}}_i-\mathbb{E}[\overline{\boldsymbol{\zeta}}_i]\right\}\right\}\rightsquigarrow\mathcal{N}\left(0,\mathbf{t}^T\overline{\boldsymbol{\Sigma}}\mathbf{t}\right),$$
so according to Lemma \ref{lemma:cramer-wold} and \eqref{eq:t:objective}, we have 
$$\frac{1}{\sqrt{N}}\sum_{i=1}^N\{\tilde{\boldsymbol{\xi}}_i-\mathbb{E}[\tilde{\boldsymbol{\xi}}_i]+\boldsymbol{\delta}_i-\mathbb{E}[\boldsymbol{\delta}_i]\}\rightsquigarrow\mathcal{N}\left(\mathbf{0},\overline{\boldsymbol{\Sigma}}\right),$$
hence \eqref{eq:CLT:main:two-sample} holds, and it's easy to check that $\overline{\boldsymbol{\Sigma}}\preceq\tilde{C}\mathbf{I}_6$ for some absolute constant $\tilde{C}$ by checking that each term in $\max_{\ell\in[6]}|\bar{\zeta}_{i\ell}|\leq\tilde{c}$ where $\tilde{c}$ is an absolute constant. Thus we have proved the lemma.
\end{proof}

\subsection{Proof of the Regret Bound}
According to \eqref{eq:two-sample:estimator:lambda:one-sample}, the one-sample binomial shrinkage estimator can be written as
$$\displaystyle\hat{\theta}_i^{\mathrm{o}}(\mathbf{Y};\boldsymbol{\lambda})=\boldsymbol{\lambda}^T\mathbf{F}_i+\bar{Y},\ \forall i\in[N],$$
where 
\begin{equation}\label{eq:notation:one-sample}
\mathbf{F}_i:=\begin{bmatrix}
\displaystyle\frac{Y_i}{n_i}-\bar{Y}\\
\displaystyle\hat{g}_i-\hat{g}
\end{bmatrix},\quad \bar{Y}=\frac{\sum_{i=1}^NY_i}{\sum_{i=1}^Nn_i},\quad \hat{g}_i=\hat{g}_i(\mathbf{X}_i),\quad \hat{g}=\frac{\sum_{j=1}^Nn_j\hat{g}_j(\mathbf{X}_j)}{\sum_{j=1}^Nn_j}.
\end{equation}
Denote $\hat{\theta}^{\mathrm{o}}(\boldsymbol{\lambda}):=\left\{\hat{\theta}_i^{\mathrm{o}}(\mathbf{Y};\boldsymbol{\lambda})\right\}_{i\in[N]}$. For any $\boldsymbol{\lambda}\in[0,1]\times\mathbb{R}$, define 
\begin{equation}\label{eq:loss:lambda:one-sample}
    \mathcal{L}_{\mathrm{o}}(\boldsymbol{\lambda}):=L_2^{\mathrm{o}}\left(\hat{\theta}^{\mathrm{o}}(\boldsymbol{\lambda});\theta\right)=\sum_{i=1}^N\mathbb{E}\left[\left(\boldsymbol{\lambda}^T\mathbf{F}_i+\bar{Y}\right)^2\right]-2\theta_i\mathbb{E}\left[\boldsymbol{\lambda}^T\mathbf{F}_i+\bar{Y}\right],
\end{equation}
where $L_2^{\mathrm{o}}(\cdot;\cdot)$ is the $L_2$ risk defined in \eqref{eq:unweighted:one-sample}. According to \eqref{eq:two-sample:estimator:lambda}, the two-sample binomial shrinkage estimator can be written as 
$$\displaystyle\hat{\theta}_i^{\mathrm{t}}(\mathbf{Y};\boldsymbol{\lambda})=\boldsymbol{\lambda}^T\mathbf{K}_i+(\bar{Y}_1-\bar{Y}_2),\ \forall i\in[N],$$
where
\begin{equation}\label{eq:notation:two-sample}
\begin{array}{rcl}
&\displaystyle\mathbf{K}_i=\begin{bmatrix}
\displaystyle\left\{\frac{Y_{i1}}{n_{i1}}-\frac{Y_{i2}}{n_{i2}}\right\}-(\bar{Y}_1-\bar{Y}_2)\\
\left\{\hat{g}_{i1}-\hat{g}_1\right\}-\{\hat{g}_{i2}-\hat{g}_2\}
\end{bmatrix},&\\
\\
&\displaystyle\bar{Y}_{\ell}=\frac{\sum_{i=1}^NY_{i\ell}}{\sum_{i=1}^Nn_{i\ell}},\ \hat{g}_{i\ell}=\hat{g}_{i\ell}(\mathbf{X}_{i\ell}),\ \hat{g}_{\ell}=\frac{\sum_{j=1}^Nn_{j\ell}\hat{g}_{j\ell}}{\sum_{j=1}^Nn_{j\ell}},\ \ \forall \ell\in\{1,2\}.&
\end{array}
\end{equation}
Denote $\hat{\theta}^{\mathrm{t}}(\boldsymbol{\lambda}):=\left\{\hat{\theta}_i^{\mathrm{t}}(\mathbf{Y};\boldsymbol{\lambda})\right\}_{i\in[N]}$. For any $\boldsymbol{\lambda}\in[0,1]\times\mathbb{R}$, define 
\begin{equation}\label{eq:loss:two-sample:lambda}
    \mathcal{L}_{\mathrm{t}}(\boldsymbol{\lambda}) :=  L_2^{\mathrm{t}}(\hat{\theta}^{\mathrm{t}}(\boldsymbol{\lambda});\theta)=\sum_{i=1}^N\mathbb{E}\left[\left\{\boldsymbol{\lambda}^T\mathbf{K}_i+(\bar{Y}_1-\bar{Y}_2)\right\}^2\right]-2(\theta_{i1}-\theta_{i2})\mathbb{E}\left[\left\{\boldsymbol{\lambda}^T\mathbf{K}_i+(\bar{Y}_1-\bar{Y}_2)\right\}\right],
\end{equation}
where $L_2^{\mathrm{t}}(\cdot;\cdot)$ is the $L_2$ risk defined in \eqref{eq:unweighted:L2:risk}. We now provide the regret bound for both one-sample and two-sample settings:
\begin{proof}[Proof of Theorem \ref{thm:regret}]
Firstly, note that 
\begin{equation}\label{eq:L:one-sample:expand}
\begin{array}{rl}
\mathcal{L}_{\mathrm{o}}(\hat{\boldsymbol{\lambda}}_{\mathrm{o}})-\mathcal{L}_{\mathrm{o}}(\boldsymbol{\lambda}_{\mathrm{o}}^*)\!\!\!\!&\displaystyle=(\hat{\boldsymbol{\lambda}}_{\mathrm{o}}-\boldsymbol{\lambda}_{\mathrm{o}}^*)^T\left\{\frac{1}{N}\sum_{i=1}^N\mathbb{E}\left[\mathbf{F}_i\mathbf{F}_i^T\right]\right\}(\hat{\boldsymbol{\lambda}}_{\mathrm{o}}-\boldsymbol{\lambda}_{\mathrm{o}}^*)\\
&\displaystyle\quad+2(\hat{\boldsymbol{\lambda}}_{\mathrm{o}}-\boldsymbol{\lambda}_{\mathrm{o}}^*)^T\left\{\frac{1}{N}\sum_{i=1}^N\mathbb{E}\left[\mathbf{F}_i(\boldsymbol{\lambda}_{\mathrm{o}}^{*T}\mathbf{F}_i+\bar{Y}-\theta_i)\right]\right\}.
\end{array}
\end{equation}
By first order condition, we have 
$$\nabla\mathcal{L}_{\mathrm{o}}(\boldsymbol{\lambda}_{\mathrm{o}}^*)=\mathbf{0},$$
which is equivalent to 
$$\frac{1}{N}\sum_{i=1}^N\mathbb{E}\left[\mathbf{F}_i(\boldsymbol{\lambda}_{\mathrm{o}}^{*T}\mathbf{F}_i+\bar{Y}-\theta_i)\right]=\mathbf{0}.$$
So according to \eqref{eq:L:one-sample:expand} we have 
\begin{equation}\label{eq:loss:regret:one-sample}
\mathcal{L}_{\mathrm{o}}(\hat{\boldsymbol{\lambda}}_{\mathrm{o}})-\mathcal{L}_{\mathrm{o}}(\boldsymbol{\lambda}_{\mathrm{o}}^*)=(\hat{\boldsymbol{\lambda}}_{\mathrm{o}}-\boldsymbol{\lambda}_{\mathrm{o}}^*)^T\left\{\frac{1}{N}\sum_{i=1}^N\mathbb{E}\left[\mathbf{F}_i\mathbf{F}_i^T\right]\right\}(\hat{\boldsymbol{\lambda}}_{\mathrm{o}}-\boldsymbol{\lambda}_{\mathrm{o}}^*).
\end{equation}
Recall from \eqref{eq:notation:one-sample} that $\mathbf{F}_i=\begin{bmatrix}
\displaystyle\frac{Y_i}{n_i}-\bar{Y}\\
\displaystyle\hat{g}_i-\hat{g}
\end{bmatrix}$, where $\displaystyle\frac{Y_i}{n_i},\bar{Y},\hat{g}_i,\hat{g}\in[0,1]$, so 
$$\frac{1}{N}\sum_{i=1}^N\mathbb{E}\left[\mathbf{F}_i\mathbf{F}_i^T\right]\preceq C_1\mathbf{I}_2$$ 
for some absolute constant $C_1$, where $\mathbf{I}_2$ is the $2$-by-$2$ identity matrix. Thus according to Theorem \ref{thm:one-sample:asymptotics} we have $$\left|\mathcal{L}_{\mathrm{o}}(\hat{\boldsymbol{\lambda}}_{\mathrm{o}})-\mathcal{L}_{\mathrm{o}}(\boldsymbol{\lambda}_{\mathrm{o}}^*)\right|=\mathrm{O}_{p}\left(\frac{1}{N}\right).$$
Similarly, expanding $\mathcal{L}_{\mathrm{t}}(\hat{\boldsymbol{\lambda}}_{\mathrm{t}})-\mathcal{L}_{\mathrm{t}}(\boldsymbol{\lambda}_{\mathrm{t}}^*)$ and using the first-order condition for $\boldsymbol{\lambda}_{\mathrm{t}}^*$, we have  
\begin{equation}\label{eq:loss:regret:two-sample}
\mathcal{L}_{\mathrm{t}}(\hat{\boldsymbol{\lambda}}_{\mathrm{t}})-\mathcal{L}_{\mathrm{t}}(\boldsymbol{\lambda}_{\mathrm{t}}^*)=(\hat{\boldsymbol{\lambda}}_{\mathrm{t}}-\boldsymbol{\lambda}_{\mathrm{t}}^*)^T\left\{\frac{1}{N}\sum_{i=1}^N\mathbb{E}\left[\mathbf{K}_i\mathbf{K}_i^T\right]\right\}(\hat{\boldsymbol{\lambda}}_{\mathrm{t}}-\boldsymbol{\lambda}_{\mathrm{t}}^*),
\end{equation}
where $\mathbf{K}_i$ is defined as in \eqref{eq:notation:two-sample}. So according to Theorem \ref{thm:CLT:two-sample} we have 
$$\left|\mathcal{L}_{\mathrm{t}}(\hat{\boldsymbol{\lambda}}_{\mathrm{t}})-\mathcal{L}_{\mathrm{t}}(\boldsymbol{\lambda}_{\mathrm{t}}^*)\right|=\mathrm{O}_{p}\left(\frac{1}{N}\right).$$
\end{proof}

\subsection{Proof for Performance Validation by Data Thinning}
\begin{proof}[Proof of Proposition~\ref{prop:data:split}]
We just need to show that for any $i\in[N]$, 
$$\mathbb{E}\left[\left(\hat{\theta}_i^{\mathrm{o}}\right)^2-2\frac{Y_i^{(1)}}{m_i}\hat{\theta}_i^{\mathrm{o}}\right]=\mathbb{E}\left[\left(\hat{\theta}_i^{\mathrm{o}}\right)^2\right]-\theta_i^{\mathrm{o}}\mathbb{E}\left[\hat{\theta}_i^{\mathrm{o}}\right].$$
$$\mathbb{E}\left[\left(\hat{\theta}_i^{\mathrm{t}}\right)^2-2\left\{\frac{Y_{i1}^{(1)}}{m_{i1}}-\frac{Y_{i2}^{(1)}}{m_{i2}}\right\}\hat{\theta}_i^{\mathrm{t}}\right]=\mathbb{E}\left[\left(\hat{\theta}_i^{\mathrm{t}}\right)^2\right]-\theta_i^{\mathrm{t}}\mathbb{E}\left[\hat{\theta}_i^{\mathrm{t}}\right].$$
Let $\mathcal{F}^{\mathrm{o}}$ and $\mathcal{F}^{\mathrm{t}}$ be the original data for one-sample and two-sample cases respectively. Then by using the property of conditional expectations, it is straightforward to see that 
$$\mathbb{E}\left[\frac{Y_i^{(1)}}{m_i}\hat{\theta}_i^{\mathrm{o}}\big|\mathcal{F}^{\mathrm{o}}\right]=\hat{\theta}_i^{\mathrm{o}}\mathbb{E}\left[\frac{Y_i^{(1)}}{m_i}\big|\mathcal{F}^{\mathrm{o}}\right]=\theta_i^{\mathrm{o}}\mathbb{E}\left[\hat{\theta}_i^{\mathrm{o}}\big|\mathcal{F}^{\mathrm{o}}\right].$$
$$\mathbb{E}\left[\left(\frac{Y_{i1}^{(1)}}{m_{i1}}-\frac{Y_{i2}^{(1)}}{m_{i2}}\right)\hat{\theta}_i^{\mathrm{t}}\big|\mathcal{F}^{\mathrm{t}}\right]=\hat{\theta}_i^{\mathrm{t}}\mathbb{E}\left[\frac{Y_{i1}^{(1)}}{m_{i1}}-\frac{Y_{i2}^{(1)}}{m_{i2}}\big|\mathcal{F}^{\mathrm{t}}\right]=\left(\theta_{i1}-\theta_{i2}\right)\mathbb{E}\left[\hat{\theta}_i^{\mathrm{t}}\big|\mathcal{F}^{\mathrm{t}}\right]=\theta_i^{\mathrm{t}}\mathbb{E}\left[\hat{\theta}_i^{\mathrm{o}}\big|\mathcal{F}^{\mathrm{t}}\right].$$
So the result follows directly. 
\end{proof}

\section{Technical Lemmas}
\begin{lemma}\label{lemma:main}
  Let $Y\sim \mathrm{Bin}(n, \theta)$. For any function $h$ on $\{0, \ldots, n\}$, we have 
    \begin{equation}
      \label{eq:T1_bias}
      \theta \E[h(Y)] - \E[\T_1 h(Y; n)] = \theta (1 - \theta)^n (\Delta h),
    \end{equation}
    and
    \begin{equation}
      \label{eq:T2_bias}
      \theta \E[h(Y)] - \E[\T_2 h(Y; n)] = (1 - \theta) \theta^n (-1)^{n+1}(\Delta h),
    \end{equation} 
    where $\mathcal{T}_1$, $\mathcal{T}_2$ and $\Delta h$ are defined as \eqref{eq:T1}, \eqref{eq:T2} and \eqref{eq:Deltah}. 
\end{lemma}
\begin{proof}[Proof of Theorem \ref{lemma:main}]
  First, we prove \eqref{eq:T1_bias}. Let
    \[g(y) = \sum_{j=0}^{n-y}h(y + j)(-1)^{j}\frac{(n-y)!}{(n-y-j)!}\frac{(y-1)!}{(y+j)!}, \quad \forall y \in \{1, \ldots, n\},\]
    and $g(0) = 0$. Then, for any $y\in \{1, \ldots, n\}$,
    \begin{align*}
      &yg(y) + (n - y)g(y+1) \\
      & = y\sum_{j=0}^{n-y}h(y + j)(-1)^{i}\frac{(n-y)!}{(n-y-j)!}\frac{(y-1)!}{(y+j)!}\\
      & \quad + (n-y)\sum_{j=0}^{n-y-1}h(y +1 + j)(-1)^{j}\frac{(n-y-1)!}{(n-y-1-j)!}\frac{y!}{(y+j+1)!}\\
      & = \sum_{i=0}^{n-y}h(y + j)(-1)^{j}\frac{(n-y)!}{(n-y-j)!}\frac{y!}{(y+j)!}\\
      & \quad + \sum_{i=1}^{n-y}h(y + j)(-1)^{j+1}\frac{(n-y)!}{(n-y-j)!}\frac{y!}{(y+j)!}\\
      & = h(y),
    \end{align*}
    where the last step follows because all summands cancel except for the first term in the first sum. When $y = 0$,
    \begin{align*}
      &yg(y) + (n - y)g(y+1) = ng(1)\\
      & = n\sum_{j=0}^{n-1}h(1 + j)(-1)^{j}\frac{(n-1)!}{(n-1-j)!}\frac{1}{(1+j)!}\\
      & = -\sum_{j=0}^{n-1}h(1 + j)(-1)^{j+1}\com{n}{j+1}\\
      & = -\sum_{j=1}^{n}h(j)(-1)^{j}\com{n}{j}\\
      & = h(0) - \Delta h.
    \end{align*}
    Putting pieces together,
    \[yg(y) + (n - y)g(y+1) = h(y) - (\Delta h) \mathbf{1}(y = 0).\]
  By Lemma \ref{prop:stein},
  \[\E[Yg(Y)] = \theta \E[Yg(Y) + (n - Y)g(Y+1)] = \theta \left\{\E[h(Y)] - (\Delta h) (1 - \theta)^n\right\}.\]
  The proof is completed by noting that the RHS is given by $\E[Yg(Y)]$.

  Next, we prove \eqref{eq:T2_bias}. Similar to \eqref{eq:T1_bias}, let $\td{Y} = n - Y$ and $\td{h}(\td{y}) = h(n - \td{y})$. The previous result then implies
  \begin{align*}
    &(1 - \theta) \theta^n (\Delta \td{h})\\
    &= (1 - \theta) \E[\td{h}(\td{Y})] - \E\left[\mathbf{1}(\td{Y} > 0)\sum_{j=0}^{n-\td{Y}}\td{h}(\td{Y} + j)(-1)^{j}\frac{(n-\td{Y})!}{(n-\td{Y}-j)!}\frac{\td{Y}!}{(\td{Y}+j)!}\right]\\
    & = (1 - \theta) \E[h(Y)] - \E\left[\mathbf{1}(Y < n)\sum_{j=0}^{Y}h(Y - j)(-1)^{j}\frac{Y!}{(Y-j)!}\frac{(n - Y)!}{(n-Y+j)!}\right].
  \end{align*}
  By definition,
  \begin{align*}
    \Delta \td{h} &= \sum_{j=0}^{n}\td{h}(j)(-1)^{j}\com{n}{j} = \sum_{j=0}^{n}h(n-j)(-1)^{j}\com{n}{j}\\
    & = \sum_{j=0}^{n}h(n)(-1)^{n-j}\com{n}{j} = (-1)^{n}\Delta h.
  \end{align*}
  The proof is then completed by rearranging the terms.
\end{proof}

\begin{lemma}\label{lemma:main:cor}
Let $Y\sim \mathrm{Bin}(n, \theta)$. Suppose $h$ is a polynomial of degree less than $n$ defined on $\{0, \ldots, n\}$, then 
\begin{equation}\label{eq:theta:E:1}
\theta \E[h(Y)] = \E\left[\mathbf{1}(Y > 0)\sum_{j=0}^{n-Y}h(Y + j)(-1)^{j}\frac{(n-Y)!}{(n-Y-j)!}\frac{Y!}{(Y+j)!}\right],
\end{equation}
    and
\begin{equation}\label{eq:theta:E:2}
\theta \E[h(Y)] = \E\left[h(Y) - \mathbf{1}(Y < n)\sum_{j=0}^{Y}h(Y - j)(-1)^{j}\frac{Y!}{(Y-j)!}\frac{(n - Y)!}{(n-Y+j)!}\right].
\end{equation}
  
\end{lemma}
\begin{proof}[Proof of Lemma \ref{lemma:main:cor}]
Following Theorem~\ref{thm:better}, $\Delta h = 0$ when $h$ is a polynomial of degree less than $n$. So \eqref{eq:theta:E:1} and \eqref{eq:theta:E:2} follow immediately from \eqref{eq:T1_bias} and \eqref{eq:T2_bias}. 
\end{proof}

\begin{lemma}\label{lemma:linear-term:rewritten}
$\mathbf{C}_{N,1}$ is equal to
$$\displaystyle\mathbf{C}_{N,1}=\begin{pmatrix}
    \displaystyle\frac{2}{N}\sum_{i=1}^N\left[\frac{\tfrac{Y_i}{n_i}\left(1-\tfrac{Y_i}{n_i}\right)}{n_i-1}-\left(\tfrac{Y_i}{n_i}-\bar{Y}\right)^2\right]\\
    \displaystyle\frac{2}{N}\sum_{i=1}^N\left[(\bar{Y}-\mathbf{1}\{Y_i>\lfloor n_i/2\rfloor\})(g_i-\bar{g})+\left(\bar{Y}-\tfrac{Y_i}{n_i}\right)(\hat{g}_i-g_i+\bar{g}-\hat{g})\right]
\end{pmatrix}.$$
\end{lemma}
\begin{proof}[Proof of Lemma \ref{lemma:linear-term:rewritten}]
Firstly, according to \eqref{eq:coefficient:linear:ML}, we have
\begin{equation}\label{eq:C:m,1:one-sample}
\begin{array}{rl}
\mathbf{C}_{N,1}\!\!\!\!\!\!&\displaystyle=\frac{1}{N}\sum_{i=1}^N2\bar{Y}(\boldsymbol{\bar{\beta}}_i+\mathbf{\Delta}_i)-2(\bar{\boldsymbol{\beta}}_i+\mathbf{\Delta}_i)\mathbf{1}\{Y_i\leq\lfloor n_i/2\rfloor\}\\
&\displaystyle\quad\quad\quad\quad-2\mathbf{1}\{Y_i>\lfloor n_i/2\rfloor\}\sum_{j=0}^{n_i-Y_i}(\boldsymbol{\bar{\beta}}_i(Y_i+j)+\mathbf{\Delta}_i)(-1)^j\frac{(n_i-Y_i)!}{(n_i-Y_i-j)!}\frac{Y_i!}{(Y_i+j)!}\\
&\displaystyle\quad\quad\quad\quad+2\mathbf{1}\{Y_i\leq\lfloor n_i/2\rfloor\}\sum_{j=0}^{Y_i}(\boldsymbol{\bar{\beta}}_i(Y_i-j)+\mathbf{\Delta}_i)(-1)^j\frac{Y_i!}{(Y_i-j)!}\frac{(n_i-Y_i)!}{(n_i-Y_i+j)!}\\
&\displaystyle=\frac{1}{N}\sum_{i=1}^N2\bar{Y}\boldsymbol{\bar{\beta}}_i-2\bar{\boldsymbol{\beta}}_i\mathbf{1}\{Y_i\leq\lfloor n_i/2\rfloor\}\\
&\displaystyle\quad\quad\quad\quad-2\mathbf{1}\{Y_i>\lfloor n_i/2\rfloor\}\sum_{j=0}^{n_i-Y_i}\boldsymbol{\bar{\beta}}_i(Y_i+j)(-1)^j\frac{(n_i-Y_i)!}{(n_i-Y_i-j)!}\frac{Y_i!}{(Y_i+j)!}\\
&\displaystyle\quad\quad\quad\quad+2\mathbf{1}\{Y_i\leq\lfloor n_i/2\rfloor\}\sum_{j=0}^{Y_i}\boldsymbol{\bar{\beta}}_i(Y_i-j)(-1)^j\frac{Y_i!}{(Y_i-j)!}\frac{(n_i-Y_i)!}{(n_i-Y_i+j)!}\\
&\quad\displaystyle+\frac{1}{N}\sum_{i=1}^N\kappa_i\mathbf{\Delta}_i,
\end{array}
\end{equation}
where $\bar{\boldsymbol{\beta}}_i, \bar{\boldsymbol{\beta}}_i(Y_i+j), \bar{\boldsymbol{\beta}}_i(Y_i-j)$ are defined as in \eqref{def:beta}, $\mathbf{\Delta}_i=(\bar{\theta}-\bar{Y},\hat{g}_i-g_i+\bar{g}-\hat{g})$, and 
\begin{equation}\label{eq:kappa:i}
\begin{array}{rl}
\kappa_i\!\!\!\!&=2\left[\bar{Y}-\mathbf{1}\{Y_i\leq\lfloor n_i/2\rfloor\}\right]\\
&\displaystyle\quad-2\mathbf{1}\{Y_i>\lfloor n_i/2\rfloor\}\sum_{j=0}^{n_i-Y_i}(-1)^j\frac{(n_i-Y_i)!}{(n_i-Y_i-j)!}\frac{Y_i!}{(Y_i+j)!}\\
&\displaystyle\quad+2\mathbf{1}\{Y_i\leq\lfloor n_i/2\rfloor\}\sum_{j=0}^{Y_i}(-1)^j\frac{Y_i!}{(Y_i-j)!}\frac{(n_i-Y_i)!}{(n_i-Y_i+j)!}.
\end{array}
\end{equation}
According to Lemma \ref{lemma:combinatorics},
$$\sum_{j=0}^{n_i-Y_i}(-1)^j\frac{(n_i-Y_i)!}{(n_i-Y_i-j)!}\frac{Y_i!}{(Y_i+j)!}=\frac{Y_i}{n_i},$$
and
$$\sum_{j=0}^{Y_i}(-1)^j\frac{Y_i!}{(Y_i-j)!}\frac{(n_i-Y_i)!}{(n_i-Y_i+j)!}=\frac{n_i-Y_i}{n_i}.$$
Hence according to \eqref{eq:kappa:i},
\begin{equation}\label{eq:kappa:i:new}
\begin{array}{rl}
    \kappa_i\!\!\!\!&\displaystyle=2\left[\bar{Y}-\mathbf{1}\{Y_i\leq\lfloor n_i/2\rfloor\}\right]-2\mathbf{1}\{Y_i>\lfloor n_i/2\rfloor\}\frac{Y_i}{n_i}+2\mathbf{1}\{Y_i\leq\lfloor n_i/2\rfloor\}\frac{n_i-Y_i}{n_i}\\
    &\displaystyle=2\bar{Y}-2\frac{Y_i}{n_i}\mathbf{1}\{Y_i\leq\lfloor n_i/2\rfloor\}-2\frac{Y_i}{n_i}\mathbf{1}\{Y_i>\lfloor n_i/2\rfloor\}=2\left(\bar{Y}-\frac{Y_i}{n_i}\right).
\end{array}
\end{equation}
Hence by \eqref{eq:C:m,1:one-sample}, we have 
\begin{equation}\label{eq:C:m,1:one-sample:expand}
\begin{array}{rl}
\mathbf{C}_{N,1}\!\!\!\!\!\!&\displaystyle=\frac{1}{N}\sum_{i=1}^N2\bar{Y}\boldsymbol{\bar{\beta}}_i-2\bar{\boldsymbol{\beta}}_i\mathbf{1}\{Y_i\leq\lfloor n_i/2\rfloor\}\\
&\displaystyle\quad\quad\quad\quad-2\mathbf{1}\{Y_i>\lfloor n_i/2\rfloor\}\sum_{j=0}^{n_i-Y_i}\boldsymbol{\bar{\beta}}_i(Y_i+j)(-1)^j\frac{(n_i-Y_i)!}{(n_i-Y_i-j)!}\frac{Y_i!}{(Y_i+j)!}\\
&\displaystyle\quad\quad\quad\quad+2\mathbf{1}\{Y_i\leq\lfloor n_i/2\rfloor\}\sum_{j=0}^{Y_i}\boldsymbol{\bar{\beta}}_i(Y_i-j)(-1)^j\frac{Y_i!}{(Y_i-j)!}\frac{(n_i-Y_i)!}{(n_i-Y_i+j)!}\\
&\quad\displaystyle+\frac{1}{N}\sum_{i=1}^N2\left(\bar{Y}-\frac{Y_i}{n_i}\right)\mathbf{\Delta}_i.
\end{array}
\end{equation}
Note that 
\begin{equation}\label{eq:beta:term1}
\begin{array}{rl}
&\displaystyle\quad\mathbf{1}\{Y_i>\lfloor n_i/2\rfloor\}\sum_{j=0}^{n_i-Y_i}\boldsymbol{\bar{\beta}}_i(Y_i+j)(-1)^j\frac{(n_i-Y_i)!}{(n_i-Y_i-j)!}\frac{Y_i!}{(Y_i+j)!}\\
&\displaystyle=\mathbf{1}\{Y_i>\lfloor n_i/2\rfloor\}\left[\sum_{j=0}^{n_i-Y_i}\left\{\boldsymbol{\bar{\beta}}_i+\begin{pmatrix}
j/n_i\\
0
\end{pmatrix}\right\}(-1)^j\frac{(n_i-Y_i)!}{(n_i-Y_i-j)!}\frac{Y_i!}{(Y_i+j)!}\right]\\
&\displaystyle=\mathbf{1}\{Y_i>\lfloor n_i/2\rfloor\}\boldsymbol{\bar{\beta}}_i\left[\sum_{j=0}^{n_i-Y_i}(-1)^j\frac{(n_i-Y_i)!}{(n_i-Y_i-j)!}\frac{Y_i!}{(Y_i+j)!}\right]\\
&\displaystyle\quad+\mathbf{1}\{Y_i>\lfloor n_i/2\rfloor\}\left[\sum_{j=0}^{n_i-Y_i}\begin{pmatrix}
j/n_i\\
0
\end{pmatrix}(-1)^j\frac{(n_i-Y_i)!}{(n_i-Y_i-j)!}\frac{Y_i!}{(Y_i+j)!}\right].
\end{array}
\end{equation}
According to Lemma \ref{lemma:combinatorics}, 
\begin{equation}\label{eq:term:1:comb:1}
\sum_{j=0}^{n_i-Y_i}(-1)^j\frac{(n_i-Y_i)!}{(n_i-Y_i-j)!}\frac{Y_i!}{(Y_i+j)!}=\frac{Y_i}{n_i}.
\end{equation}
According to Lemma \ref{lemma:combinatorics:2}, 
\begin{equation}\label{eq:term:1:comb:2}
\sum_{j=0}^{n_i-Y_i}(-1)^j\frac{j}{n_i}\frac{(n_i-Y_i)!}{(n_i-Y_i-j)!}\frac{Y_i!}{(Y_i+j)!}=-\frac{Y_i(n_i-Y_i)}{n_i^2(n_i-1)}.
\end{equation}
So taking \eqref{eq:term:1:comb:1} and \eqref{eq:term:1:comb:2} into \eqref{eq:beta:term1}, we have 
\begin{equation}\label{eq:term:1:new}
\begin{array}{rl}
&\displaystyle\quad\mathbf{1}\{Y_i>\lfloor n_i/2\rfloor\}\sum_{j=0}^{n_i-Y_i}\boldsymbol{\bar{\beta}}_i(Y_i+j)(-1)^j\frac{(n_i-Y_i)!}{(n_i-Y_i-j)!}\frac{Y_i!}{(Y_i+j)!}\\
&\displaystyle=\mathbf{1}\left\{Y_i>\left\lfloor\frac{n_i}{2}\right\rfloor\right\}\left[\frac{Y_i}{n_i}\bar{\boldsymbol{\beta}_i}+\begin{pmatrix}
\displaystyle-\frac{Y_i(n_i-Y_i)}{n_i^2(n_i-1)}\\
0
\end{pmatrix}\right]\\
&\displaystyle=\mathbf{1}\left\{Y_i>\left\lfloor\frac{n_i}{2}\right\rfloor\right\}\begin{pmatrix}
\displaystyle\frac{Y_i}{n_i}\left(\frac{Y_i}{n_i}-\bar{\theta}\right)-\frac{Y_i(n_i-Y_i)}{n_i^2(n_i-1)}\\
g_i-\bar{g}
\end{pmatrix}
\end{array}
\end{equation}
Additionally, note that 
\begin{equation}\label{eq:term:2:beta}
\begin{array}{rl}
&\quad\displaystyle\mathbf{1}\{Y_i\leq\lfloor n_i/2\rfloor\}\sum_{j=0}^{Y_i}\boldsymbol{\bar{\beta}}_i(Y_i-j)(-1)^j\frac{Y_i!}{(Y_i-j)!}\frac{(n_i-Y_i)!}{(n_i-Y_i+j)!}\\
&\displaystyle=\mathbf{1}\{Y_i\leq\lfloor n_i/2\rfloor\}\sum_{j=0}^{Y_i}\left\{\boldsymbol{\bar{\beta}}_i-\begin{pmatrix}
j/n_i\\
0
\end{pmatrix}\right\}(-1)^j\frac{Y_i!}{(Y_i-j)!}\frac{(n_i-Y_i)!}{(n_i-Y_i+j)!}\\
&\displaystyle=\mathbf{1}\{Y_i\leq\lfloor n_i/2\rfloor\}\boldsymbol{\bar{\beta}}_i\sum_{j=0}^{Y_i}(-1)^j\frac{Y_i!}{(Y_i-j)!}\frac{(n_i-Y_i)!}{(n_i-Y_i+j)!}\\
&\displaystyle\quad-\mathbf{1}\{Y_i\leq\lfloor n_i/2\rfloor\}\sum_{j=0}^{Y_i}\begin{pmatrix}
j/n_i\\
0
\end{pmatrix}(-1)^j\frac{Y_i!}{(Y_i-j)!}\frac{(n_i-Y_i)!}{(n_i-Y_i+j)!}.
\end{array}
\end{equation}
According to Lemma \ref{lemma:combinatorics}, 
\begin{equation}\label{eq:term2:comb:1}
\sum_{j=0}^{Y_i}(-1)^j\frac{Y_i!}{(Y_i-j)!}\frac{(n_i-Y_i)!}{(n_i-Y_i+j)!}=\frac{n_i-Y_i}{n_i},
\end{equation}
and according to Lemma \ref{lemma:combinatorics:2}, 
\begin{equation}\label{eq:term2:comb:2}
\sum_{j=0}^{Y_i}\begin{pmatrix}
j/n_i\\
0
\end{pmatrix}(-1)^j\frac{Y_i!}{(Y_i-j)!}\frac{(n_i-Y_i)!}{(n_i-Y_i+j)!}=-\frac{Y_i(n_i-Y_i)}{n_i^2(n_i-1)},
\end{equation}
thus taking \eqref{eq:term2:comb:1} and \eqref{eq:term2:comb:2} into \eqref{eq:term:2:beta}, we have 
\begin{equation}\label{eq:term:2:beta:2}
\begin{array}{rl}
&\quad\displaystyle\mathbf{1}\{Y_i\leq\lfloor n_i/2\rfloor\}\sum_{j=0}^{Y_i}\boldsymbol{\bar{\beta}}_i(Y_i-j)(-1)^j\frac{Y_i!}{(Y_i-j)!}\frac{(n_i-Y_i)!}{(n_i-Y_i+j)!}\\
&\displaystyle=\mathbf{1}\{Y_i\leq\lfloor n_i/2\rfloor\}\begin{pmatrix}
\displaystyle\frac{n_i-Y_i}{n_i}\left(\frac{Y_i}{n_i}-\bar{\theta}\right)+\frac{Y_i(n_i-Y_i)}{n_i^2(n_i-1)}\\
g_i-\bar{g}
\end{pmatrix}
\end{array}
\end{equation}
Hence, according to \eqref{eq:C:m,1:one-sample:expand}, 
\begin{equation}\label{eq:C:m,1:one-sample:new}
\begin{array}{rl}
\mathbf{C}_{N,1}\!\!\!\!&\displaystyle=\frac{1}{N}\sum_{i=1}^N2\bar{Y}\begin{pmatrix}
\displaystyle\frac{Y_i}{n_i}-\bar{\theta}\\
g_i-\bar{g}
\end{pmatrix}+\frac{1}{N}\sum_{i=1}^N2\mathbf{1}\{Y_i\leq\lfloor n_i/2\rfloor\}\begin{pmatrix}
\displaystyle-\frac{Y_i}{n_i}\left(\frac{Y_i}{n_i}-\bar{\theta}\right)+\frac{Y_i(n_i-Y_i)}{n_i^2(n_i-1)}\\
0
\end{pmatrix}\\
&\displaystyle\quad-\frac{1}{N}\sum_{i=1}^N2\mathbf{1}\{Y_i>\lfloor n_i/2\rfloor\}\begin{pmatrix}
\displaystyle\frac{Y_i}{n_i}\left(\frac{Y_i}{n_i}-\bar{\theta}\right)-\frac{Y_i(n_i-Y_i)}{n_i^2(n_i-1)}\\
g_i-\bar{g}
\end{pmatrix}\\
&\displaystyle\quad+\frac{1}{N}\sum_{i=1}^N2\left(\bar{Y}-\frac{Y_i}{n_i}\right)\begin{pmatrix}
    \bar{\theta}-\bar{Y}\\
    \hat{g}_i-g_i+\bar{g}-\hat{g}
\end{pmatrix}\\
&\displaystyle=\begin{pmatrix}
    \displaystyle\frac{2}{N}\sum_{i=1}^N\left[\frac{\tfrac{Y_i}{n_i}\left(1-\tfrac{Y_i}{n_i}\right)}{n_i-1}-\left(\tfrac{Y_i}{n_i}-\bar{Y}\right)^2\right]\\
    \displaystyle\frac{2}{N}\sum_{i=1}^N\left[(\bar{Y}-\mathbf{1}\{Y_i>\lfloor n_i/2\rfloor\})(g_i-\bar{g})+\left(\bar{Y}-\tfrac{Y_i}{n_i}\right)(\hat{g}_i-g_i+\bar{g}-\hat{g})\right]
\end{pmatrix}.
\end{array}
\end{equation}
\end{proof}

\begin{lemma}\label{lemma:D:N,1:rewritten}
$\displaystyle\mathbf{D}_{N,1}=\frac{2}{N}\sum_{i=1}^N\begin{pmatrix}
    W_i(\Delta Y_i)+2I_{i1}v_{i1}+(2I_{i2}-1)v_{i2}\\
    W_i\Gamma_i
\end{pmatrix}$,
where
$$\Delta Y_i=\left(\frac{Y_{i1}}{n_{i1}}-\frac{Y_{i2}}{n_{i2}}\right)-(\bar{Y}_1-\bar{Y}_2),\quad \Gamma_i=(\hat{g}_{i1}-\hat{g}_{i2})-(\hat{g}_1-\hat{g}_2),$$
$$v_{i1}=\frac{Y_{i1}(n_{i1}-Y_{i1})}{n_{i1}^2(n_{i1}-1)},\quad v_{i2}=\frac{Y_{i2}(n_{i2}-Y_{i2})}{n_{i2}^2(n_{i2}-1)},$$
$$I_{i1}:=\mathbf{1}\{Y_{i1}>\lfloor n_{i1}/2\rfloor\},\quad I_{i2}:=\mathbf{1}\{Y_{i2}>\lfloor n_{i2}/2\rfloor\},$$
$$W_i:=(\bar{Y}_1-\bar{Y}_2)+2I_{i1}\left(1-\frac{Y_{i1}}{n_{i1}}\right)+(2I_{i2}-1)\left(\frac{Y_{i2}}{n_{i2}}-1\right).$$
\end{lemma}
\begin{proof}[Proof of Lemma \ref{lemma:D:N,1:rewritten}]
According to \eqref{eq:coefficient:linear:ML:two-sample}, 
\begin{equation}\label{eq:D:1}
2(\bar{Y}_1-\bar{Y}_2)\left\{\frac{1}{N}\sum_{i=1}^N(\boldsymbol{\beta}_{i1}-\boldsymbol{\beta}_{i2})\right\}=\frac{2(\bar{Y}_1-\bar{Y}_2)}{N}\sum_{i=1}^N\begin{pmatrix}
\displaystyle\left(\frac{Y_{i1}}{n_{i1}}-\frac{Y_{i2}}{n_{i2}}\right)-(\bar{Y}_1-\bar{Y}_2)\\
\displaystyle(\hat{g}_{i1}-\hat{g}_{i2})-(\hat{g}_1-\hat{g}_2)
\end{pmatrix},
\end{equation}
\begin{equation}\label{eq:D:2}
\begin{array}{rl}
&\displaystyle\quad\frac{1}{N}\sum_{i=1}^N2\mathbf{1}(Y_{i1} > \lfloor n_{i1}/2\rfloor)\!\!\!\sum_{j=0}^{n_{i1}-Y_{i1}}\!\!\!(\boldsymbol{\beta}_{i1}(Y_{i1}+j)-\boldsymbol{\beta}_{i2})(-1)^{j}\frac{(n_{i1}-Y_{i1})!}{(n_{i1}-Y_{i1}-j)!}\frac{Y_{i1}!}{(Y_{i1}+j)!}\\
&\displaystyle=\frac{1}{N}\sum_{i=1}^N2\mathbf{1}(Y_{i1} > \lfloor n_{i1}/2\rfloor)\!\!\!\sum_{j=0}^{n_{i1}-Y_{i1}}\!\!\!\begin{pmatrix}
\displaystyle\frac{Y_{i1}}{n_{i1}}-\frac{Y_{i2}}{n_{i2}}+\frac{j}{n_{i1}}-(\bar{Y}_1-\bar{Y}_2)\\
\displaystyle(\hat{g}_{i1}-\hat{g}_{i2})-(\hat{g}_1-\hat{g}_2)
\end{pmatrix}\\
&\quad\quad\quad\quad\quad\quad\quad\quad\quad\quad\quad\quad\quad\quad\displaystyle\times(-1)^{j}\frac{(n_{i1}-Y_{i1})!}{(n_{i1}-Y_{i1}-j)!}\frac{Y_{i1}!}{(Y_{i1}+j)!}\\
\\
&\displaystyle=_{(a)}\frac{1}{N}\sum_{i=1}^N2\mathbf{1}\left\{Y_{i1}>\left\lfloor\frac{n_{i1}}{2}\right\rfloor\right\}\begin{pmatrix}
\displaystyle\left[\left(\frac{Y_{i1}}{n_{i1}}-\frac{Y_{i2}}{n_{i2}}\right)-(\bar{Y}_1-\bar{Y}_2)\right]\frac{Y_{i1}}{n_{i1}}-\frac{Y_{i1}(n_{i1}-Y_{i1})}{n_{i1}^2(n_{i1}-1)}\\
\displaystyle[(\hat{g}_{i1}-\hat{g}_{i2})-(\hat{g}_1-\hat{g}_2)]\frac{Y_{i1}}{n_{i1}}
\end{pmatrix},
\end{array}
\end{equation}
where (a) of \eqref{eq:D:2} uses Lemma \ref{lemma:combinatorics} and Lemma \ref{lemma:combinatorics:2}. 
\begin{equation}\label{eq:D:3}
    \frac{2}{N}\sum_{i=1}^N(\boldsymbol{\beta}_{i1}-\boldsymbol{\beta}_{i2})\mathbf{1}\left\{Y_{i1}\leq \left\lfloor \frac{n_{i1}}{2}\right\rfloor\right\}=\frac{2}{N}\sum_{i=1}^N\mathbf{1}\left\{Y_{i1}\leq\left\lfloor \frac{n_{i1}}{2}\right\rfloor\right\}\begin{pmatrix}
\displaystyle\left(\frac{Y_{i1}}{n_{i1}}-\frac{Y_{i2}}{n_{i2}}\right)-(\bar{Y}_1-\bar{Y}_2)\\
\displaystyle(\hat{g}_{i1}-\hat{g}_{i2})-(\hat{g}_1-\hat{g}_2)
\end{pmatrix}.
\end{equation}
\begin{equation}\label{eq:D:4}
\begin{array}{rl}
&\quad\displaystyle\frac{2}{N}\sum_{i=1}^N\mathbf{1}(Y_{i1}\leq \lfloor n_{i1}/2\rfloor)\sum_{j=0}^{Y_{i1}}(\boldsymbol{\beta}_{i1}(Y_{i1}-j)-\boldsymbol{\beta}_{i2})(-1)^{j}\frac{Y_{i1}!}{(Y_{i1}-j)!}\frac{(n_{i1}-Y_{i1})!}{(n_{i1}-Y_{i1}+j)!}\\
&\displaystyle=\frac{2}{N}\sum_{i=1}^N\mathbf{1}(Y_{i1} > \lfloor n_{i1}/2\rfloor)\sum_{j=0}^{Y_{i1}}\begin{pmatrix}
\displaystyle\frac{Y_{i1}}{n_{i1}}-\frac{Y_{i2}}{n_{i2}}-\frac{j}{n_{i1}}-(\bar{Y}_1-\bar{Y}_2)\\
\displaystyle(\hat{g}_{i1}-\hat{g}_{i2})-(\hat{g}_1-\hat{g}_2)
\end{pmatrix}\\
&\quad\quad\quad\quad\quad\quad\quad\quad\quad\quad\quad\quad\quad\quad\displaystyle\times(-1)^{j}\frac{(n_{i1}-Y_{i1})!}{(n_{i1}-Y_{i1}+j)!}\frac{Y_{i1}!}{(Y_{i1}-j)!}\\
\\
&\displaystyle=_{(a)}\frac{2}{N}\sum_{i=1}^N\mathbf{1}\{Y_{i1} > \lfloor n_{i1}/2\rfloor\}\begin{pmatrix}
\displaystyle\left[\frac{Y_{i1}}{n_{i1}}-\frac{Y_{i2}}{n_{i2}}-(\bar{Y}_1-\bar{Y}_2)\right]\frac{n_{i1}-Y_{i1}}{n_{i1}}+\frac{Y_{i1}(n_{i1}-Y_{i1})}{n_{i1}^2(n_{i1}-1)}\\
\\
\displaystyle\left[(\hat{g}_{i1}-\hat{g}_{i2})-(\hat{g}_1-\hat{g}_2)\right]\frac{n_{i1}-Y_{i1}}{n_{i1}}
\end{pmatrix}
\end{array}
\end{equation}
where (a) of \eqref{eq:D:4} follows from Lemma \ref{lemma:combinatorics} and \ref{lemma:combinatorics:2}, by setting $Y=n_{i1}-Y_{i1}$, $n=n_{i1}$ in both Lemmas.
\begin{equation}\label{eq:D:5}
\begin{array}{rl}
&\displaystyle\quad\frac{2}{N}\sum_{i=1}^N\mathbf{1}\{Y_{i2} > \lfloor n_{i2}/2\rfloor\}\!\!\!\sum_{j=0}^{n_{i2}-Y_{i2}}\!\!\!(\boldsymbol{\beta}_{i1}-\boldsymbol{\beta}_{i2}(Y_{i2}+j))(-1)^{j}\frac{(n_{i2}-Y_{i2})!}{(n_{i2}-Y_{i2}-j)!}\frac{Y_{i2}!}{(Y_{i2}+j)!}\\
&\displaystyle=\frac{2}{N}\sum_{i=1}^N\mathbf{1}\{Y_{i2} > \lfloor n_{i2}/2\rfloor\}\!\!\!\sum_{j=0}^{n_{i2}-Y_{i2}}\!\!\!\begin{pmatrix}
\displaystyle\left[\frac{Y_{i1}}{n_{i1}}-\frac{Y_{i2}}{n_{i2}}-(\bar{Y}_1-\bar{Y}_2)\right]-\frac{j}{n_{i2}}\\
(\hat{g}_{i1}-\hat{g}_{i2})-(\hat{g}_1-\hat{g}_2)
\end{pmatrix}\\
&\quad\quad\quad\quad\quad\quad\quad\quad\quad\quad\quad\quad\quad\quad\displaystyle\times(-1)^{j}\frac{(n_{i2}-Y_{i2})!}{(n_{i2}-Y_{i2}-j)!}\frac{Y_{i2}!}{(Y_{i2}+j)!}\\
\\
&\displaystyle=_{(a)}\frac{2}{N}\sum_{i=1}^N\mathbf{1}\{Y_{i2} > \lfloor n_{i2}/2\rfloor\}\begin{pmatrix}
\displaystyle\left[\frac{Y_{i1}}{n_{i1}}-\frac{Y_{i2}}{n_{i2}}-(\bar{Y}_1-\bar{Y}_2)\right]\frac{Y_{i2}}{n_{i2}}+\frac{Y_{i2}(n_{i2}-Y_{i2})}{n_{i2}^2(n_{i2}-1)}\\
\\
\displaystyle[(\hat{g}_{i1}-\hat{g}_{i2})-(\hat{g}_1-\hat{g}_2)]\frac{Y_{i2}}{n_{i2}}
\end{pmatrix},
\end{array}
\end{equation}
where (a) of \eqref{eq:D:5} follows from Lemma \ref{lemma:combinatorics} and Lemma \ref{lemma:combinatorics:2}. 
\begin{equation}\label{eq:D:6}
\frac{2}{N}\sum_{i=1}^N(\boldsymbol{\beta}_{i1}-\boldsymbol{\beta}_{i2})\mathbf{1}\left\{Y_{i2}\leq\left\lfloor \frac{n_{i2}}{2}\right\rfloor\right\}=\frac{2}{N}\sum_{i=1}^N\mathbf{1}\left\{Y_{i2}\leq\left\lfloor\frac{n_{i2}}{2}\right\rfloor\right\}\begin{pmatrix}
\displaystyle\left(\frac{Y_{i1}}{n_{i1}}-\frac{Y_{i2}}{n_{i2}}\right)-(\bar{Y}_1-\bar{Y}_2)\\
\displaystyle(\hat{g}_{i1}-\hat{g}_{i2})-(\hat{g}_1-\hat{g}_2)
\end{pmatrix}.
\end{equation}
\begin{equation}\label{eq:D:7}
\begin{array}{rl}
&\quad\displaystyle\frac{2}{N}\sum_{i=1}^N\mathbf{1}(Y_{i2}\leq \lfloor n_{i2}/2\rfloor)\sum_{j=0}^{Y_{i2}}(\boldsymbol{\beta}_{i1}-\boldsymbol{\beta}_{i2}(Y_{i2}-j))(-1)^{j}\frac{Y_{i2}!}{(Y_{i2}-j)!}\frac{(n_{i2} - Y_{i2})!}{(n_{i2}-Y_{i2}+j)!}\\
&\displaystyle=\frac{2}{N}\sum_{i=1}^N\mathbf{1}\{Y_{i2}\leq \lfloor n_{i2}/2\rfloor\}\sum_{j=0}^{Y_{i2}}\begin{pmatrix}
\displaystyle\left[\frac{Y_{i1}}{n_{i1}}-\frac{Y_{i2}}{n_{i2}}-(\bar{Y}_1-\bar{Y}_2)\right]+\frac{j}{n_{i2}}\\
(\hat{g}_{i1}-\hat{g}_{i2})-(\hat{g}_1-\hat{g}_2)
\end{pmatrix}\\
&\quad\quad\quad\quad\quad\quad\quad\quad\quad\quad\quad\quad\quad\quad\displaystyle\times(-1)^{j}\frac{(n_{i2}-Y_{i2})!}{(n_{i2}-Y_{i2}+j)!}\frac{Y_{i2}!}{(Y_{i2}-j)!}\\
\\
&\displaystyle=_{(a)}\frac{2}{N}\sum_{i=1}^N\mathbf{1}\{Y_{i2}\leq \lfloor n_{i2}/2\rfloor\}\begin{pmatrix}
\displaystyle\left[\frac{Y_{i1}}{n_{i1}}-\frac{Y_{i2}}{n_{i2}}-(\bar{Y}_1-\bar{Y}_2)\right]\frac{n_{i2}-Y_{i2}}{n_{i2}}-\frac{Y_{i2}(n_{i2}-Y_{i2})}{n_{i2}^2(n_{i2}-1)}\\
\\
\displaystyle[(\hat{g}_{i1}-\hat{g}_{i2})-(\hat{g}_1-\hat{g}_2)]\frac{n_{i2}-Y_{i2}}{n_{i2}}
\end{pmatrix},
\end{array}
\end{equation}
where (a) of \eqref{eq:D:7} follows from Lemma \ref{lemma:combinatorics} and Lemma \ref{lemma:combinatorics:2} by setting $Y=n_{i2}-Y_{i2}$ and $n=n_{i2}$. 
Thus \eqref{eq:D:1} -- \eqref{eq:D:7} imply that 
\begin{equation}\label{eq:D:N,1:expand}
\begin{array}{rl}
\mathbf{D}_{N,1}\!\!\!\!&\displaystyle=\frac{2(\bar{Y}_1-\bar{Y}_2)}{N}\sum_{i=1}^N\begin{pmatrix}
\displaystyle\left(\frac{Y_{i1}}{n_{i1}}-\frac{Y_{i2}}{n_{i2}}\right)-(\bar{Y}_1-\bar{Y}_2)\\
\displaystyle(\hat{g}_{i1}-\hat{g}_{i2})-(\hat{g}_1-\hat{g}_2)
\end{pmatrix}\\
\\
&\quad\displaystyle-\frac{2}{N}\sum_{i=1}^N\mathbf{1}\left\{Y_{i1}>\left\lfloor\frac{n_{i1}}{2}\right\rfloor\right\}\begin{pmatrix}
\displaystyle\left[\left(\frac{Y_{i1}}{n_{i1}}-\frac{Y_{i2}}{n_{i2}}\right)-(\bar{Y}_1-\bar{Y}_2)\right]\frac{Y_{i1}}{n_{i1}}-\frac{Y_{i1}(n_{i1}-Y_{i1})}{n_{i1}^2(n_{i1}-1)}\\
\\
\displaystyle[(\hat{g}_{i1}-\hat{g}_{i2})-(\hat{g}_1-\hat{g}_2)]\frac{Y_{i1}}{n_{i1}}
\end{pmatrix}\\
\\
&\quad\displaystyle-\frac{2}{N}\sum_{i=1}^N\mathbf{1}\left\{Y_{i1}\leq\left\lfloor \frac{n_{i1}}{2}\right\rfloor\right\}\begin{pmatrix}
\displaystyle\left(\frac{Y_{i1}}{n_{i1}}-\frac{Y_{i2}}{n_{i2}}\right)-(\bar{Y}_1-\bar{Y}_2)\\
\displaystyle(\hat{g}_{i1}-\hat{g}_{i2})-(\hat{g}_1-\hat{g}_2)
\end{pmatrix}\\
\\
&\quad\displaystyle+\frac{2}{N}\sum_{i=1}^N\mathbf{1}\{Y_{i1} > \lfloor n_{i1}/2\rfloor\}\begin{pmatrix}
\displaystyle\left[\frac{Y_{i1}}{n_{i1}}-\frac{Y_{i2}}{n_{i2}}-(\bar{Y}_1-\bar{Y}_2)\right]\frac{n_{i1}-Y_{i1}}{n_{i1}}+\frac{Y_{i1}(n_{i1}-Y_{i1})}{n_{i1}^2(n_{i1}-1)}\\
\\
\displaystyle\left[(\hat{g}_{i1}-\hat{g}_{i2})-(\hat{g}_1-\hat{g}_2)\right]\frac{n_{i1}-Y_{i1}}{n_{i1}}
\end{pmatrix}\\
\\
&\displaystyle\quad+\frac{2}{N}\sum_{i=1}^N\mathbf{1}\{Y_{i2} > \lfloor n_{i2}/2\rfloor\}\begin{pmatrix}
\displaystyle\left[\frac{Y_{i1}}{n_{i1}}-\frac{Y_{i2}}{n_{i2}}-(\bar{Y}_1-\bar{Y}_2)\right]\frac{Y_{i2}}{n_{i2}}+\frac{Y_{i2}(n_{i2}-Y_{i2})}{n_{i2}^2(n_{i2}-1)}\\
\\
\displaystyle[(\hat{g}_{i1}-\hat{g}_{i2})-(\hat{g}_1-\hat{g}_2)]\frac{Y_{i2}}{n_{i2}}
\end{pmatrix}\\
\\
&\displaystyle\quad+\frac{2}{N}\sum_{i=1}^N\mathbf{1}\left\{Y_{i2}\leq\left\lfloor\frac{n_{i2}}{2}\right\rfloor\right\}\begin{pmatrix}
\displaystyle\left(\frac{Y_{i1}}{n_{i1}}-\frac{Y_{i2}}{n_{i2}}\right)-(\bar{Y}_1-\bar{Y}_2)\\
\displaystyle(\hat{g}_{i1}-\hat{g}_{i2})-(\hat{g}_1-\hat{g}_2)
\end{pmatrix}\\
\\
&\displaystyle\quad+\frac{2}{N}\sum_{i=1}^N\mathbf{1}\{Y_{i2}\leq \lfloor n_{i2}/2\rfloor\}\begin{pmatrix}
\displaystyle\left[\frac{Y_{i1}}{n_{i1}}-\frac{Y_{i2}}{n_{i2}}-(\bar{Y}_1-\bar{Y}_2)\right]\frac{n_{i2}-Y_{i2}}{n_{i2}}-\frac{Y_{i2}(n_{i2}-Y_{i2})}{n_{i2}^2(n_{i2}-1)}\\
\\
\displaystyle[(\hat{g}_{i1}-\hat{g}_{i2})-(\hat{g}_1-\hat{g}_2)]\frac{n_{i2}-Y_{i2}}{n_{i2}}
\end{pmatrix}.
\end{array}
\end{equation}
Denote 
$$\Delta Y_i=\left(\frac{Y_{i1}}{n_{i1}}-\frac{Y_{i2}}{n_{i2}}\right)-(\bar{Y}_1-\bar{Y}_2),\quad \Gamma_i=(\hat{g}_{i1}-\hat{g}_{i2})-(\hat{g}_1-\hat{g}_2),$$
$$v_{i1}=\frac{Y_{i1}(n_{i1}-Y_{i1})}{n_{i1}^2(n_{i1}-1)},\quad v_{i2}=\frac{Y_{i2}(n_{i2}-Y_{i2})}{n_{i2}^2(n_{i2}-1)},$$
$$I_{i1}:=\mathbf{1}\{Y_{i1}>\lfloor n_{i1}/2\rfloor\},\quad I_{i2}:=\mathbf{1}\{Y_{i2}>\lfloor n_{i2}/2\rfloor\},$$
$$W_i:=(\bar{Y}_1-\bar{Y}_2)+2I_{i1}\left(1-\frac{Y_{i1}}{n_{i1}}\right)+(2I_{i2}-1)\left(\frac{Y_{i2}}{n_{i2}}-1\right).$$

\eqref{eq:D:N,1:expand} implies that
$$\mathbf{D}_{N,1}=\frac{2}{N}\sum_{i=1}^N\begin{pmatrix}
    W_i(\Delta Y_i)+2I_{i1}v_{i1}+(2I_{i2}-1)v_{i2}\\
    W_i\Gamma_i
\end{pmatrix}.$$

\end{proof}

\begin{lemma}\label{lemma:combinatorics}
For any $n,Y$ such that $0\leq Y\leq n$ and $n\geq2$, we have  
$$\sum_{j=0}^{n-Y}(-1)^j\frac{(n-Y)!}{(n-Y-j)!}\frac{Y!}{(Y+j)!}=\frac{Y}{n}.$$
\end{lemma}
\begin{proof}
For any $n,Y$ such that $n\geq Y\geq0$, denote
$$S(n,Y):=\sum_{j=0}^{n-Y}(-1)^j\frac{(n-Y)!}{(n-Y-j)!}\frac{Y!}{(Y+j)!}=\sum_{j=0}^{n-Y}(-1)^j\frac{{n-Y\choose j}}{{Y+j\choose j}}.$$
we now show that $S(n,Y)=Y/n$. Suppose $Y=n$, then $S(n,n)=1=Y/n$. Suppose $S(n,k)=k/n$ for some $1\leq k\leq n$, then we want to show that $$S(n,k-1)=(k-1)/n.$$ 
Note that 
\begin{equation}\label{eq:S(n,n-Y)}
\begin{array}{rl}
S(n,k-1)\!\!\!\!\!\!&\displaystyle=\sum_{j=0}^{n-k+1}(-1)^j\frac{(n-k+1)!}{(n-k+1-j)!}\frac{(k-1)!}{(k-1+j)!}=\sum_{j=-1}^{n-k}(-1)^{j+1}\frac{(n-k+1)!}{(n-k-j)!}\frac{(k-1)!}{(k+j)!}\\
&\displaystyle=1-\frac{n-k+1}{k}\sum_{j=0}^{n-k}(-1)^j\frac{(n-k)!}{(n-k-j)!}\frac{k!}{(k+j)!}=1-\frac{n-k+1}{k}\frac{k}{n}=\frac{k-1}{n}.
\end{array}
\end{equation}
Thus by induction, we have shown that 
\begin{equation}\label{eq:S(n,Y)}
    S(n,Y)=\sum_{j=0}^{n-Y}(-1)^j\frac{{n-Y\choose j}}{{Y+j\choose j}}=\frac{Y}{n}\ \mbox{for any}\ 0\leq Y\leq n.
\end{equation}
\end{proof}

\begin{lemma}\label{lemma:combinatorics:2}
For any $n\geq2$ and $0\leq Y\leq n$, we have 
$$\sum_{j=0}^{n-Y}\frac{j}{n}(-1)^j\frac{(n-Y)!}{(n-Y-j)!}\frac{Y!}{(Y+j)!}=-\frac{Y(n-Y)}{n^2(n-1)}.$$
\end{lemma}
\begin{proof}[Proof of Lemma \ref{lemma:combinatorics:2}]
Denote 
$$S_{n,Y}:=\sum_{j=0}^{n-Y}\frac{j}{n}(-1)^j\frac{(n-Y)!}{(n-Y-j)!}\frac{Y!}{(Y+j)!},\quad \forall 0\leq Y\leq n,\ n\geq2,$$
and set 
$$B_{n,Y}:=nS_{n,Y}=\sum_{j=0}^{n-Y}j(-1)^j\frac{(n-Y)!}{(n-Y-j)!}\frac{Y!}{(Y+j)!}.$$
In the following we prove by induction for $n\geq2$ that 
$$B_{n,Y}=-\frac{Y(n-Y)}{n(n-1)}\iff S_{n,Y}=-\frac{Y(n-Y)}{n^2(n-1)}.$$
Firstly , for $n=2$, we check the cases for $Y=0,1,2$ directly and see that the formula indeed holds for $n=2$:
\begin{itemize}
    \item $Y=0$: $B_{2,0}=0 = -\tfrac{0\cdot 2}{2\cdot 1}$.
    \item $Y=1$: $B_{2,1} = 1\cdot(-1)\cdot \tfrac{1!}{0!}\tfrac{1!}{2!} = -\tfrac{1}{2} = -\tfrac{1\cdot 1}{2\cdot 1}$.
    \item $Y=2$: $B_{2,2}=0 = -\tfrac{2\cdot 0}{2\cdot 1}$.
\end{itemize}
Assume that for some $n\geq2$ it's true that $B_{n,Y}=-\tfrac{Y(n-Y)}{n(n-1)}$, then we prove the formula also holds for $n+1$. Note that 
$$\frac{(n+1-Y)!}{(n+1-Y-j)!}=j!{n+1-Y\choose j},\quad \frac{Y!}{(Y+j)!}=\frac{1}{j!{Y+j\choose j}}.$$
Hence 
$$B_{n+1,Y}=\sum_{j=0}^{n+1-Y}j(-1)^j\frac{{n+1-Y\choose j}}{{Y+j\choose j}}.$$
Using the fact that 
$$\binom{n+1-Y}{j}=\binom{n-Y}{j}+\binom{n-Y}{j-1},$$
we have 
$$B_{n+1,Y}=(I)+(II),$$
where 
$$(I)=\sum_{j=0}^{n-Y}j(-1)^j\frac{\binom{n-Y}{j}}{\binom{Y+j}{j}}=B_{n,Y}=-\frac{Y(n-Y)}{n(n-1)},$$
and
$$(II)=\sum_{j=0}^{n-Y}(j+1)(-1)^{j+1}\frac{\binom{n-Y}{j}}{\binom{Y+j+1}{j+1}}=-\frac{2}{Y+1}B_{n,Y}-\frac{1}{Y+1}C_{n,Y},$$
where 
$$C_{n,Y}=\sum_{k=0}^{n-Y}(-1)^k\frac{\binom{n-Y}{k}}{\binom{Y+k}{k}}.$$
According to \eqref{eq:S(n,Y)}, we have $C_{n,Y}=\tfrac{Y}{n}$. Hence we have 
$$B_{n+1,Y}=-\frac{Y(n+1-Y)}{(n+1)n}.$$
Thus the conclusion holds.
\end{proof}


\begin{lemma}[Lindeberg–Feller Multivariate Central Limit Theorem, \citep{billingsley2013convergence}]\label{lemma:Lindeberg-Feller:CLT:multivariate}
Let $\{r_n\}$ be a monotonically increasing sequence of integers. Let $\{X_{n,\ell}\}_{\ell\in[r_n]}$ be independent random variables in $\mathbb{R}^d$ with mean zero. If for all $\epsilon>0$, $\displaystyle\sum_{i=1}^{r_n}\mathbb{E}\left[\|X_{n,i}\|_2^2\mathbf{1}\{\|X_{n,i}\|_2>\epsilon\}\right]\rightarrow0$, and $\sum_{i=1}^n\mathrm{Cov}\{X_{n,i}\}\rightarrow\boldsymbol{\Sigma}$, then 
$$\sum_{i=1}^nX_{n,i}\rightsquigarrow\mathcal{N}\left(\mathbf{0},\boldsymbol{\Sigma}\right).$$
\end{lemma}

\begin{lemma}[Cramer-Wold Theorem, \citep{cramer1936some}]\label{lemma:cramer-wold} 
Let $\mathbf{X}_n=(X_{n1},\ldots,X_{nk})$ and $\mathbf{X}=(X_1,\ldots,X_k)$ be random vectors of dimension $k$. Then $\mathbf{X}_n\rightsquigarrow\mathbf{X}$ if and only if 
$$\sum_{i=1}^kt_iX_{ni}\rightsquigarrow\sum_{i=1}^kt_iX_i$$
for each $(t_1,\ldots,t_k)\in\mathbb{R}^k$. 
\end{lemma}

\end{document}